%% file: decision_problems_relational_patterns.tex
\RequirePackage{amsmath}
\documentclass[]{eptcs}
\providecommand{\event}{AFL 2026} 
\providecommand{\publicationstatus}{}
\renewcommand{\publicationstatus}{}

\usepackage{iftex}

\ifpdf
  \usepackage{underscore}         
  \usepackage[T1]{fontenc}        
\else
  \usepackage{breakurl}           
\fi

\title{An Analysis of Decision Problems for Relational Pattern Languages under Various Constraints}
\author{Klaus Jansen\quad
Dirk Nowotka\quad
Lis Pirotton\quad
Corinna Wambsganz \quad
Max Wiedenhöft
\institute{Department of Computer Science, Kiel University, Kiel, Germany}
\email{\{kj,dn,lpi,cwa,maw\}@informatik.uni-kiel.de}
}
\def\titlerunning{Decision Problems for Relational Pattern Languages}
\def\authorrunning{K. Jansen et al.}

\usepackage[]{lineno}

\usepackage{amsfonts, amssymb, dsfont, amsthm}
\usepackage{todonotes}
\usepackage{graphicx}
\usepackage[utf8]{inputenc}
\usepackage[T1]{fontenc}
\usepackage{color}
\usepackage{float}
\usepackage{wasysym}
\usepackage{blkarray,booktabs,bigstrut}
\usepackage{verbatim}
\usepackage{tabularx}
\usepackage{upgreek}
\usepackage[nocompress]{cite}
\usepackage{hyperref}
\usepackage{cleveref}
\usepackage{enumitem}
\usepackage{marvosym}
\usepackage{booktabs}
\usepackage{tabularx}
\usepackage{array}
\usepackage[table]{xcolor}

\DeclareMathOperator{\SubSeq}{SubSeq}

\DeclareMathOperator{\var}{var}

\DeclareMathOperator{\enc}{enc}

\def\N{\mathbb{N}}
\def\ta{\mathtt{a}}
\def\tb{\mathtt{b}}
\def\tc{\mathtt{c}}

\def\relpat{\mathtt{RelPat}_{\Sigma,R}}

\def\ValC{\mathtt{ValC}}

\newcommand{\al}{\operatorname{alph}}

\newcolumntype{C}{>{\hspace{3pt}}c<{\hspace{3pt}}}

\newtheorem{theorem}{Theorem}
\newtheorem{corollary}[theorem]{Corollary}
\newtheorem{lemma}[theorem]{Lemma}
\newtheorem{proposition}[theorem]{Proposition}
\newtheorem{definition}[theorem]{Definition}
\newtheorem*{remark}{Remark}

\begin{document}

\maketitle

\begin{abstract}
\input{01_abstract}
\vspace{0.2cm}

\end{abstract}

\section{Introduction}
\input{02_introduction}
\vspace{0.5cm}

\section{Preliminaries}
\input{03_preliminaries}

\section{The Equivalence Problem}
\input{04_equivalence}

\section{The Inclusion Problem}
\input{05_00_inclusion}

\section{The Membership Problem}
\input{06_membership}

\section{Final Discussion}
\label{sec:conclusion}
\input{10_conclusion}

\subsubsection*{Acknowledgments.}
We thank the anonymous reviewers for their valuable comments and suggestions.

\bibliographystyle{eptcsini}
\bibliography{decision_problems_relational_patterns}

\newpage

\appendix

\section{Additional Content for the Equivalence Problem}
\label{appendix:equivalence}
\input{13_appendix_equivalence}

\section{Additional Content for the Inclusion Problem}
\input{12_appendix_inclusion}

\newpage

\section{Additional Content for the Membership Problem}
\label{appendix:membership}
\input{11_appendix_membeship}
\end{document}

%% file: 01_abstract.tex
Patterns are words with terminals and variables. The language of a pattern is the set of words obtained by uniformly substituting all variables with words that contain only terminals. In their original definition, patterns only allow for multiple distinct occurrences of some variables to be related by the equality relation, represented by using the same variable multiple times. In an extended notion, called relational patterns and relational pattern languages, variables may be related by arbitrary other relations, achieved by using regular patterns and relating individual variables independently from the patterns structure separately. We extend the ongoing investigation of the main decision problems for patterns (namely, the equivalence problem, the inclusion problem, and the membership problem) to relational pattern languages under a wide range of relevant individual relations, providing a comprehensive foundation in all three research directions.

%% file: 02_introduction.tex
A \emph{pattern} (with variables) is a finite word that consists only of symbols from a finite set of \emph{(terminal) letters} $\Sigma = \{\ta_1, ..., \ta_\sigma\}$ and from an infinite but countable set of \emph{variables} $X = \{x_1, x_2, ...\}$ with $\Sigma\cap X = \emptyset$. By applying terminal preserving morphisms, called \emph{substitutions}, to a pattern, we can obtain words that consist only of terminal letters. The \emph{language} of a pattern is then just the set of all words that can be obtained through the application of arbitrary substitutions.
Originally, pattern languages as defined and introduced by Angluin~\cite{Angluin1980} only considered words that are obtained by what is now called \emph{non-erasing substitutions}, i.e., substitutions where each variable is substituted by a word of at least length one. Hence, these languages are now also called \emph{non-erasing (NE) pattern languages}. Shinohara~\cite{Shinohara1983} extended this notion by allowing for the empty word to be used to substitute a variable, introducing so-called \emph{extended/erasing (E) pattern languages}. 

For example, consider the pattern $\alpha = x_1 \ta\ta x_1 \tb x_2$. Using the non-erasing substitution $h$ that sets $h(x_1) = \tb\tb$ and $h(x_2)= \ta$, we obtain the word $h(\alpha) = \tb\tb\ta\ta\tb\tb\tb\ta$. If we also considered the E-pattern language of $\alpha$, we could set all variables to the empty word and obtain the word $\ta\ta\tb$ which we could not obtain in the NE pattern language of $\alpha$.

As patterns form a natural and compact device for the definition of formal languages and due to their practical and simple definition, they occur in numerous areas in computer science and discrete mathematics. Among others, these include for example the areas of unavoidable patterns in words~\cite{Jiang1994,lothaire1997}, word equations~\cite{lothaire1997}, algorithmic learning theory~\cite{Angluin1980,FERNAU201844,Shinohara1995}, database theory~\cite{FreydenbergerP21,SchmidSchweikardtPODS2022}, or the theory of extended regular expressions with back references\cite{FREYDENBERGER20191}. Furthermore, many practical areas such as machine learning, database systems, or bio-informatics deal with the question of identifying or describing patterns in (sets of) strings~\cite{Mousawi2024}, for example in the analysis of protein data~\cite{Arikawa1993}, in pattern matching~\cite{10.1007/978-3-642-03784-9_29}, or the design of algorithms for program synthesis~\cite{10.1145/4472.4476}.

Three main decision problems emerge when it comes to the investigation of patterns and pattern languages. These are the general \emph{equivalence problem}, the \emph{inclusion problem}, and the \emph{membership problem} and its variations~\cite{gawrychowski_et_al:LIPIcs.MFCS.2021.48,Manea2022, Fleischmann2023}, all of which are considered in the non-erasing (NE) as well as erasing (E) cases. The equivalence problem determines whether the languages of two given patterns are equal. In Angluin's original work~\cite{Angluin1980}, it has been shown to be trivially decidable for all NE pattern languages. For erasing pattern languages, however, the general decidability of this problem has been the \emph{most significant open problem} in the field of decision problems for patterns~\cite{Jiang1995,Reidenbach2004-1,Reidenbach2004-2,Ohlebusch1996,Reidenbach2007}. The inclusion problem checks whether the language of a pattern is included as a subset in another pattern's language. It has been shown to be undecidable for unbounded alphabets by Jiang et al.~\cite{Jiang1995}. Their result has been extended by Freydenberger and Reidenbach in~\cite{Freydenberger2010} as well as by Bremer and Freydenberger in~\cite{BREMER201215}, where it has been shown to be undecidable for all alphabets $\Sigma$ of fixed size $|\Sigma|\geq 2$, for erasing and non-erasing pattern languages. The membership problem determines whether a word belongs to a pattern's language. It has been shown to be generally NP-complete in the non-erasing as well as erasing cases~\cite{Angluin1980,Jiang1994} and has been investigated extensively in various communities~\cite{AMIR2007514, Ordyniak2015, FERNAU2015287, Fernau2016, IBARRA1995179, Manea2019, Manea2020}.
In the special case of terminal-free patterns (containing only variables), the membership problem remains hard in both erasing and non-erasing cases, as does the non-erasing equivalence problem. However, inclusion for terminal-free E-pattern languages is NP-complete, as it reduces to checking for a morphism between patterns~\cite{Jiang1995,DBLP:journals/ipl/EhrenfeuchtR79a,File1988}, resulting also in decidability of equivalence. In contrast, Saarela~\cite{saarela:LIPIcs.ICALP.2020.140} showed that inclusion for terminal-free NE pattern languages is undecidable over a binary alphabet.

Various extensions to the notion of pattern languages have been introduced over time either to get closer to an answer for the remaining open problems or to obtain additional expressibility that is usable in some practical context. Some examples include the bounded scope coincidence degree, patterns with bounded treewidth, $k$-local patterns, or strongly-nested patterns (see \cite{Day2018} and references therein). Koshiba~\cite{Koshiba1995} introduced so-called \emph{typed patterns} where the substitutions of variables can be restricted by arbitrary recursive languages, so-called \emph{types}. This has been recently extended by Geilke and Zilles in~\cite{Geilke2011} to the notion of \emph{relational patterns} and \emph{relational pattern languages}. Here, relations other than equality may be used to restrict the valid substitution of patterns. 

This can be achieved by using regular patterns (each variable occurs only once) together with separate relational constraints on variables. The resulting language consists of all words obtained by substitutions satisfying these constraints, thereby extending classical pattern languages to arbitrary relations between variables. Consider, for example, the pattern $\beta = x_1\tc\tc x_2$. Assume $x_1$ and $x_2$ should be substituted by the reversal of each other, i.e., $(h(x_1),h(x_2))\in r_{rev}$, where $h$ is a substitution and $r_{rev}$ is the reversal relation defined by $r_{rev} := \{\ (w,w^R) \mid w\in\Sigma^*\ \}$ (where $w^R$ is the reversal of a word $w$). Then, a substitution $h$ with $h(x_1) = \ta\tb$ and $h(x_2) = \tb\ta$ is called \emph{valid} regarding the pattern $\beta$, we get $h(\beta) = \ta\tb\tc\tc\tb\ta$, and we say that $h(\beta)$ is in the relational pattern language of $\beta$ under the relational constraints that relate $x_1$ and $x_2$ by $r_{rev}$. If we considered a substitution $h'$, however, in which $h'(x_1) = \ta\tb$ and also $h'(x_2)=\ta\tb$, then clearly $h'(x_1) \neq h'(x_2)^R$ and therefore $h'$ is not a valid substitution in this case, resulting in $h'(\beta)$ not being in that relational pattern language of $\beta$.

Geilke and Zilles~\cite{Geilke2011} established general properties of relational pattern languages. For equal length, Holte et al.~\cite{pmlr-v167-holte22a} proved decidability of NE-equivalence for alphabets of size at least three, and Mousawi and Zilles~\cite{Mousawi2024,mousawi2026positivecharacteristicsetsrelational} obtained the corresponding E-result. An earlier extension to the binary-case is currently claimed to be flawed in their arXiv paper~\cite{mousawi2026positivecharacteristicsetsrelational}. However, they still claim a classification for a large subclass of patterns in the binary setting, keeping the general case open, though. For reversal, positive characteristic sets exist in the NE case, but no effective construction is known, while no such family exists in the binary E case. In addition to that, in the cases of two special kinds of relations/constraints, namely regular constraints \cite{Nowotka2024} or length constraints\cite{Nowotka2025}, it was recently shown, among other results, that the still open equivalence problem for erasing pattern languages actually becomes undecidable if these constraints are allowed in addition to equality between variables. This sets a first upper bound of constraints necessary to obtain undecidability for that problem. In comparison to other relations, however, these constraints appear to be very strong, leaving the question of what happens in the cases of weaker relations that do not restrict the substitution space with explicit values. Except the previously stated results, this area seems to be largely unexplored.

This paper considers the following relations: \emph{equality} ($r_{=}$), \emph{equal length} ($r_{|w|}$), \emph{subsequence} ($r_{ssq}$), \emph{abelian equivalence} ($r_{ab}$), \emph{alphabet permutation} ($r_{perm}$), \emph{reversal} ($r_{rev}$), (nonempty) \emph{commutation} ($r_{com^+}$, $r_{com^*}$), the \emph{star relation} ($r_{*}$), and the common (non)emptiness relation ($r_{\varepsilon_=}$). These relations are not meant to form an exhaustive list. Rather, they constitute a selection of natural and representative binary relations on words, covering several standard notions of similarity. Many other relations could be studied within this framework (e.g., see discussion in Section~\ref{sec:conclusion}). We leave such extensions for future work. 
We study the equivalence, inclusion, and membership problems under these relations, providing a comprehensive overview of many cases (Table~\ref{tab:complexity-full}).
Section 2 introduces the required notation and models. Sections 3 to 5 address the equivalence, inclusion, and membership problems, respectively. Section 6 discusses open problems and directions for future research. Due to space constraints, some proofs or extended sketches were moved into the appendix, marked with (\Rightscissors) and a reference to the respective subsection.

\begin{table}[ht]
    \centering
    \begingroup
    \newcommand{\resultref}[1]{
        \textsuperscript{\normalfont\scriptsize #1}
    }

    \setlength\tabcolsep{4pt}
    \renewcommand{\arraystretch}{1.2}
    \resizebox{\textwidth}{!}{
    \begin{tabular}{ | C | C | C | C | C | C | C | C | C | C | C | }
        \hline
         & $r_{=}$ 
         & $r_{|w|}$
         & $r_{ssq}$ 
         & $r_{ab}$ 
         & $r_{perm}$ 
         & $r_{rev}$ 
         & $r_{com^*}$
         & $r_{com^+}$
         & $r_{*}$
         & $r_{\varepsilon_=}$\\
        \hline

        E ($\in$)
        & NPC\resultref{a}
        & \textcolor{red}{NPC}\resultref{T15}
        & \textcolor{red}{NPC}\resultref{T15}
        & \textcolor{red}{NPC}\resultref{T15}
        & \textcolor{red}{NPC}\resultref{T15}
        & \textcolor{red}{NPC}\resultref{T15}
        & \textcolor{red}{NPC}\resultref{T15}
        & \textcolor{red}{NPC}\resultref{T15}
        & \textcolor{red}{NPC}\resultref{T15} 
        & \textcolor{red}{NPC}\resultref{T15} \\
        \hline

        E ($\subseteq$)
        & UD\resultref{b}
        & \textcolor{lightgray}{\textbf{?}}
        & \textcolor{red}{UD}\resultref{C8}
        & \textcolor{red}{UD}\resultref{T12}
        & \textcolor{lightgray}{\textbf{?}}
        & \textcolor{red}{UD}\resultref{T12}
        & \textcolor{lightgray}{\textbf{?}}
        & \textcolor{lightgray}{\textbf{?}}
        & \textcolor{red}{UD}\resultref{C8}
        & \textcolor{red}{D}\resultref{T10}\\
        \hline

        E ($=$)
        & \textcolor{lightgray}{\textbf{?}}
        & $\mathrm{D}_{|\Sigma|\geq 3}$\resultref{d}
        & \textcolor{red}{$\geq r_=$}\resultref{P7}
        & \textcolor{lightgray}{\textbf{?}}
        & \textcolor{lightgray}{\textbf{?}}
        & \textcolor{red}{$\geq r_=$}\resultref{P7}
        & \textcolor{lightgray}{\textbf{?}}
        & \textcolor{lightgray}{\textbf{?}}
        & \textcolor{red}{$\geq r_=$}\resultref{P7}
        & \textcolor{red}{D}\resultref{T5}\\
        \hline

        NE ($\in$)
        & NPC\resultref{e}
        & \textcolor{red}{NPC}\resultref{T15}
        & \textcolor{red}{NPC}\resultref{T15}
        & \textcolor{red}{NPC}\resultref{T15}
        & \textcolor{red}{NPC}\resultref{T15}
        & \textcolor{red}{NPC}\resultref{T15}
        & \textcolor{red}{NPC}\resultref{T15}
        & \textcolor{red}{NPC}\resultref{T15}
        & \textcolor{red}{NPC}\resultref{T15} 
        & \textcolor{red}{P}\resultref{C16}\\
        \hline

        NE ($\subseteq$)
        & UD\resultref{b}
        & \textcolor{lightgray}{\textbf{?}}
        & \textcolor{red}{UD}\resultref{C8}
        & \textcolor{red}{UD}\resultref{T9}
        & \textcolor{lightgray}{\textbf{?}}
        & \textcolor{red}{UD}\resultref{T9}
        & \textcolor{lightgray}{\textbf{?}}
        & \textcolor{lightgray}{\textbf{?}}
        & \textcolor{red}{UD}\resultref{C8}
        & \textcolor{red}{P}\resultref{S4.1}\\
        \hline

        NE ($=$)
        & P\resultref{e}
        & $\mathrm{D}_{|\Sigma|\geq 3}$\resultref{c}
        & \textcolor{red}{$?^{\dagger}$}\resultref{L1}
        & \textcolor{red}{P}\resultref{C4}
        & \textcolor{lightgray}{\textbf{?}}
        & \textcolor{red}{$?^{\dagger}$}\resultref{L1}
        & \multicolumn{2}{c|}{
            \textcolor{red}{P}\resultref{C4}
        }
        & \textcolor{red}{$?^{\dagger}$}\resultref{L1}
        & \textcolor{red}{P}\resultref{S3.1}\\
        \hline
    \end{tabular}
    }

    \caption{
        Current state of decision problems for relational pattern languages
        under the considered relations. Red entries are new results.
        NPC = NP-complete; P = polynomial-time decidable;
        D/UD = decidable/undecidable; E/NE = erasing/non-erasing;
        $\geq r_=$ denotes a polynomial-time many-one reduction from
        the equality case; $\dagger$ marks non-trivial necessary
        conditions without a decidability classification.%
    }
    \label{tab:complexity-full}

    \vspace{2mm}

    \begin{minipage}{.98\textwidth}
        \scriptsize
        \raggedright
        \textit{Pointers.}
        Superscripts T, C, P, L, and S refer to theorem, corollary,
        proposition, lemma, and section numbers in this paper, respectively.\\
        \resultref{a} Jiang et al.~\cite{Jiang1994};
        \resultref{b} Jiang et al.~\cite{Jiang1995},
        Freydenberger and Reidenbach~\cite{Freydenberger2010},
        and Bremer and Freydenberger~\cite{BREMER201215};\\
        \resultref{c} Holte et al.~
        \cite[Theorem~7 and Corollary~16]{pmlr-v167-holte22a};
        \resultref{d} Mousawi and Zilles~
        \cite[Theorem~5]{mousawi2026positivecharacteristicsetsrelational};
        \resultref{e} Angluin~\cite{Angluin1980}.
    \end{minipage}
    \endgroup
\end{table}

\vspace{-0.8cm} 

%% file: 03_preliminaries.tex
Let $\N$ denote the natural numbers $\{1, 2, 3, \dots\}$ and let $\N_0 := \N \cup \{0\}$.
For $n,m \in \mathbb{N}$ set $[m,n] := \{k \in \mathbb{N} \mid m \leq k \leq n\}$. 
Denote $[n] := [1,n]$ and $[n]_0 := [0,n]$. 
An \emph{alphabet} $\Sigma$ is a non-empty finite set whose elements are called \emph{letters}.  
A \emph{word} is a finite sequence of letters from $\Sigma$. 
Let $\Sigma^*$ be the set of all finite words over $\Sigma$ and set $\Sigma^+ := \Sigma^* \setminus \{\varepsilon\}$.
We call the number of letters in a word $w \in \Sigma^*$ \emph{length} of $w$, denoted by $|w|$.
Therefore, we have $|\varepsilon| = 0$. 
Denote $\Sigma^k := \{w \in \Sigma^* \mid |w| = k\}$.
For $w \in \Sigma^*$, let $w[i]$ denote $w$'s $i^{th}$ letter for all $i \in [\vert w \vert]$. 
For compactness reasons, we denote $w[i] \cdots w[j]$ by $w[i \cdots j]$ for all $i,j \in [\vert w \vert]$ with $i < j$.
Denote by $|w|_\ta = |\{\ i\in [|w|] \mid w[i] = \ta\}|$ the number of times the letter $\ta$ occurs in $w$.
Set $\al(w) :=  \{\ta \in \Sigma \mid \exists i \in [\vert w \vert] : w[i] = \ta\}$ as $w$’s alphabet.
A \emph{subsequence} $u$ of a word $w$ is a word obtained by deleting arbitrary letters of $w$. Denote by $\SubSeq(w)$ the set of subsequences of $w$. Two words $u,v\in\Sigma^*$ are abelian equivalent, denoted by $u\equiv_{ab}v$, iff $|u|_\ta = |v|_\ta$ for all $\ta\in\Sigma$.

Let $X := \{x_1, x_2, ...\}$ be a countable set of variables such that $\Sigma \cap X = \emptyset$.
A \emph{pattern} is then a non-empty, finite word over $\Sigma \cup X$.
The set of all patterns over $\Sigma \cup X$ is denoted by $Pat_\Sigma$. 
For example, $x_1 \ta x_2 \tb \ta x_2 x_3$ is a pattern over $\Sigma = \{\ta,\tb\}$ with $x_1,x_2,x_3\in X$.
For a pattern $\alpha\in Pat_\Sigma$, let $\var(\alpha) := \{\ x \in X\ |\ |p|_x \geq 1\ \}$ denote the set of variables occurring in $p$.
A \emph{substitution of $\alpha$} is a morphism $h : (\Sigma \cup X)^* \to \Sigma^*$ such that $h(\ta) = \ta$ for all $\ta \in \Sigma$ and 
$h(x) \in \Sigma^*$ for all $x \in X$. 
If we have $h(x) \neq \varepsilon$ for all $x \in \var(p)$, we call $h$ a \emph{non-erasing substitution} for $\alpha$. 
Otherwise $h$ is an \emph{erasing substitution} for $\alpha$. The set of all substitutions w.r.t.~$\Sigma$ is denoted by $H_\Sigma$.
If $\Sigma$ is clear from the context, we may write just $H$.
For a pattern $\alpha\in Pat_\Sigma$, its erasing pattern language $L_E(\alpha)$ and non-erasing pattern language $L_{NE}(\alpha)$
are defined, respectively, by
\begin{align*}
	L_{E}(\alpha) &= \{\ h(\alpha)\ |\ h\in H, h(x) \in \Sigma^* \text{ for all } x\in\var(\alpha)\}, \text{ and } \\
	L_{NE}(\alpha) &= \{\ h(\alpha)\ |\ h\in H, h(x) \in \Sigma^+ \text{ for all } x\in\var(\alpha)\}.
\end{align*}
Let $RegPat_\Sigma = \{ \alpha\in Pat_\Sigma \mid \forall x\in X: |\alpha|_{x} \leq 1\}$ denote the set of all regular patterns, i.e., the set of all patterns where each variable occurs only once. Let $R$ be some set of relations and denote for each $n\in\N$ by $R_n\subset R$ the subset of $n$-ary relations in $R$. Define $Rel_{R} = \{ (r,x_{i_1},... ,x_{i_n}) \mid r\in R_n, x_{i_1},... ,x_{i_n} \in X, i_1,...i_n\in\N, n\in\N \}$ to be the set of all combinations to relate variables in $X$ over the set of relations $R$. Then, a pair $(\alpha,r_\alpha)\in RegPat_\Sigma \times 2^{Rel_R}$ is a relational pattern over the terminal alphabet $\Sigma$ and the set of relations $R$. We denote the set of all relational patterns over $\Sigma$ and $R$ by $\relpat$. For some $(\alpha,r_\alpha)\in \relpat$ and $h\in H$, we say that $h$ is a $r_\alpha$-valid substitution if for all $(r',x_{i_1},...,x_{i_n})\in r_\alpha$ we have $(h(x_{i_1}),...,h(x_{i_n})) \in r'$, i.e., all variables are substituted by words related to each other regarding the relation $r'$. 
The set of all $r_\alpha$-valid substitutions w.r.t.~$\Sigma$ is denoted by $H_{\Sigma,r_\alpha}$.
If $\Sigma$ or $r_\alpha$ are apparent, we may also just write $H_{r_\alpha}$ or $H$. 
For any $(\alpha,r_\alpha)\in\relpat$, we denote its \emph{erasing relational pattern language} $L_E(\alpha,r_\alpha)$ and \emph{non-erasing relational pattern language} $L_{NE}(\alpha,r_\alpha)$, respectively, by
\begin{align*}
	L_E(\alpha,r_\alpha) &= \{\ h(\alpha)\ |\ h\in H_{\Sigma,r_\alpha}\ \} \\
	L_{NE}(\alpha,r_\alpha) &= \{\ h(\alpha)\ |\ h\in H_{\Sigma,r_\alpha} \text{ and } h(x)\in\Sigma^+\text{ for all } x\in\var(\alpha)\ \}.
\end{align*}
In this paper, we primarily investigate the considered decision problems under binary relations, i.e., $|R| = 1$ and, if $R = \{r\}$, then $r$ is 2-ary/binary. In particular, we are interested in the selection of specific customary binary relations given in Table~\ref{tab:relations}. Denote by $\mathcal{R}$ the set of all these relations.

\begin{table}[h!]
    \centering
    \small
    \setlength{\tabcolsep}{6pt}
    \renewcommand{\arraystretch}{1.08}
    \begin{tabularx}{\textwidth}{@{}>{\raggedleft\arraybackslash}p{2.2cm}X@{}}
        \toprule
        \textbf{Relation} & \textbf{Description} \\
        \midrule
        $r_{=}$ 
            & equality of words, i.e., $(u,v)\in r_{=}$ iff $u = v$. \\

        $r_{|w|}$ 
            & length equality of words, i.e., $(u,v)\in r_{|w|}$ iff $|u| = |v|$. \\

        $r_{ssq}$ 
            & subsequence relation, i.e., $(u,v)\in r_{ssq}$ iff $u\in\SubSeq(v)$. \\

        $r_{ab}$ 
            & abelian equivalence of words, i.e., $(u,v)\in r_{ab}$ iff $u \equiv_{ab} v$. \\

        $r_{perm}$ 
            & alphabet permutation of words, i.e., $(u,v)\in r_{perm}$ iff $u$ can be obtained from $v$ by a bijective renaming of the letters. \\

        $r_{rev}$ 
            & reversal of words, i.e., $(u,v)\in r_{rev}$ iff $u = v^R$. \\

        $r_{com^*}$ 
            & commutation between words, i.e., $(u,v)\in r_{com^*}$ iff there exists $z\in\Sigma^*$ such that $u,v\in\{z\}^*$. \\

        $r_{com^+}$ 
            & nonempty commutation between words, i.e., $(u,v)\in r_{com^+}$ iff there exists $z\in\Sigma^+$ such that $u,v\in\{z\}^+$. \\

        $r_{*}$ 
            & word in the Kleene star of the other word, i.e., $(u,v)\in r_{*}$ iff $u\in\{v\}^*$. \\

        $r_{\varepsilon_=}$ 
            & common emptiness or common nonemptiness, i.e., $(u,v)\in r_{\varepsilon_=}$ iff $u = \varepsilon = v$ or $u,v\in\Sigma^+$. \\
        \bottomrule
    \end{tabularx}
    \caption{The set of relations $\mathcal{R}$ considered in this paper.}
    \label{tab:relations}
\end{table}

%% file: 04_equivalence.tex
We first study equivalence, which in the erasing case already contains the long-standing open problem for $r_=$. In the non-erasing case, we characterize relations for which equivalence is determined by equality of the patterns and of the reflexive, symmetric, and transitive closures of the constraints, covering $r_{ab}$, $r_{com^+}$, and $r_=$. For $r_{rev}$, $r_{ssq}$, and $r_*$, we show that this characterization is insufficient.

In the erasing case, we identify $r_{\varepsilon_=}$ as another relation, besides $r_{|w|}$, with decidable equivalence, while for $r_{rev}$, $r_*$, and $r_{ssq}$, equivalence under $r_=$ polynomial-time many-one reduces to the corresponding equivalence problems.

\subsection{The Non-Erasing Case}

Given any relational pattern $(\alpha,r_\alpha)\in\relpat$, we denote by $[r_\alpha]$ the reflexive, symmetric, and transitive closure on the set of binary related variables in $r_\alpha$. For example, if $(r_=,x_1,x_2),(r_=,x_2,x_3)\in r_{\alpha}$, then we would have $(r_=,x_i,x_j)\in[r_\alpha]$, for any combination $i,j\in\{1,2,3\}$. For the remainder of this section, we assume w.l.o.g. a normal form on the patterns such that all variables are always introduced inside a pattern with increasing indices, starting from 1 (i.e. $x_1$, $x_2$, $x_3$, and so on). The first results gives us a strong condition for relations $r$ that are reflexive on all letters of the alphabet and do not relate two distinct letters.

\begin{lemma}\label{lemma:equivalence-length1-rel-property}
    Assume $|\Sigma|\geq 2$ and let $R = \{r\}$, for a relation $r$ with the following properties: For all $\ta\in\Sigma$, we have $(\ta,\ta)\in r$, and if $(\ta,\tb)\in r$, for two letters $\ta,\tb\in\Sigma$, then $\ta = \tb$. Let $(\alpha,r_\alpha),(\beta,r_\beta)\in\relpat$ be two patterns in normal form. If $L_{NE}(\alpha,r_\alpha) = L_{NE}(\beta,r_\beta)$ then $\alpha = \beta$ and $[r_\alpha] = [r_\beta]$.
\end{lemma}

\begin{proof}
    First, we see that $|\alpha| = |\beta|$ must hold. Suppose otherwise and assume w.l.o.g. that $|\alpha| < |\beta|$. Take some $r_\alpha$-valid substitution $h$ that substitutes each variable by just a single letter. Then $|h(\alpha)| = |\alpha|$. In the NE-case, there can't exist any $r_\beta$-valid substitution $h'$ with $|h'(\beta)| = |\alpha| < |\beta|$. A contradiction.

    Now suppose $\alpha \neq \beta$. As both are regular patterns in normal form, that means there exists a first position $i\in[|\alpha|]$ such that w.l.o.g. $\alpha[i] = \ta$, for some letter $\ta\in\Sigma$, and $\beta[i] = x_j$, for some variable $x_j\in X$. Now let $h$ be a $r_\beta$-valid substitution that substitutes each variable by just a single letter and, in particular, set $h(x_j) = \tb$, for some letter $\tb$ with $\ta\neq\tb$. As $|\alpha| = |\beta|$, any $r_\alpha$-valid substitution $h'$ with $h'(\alpha) = h(\beta)$ would have to match $h'(\alpha[k]) = h(\beta[k])$, for all $k\in[|\alpha|]$, and, in particular, would need to match $h'(\alpha[i]) = h'(\ta) = \ta \neq \tb = h(x_j) = h(\beta[i])$. A contradiction. That concludes $\alpha = \beta$.

    Now suppose $[r_\alpha] \neq [r_\beta]$. Then there exist two positions $i,j\in[|\alpha|]$ such that $\alpha[i] = \beta[i] = x_{i'}$, for some $i'\in[|\alpha|]$, and $\alpha[j] = \beta[j] = x_{j'}$, for some $j'\in[|\alpha|]$, and w.l.o.g. $(r,x_{i'},x_{j'})\in [r_\alpha]$ but $(r,x_{i'},x_{j'})\notin [r_\beta]$. Now choose a substitution $h$ that substitutes each variable just by a single letter such that $h(x_{i'}) = \ta$ and $h(x_{j'}) = \tb$, for two distinct letters $\ta,\tb\in\Sigma$. As $r$ forces all blocks of symmetrically and transitively connected variables to be substituted by the same letter, in the case of length 1 substitutions, we know that $h$ cannot be $r_\alpha$-valid. However, there exists such an $h$ that is $r_\beta$-valid by choosing substitutions of the other variables accordingly. Hence, there exists such an $h(\beta)\in L_{NE}(\beta,r_\beta)$, but we always have $h(\beta)\notin L_{NE}(\alpha,r_\alpha)$. A contradiction.
\end{proof}

This result already confirms a very specific structure for all relations with the above stated properties. In particular, this covers the relations $r_{=}$, $r_{ab}$, $r_{ssq}$, $r_{rev}$, $r_{com^+}/r_{com^*}$ (equivalent for non-erasing pattern languages), and $r_*$. This property, however, is not exhaustive to characterize equivalence. So, adding to the previous property, if we now consider equivalence relations in general, we obtain the following.

\begin{lemma}\label{lemma:equivalence-equivalencerel-langproperty}
    Let $R = \{r\}$ for an equivalence relation $r$. Let $(\alpha,r_\alpha)\in\relpat$ be some relational pattern in normal form. We have $L_{NE}(\alpha,r_\alpha) = L_{NE}(\alpha,[r_\alpha])$.
\end{lemma}
\begin{proof}
    ($\Rightarrow$): First, let $w\in L_{NE}(\alpha,r_\alpha)$. Then there exists some $r_\alpha$-valid substitution $h$ with $h(\alpha) = w$. Let $(r,x_i,x_j)\in r_\alpha$, for two variables $x_i,x_j\in\var(\alpha)$. Then $(h(x_i),h(x_j))\in r$ by definition. As $r$ is an equivalence relation, we also have $(h(x_j),h(x_i))\in r$, $(h(x_i),h(x_i))\in r$, and, if there exists some $(r,x_j,x_k)\in r_\alpha$ (hence $(h(x_j),h(x_k))\in r$), we also have $(h(x_i),h(x_k))\in r$. Hence, $h$ satisfies all constraints in $[r_\alpha]$ and thereby is $[r_\alpha]$-valid, resulting in $w\in L_{NE}(\alpha,[r_\alpha])$.

    ($\Leftarrow$): For the other direction, let $w\in L_{NE}(\alpha,[r_\alpha])$ and let $h$ be an $[r_\alpha]$-valid substitution such that $h(\alpha) = w$. Notice that $r_\alpha\subseteq[r_\alpha]$. Hence, all constraints in $r_\alpha$ are satisfied by $h$, resulting in $h$ being $r_\alpha$-valid. Thus, $w\in L_{NE}(\alpha,r_\alpha)$.
\end{proof}

Combining both previous results, we identify a set of relations for which equivalence of non-erasing languages can be characterized by equality of patterns and the corresponding closures.

\begin{theorem}\label{theorem:equivalence-subst1prop-equivrel-equality}
    Let $R = \{r\}$ for $r$ being an equivalence relation with the following properties: For all $\ta\in\Sigma$, we have $(\ta,\ta)\in r$, and if $(\ta,\tb)\in r$, for two letters $\ta,\tb\in\Sigma$, then $\ta = \tb$. Let $(\alpha,r_\alpha),(\beta,r_\beta)\in\relpat$ be two patterns in normal form. $L_{NE}(\alpha,r_\alpha) = L_{NE}(\beta,r_\beta)$ if and only if $\alpha = \beta$ and $[r_\alpha] = [r_\beta]$.
\end{theorem}
\begin{proof}
    We obtain the first direction immediately from \Cref{lemma:equivalence-length1-rel-property}. For the other direction, assume $\alpha = \beta$ and $[r_\alpha] = [r_\beta]$. By \Cref{lemma:equivalence-equivalencerel-langproperty} we know $L_{NE}(\alpha,r_\alpha) = L_{NE}(\alpha,[r_\alpha])$ and $L_{NE}(\beta,r_\beta) = L_{NE}(\beta,[r_\beta])$. As $\alpha = \beta$ and $[r_\alpha] = [r_\beta]$, we have $L_{NE}(\alpha,[r_\alpha]) = L_{NE}(\beta,[r_\beta])$. Thus, $L_{NE}(\alpha,r_\alpha) = L_{NE}(\beta,r_\beta)$.
\end{proof}

For the relations specifically considered in this paper we thus obtain the following.

\begin{corollary}[\Rightscissors \ref{proof:corollary:equivalence-ne-polydec-somerel}]\label{corollary:equivalence-ne-polydec-somerel}
    Let $R = \{r\}$ for some $r\in\{r_=, r_{ab}, r_{com^+}, r_{com^*}\}$. Given the relational patterns $(\alpha,r_\alpha),(\beta,r_\beta)\in\relpat$, we can check in polynomial time whether $L_{NE}(\alpha,r_\alpha) = L_{NE}(\beta,r_\beta)$.
\end{corollary}

\begin{remark}\label{remark:equivalence-nonerasing-commonempty-linear}
    Notice that the relation $r_{\varepsilon_=}$ has no effect in the non-erasing case. In particular, for any relational pattern $(\alpha,r_\alpha)\in\relpat$ with $R = \{r_{\varepsilon_=}\}$ we have $L_{NE}(\alpha,r_\alpha) = L_{NE}(\alpha)$ for the regular pattern $\alpha$. Hence, equivalence between two such patterns is decidable in linear time in this case~\cite{Angluin1980}.
\end{remark}

The other relations considered in \Cref{lemma:equivalence-length1-rel-property} that are not covered in \Cref{corollary:equivalence-ne-polydec-somerel}, namely $r_{ssq}$, $r_{rev}$, $r_{*}$, are clearly not equivalence relations. Hence, decidability is not immediately clear in these cases. Thus, they probably require a deeper investigation to answer decidability in general. We conjecture they are all decidable.

\subsection{The Erasing Case}
In this part, our primary focus lies on the relation $r_{\varepsilon_=}$. 
First, note that $r_{\varepsilon_=}$ is reflexive, symmetric, and transitive. We extend the notation of the prior section by defining equivalence classes $\mathcal{C}_\alpha = \{C_{\alpha,1},\dots, C_{\alpha,k_\alpha}\}$ of variables $C_{\alpha,i} \subseteq \var(\alpha)$, $i \in [k_\alpha]$ that are in relation with one another. 

In the following, we prove that the equivalence problem for erasing relational pattern languages is indeed decidable in this case. First, we present a general decidability procedure. After that, we identify multiple cases that even result in linear-time and polynomial-time algorithms to decide equivalence.

\begin{theorem}[\Rightscissors\ref{proof:theorem:equivalence-commonEmptiness-decidable-pspace}]\label{theorem:equivalence-commonEmptiness-decidable-pspace}
    Let $R = \{r_{\varepsilon_=}\}$ and $\Sigma$ be some alphabet. Let $(\alpha,r_\alpha),(\beta,r_\beta)\in\relpat$ be two relational patterns. Checking $L_{E}(\alpha,r_\alpha) = L_{E}(\beta,r_\beta)$ is decidable.
\end{theorem}
\begin{proof}[Proof Sketch.]
First, consider $(\alpha,r_\alpha)$ and for each subset $S \subseteq \mathcal{C}_\alpha$, replace the variables $v \in \bigcup_{C_{\alpha,i} \in S} C_{\alpha,i}$ by any word in $\Sigma^+$ and any other variable $v \notin \bigcup_{C_{\alpha,i} \in S} C_{\alpha,i}$ by $\varepsilon$. We call the resulting language $L_{E}(\alpha,r_\alpha,S)$. Using that, we obtain $L_{E}(\alpha,r_\alpha) = \bigcup_{S \subseteq \mathcal{C}_\alpha} L_{E}(\alpha,r_\alpha,S)$. 
Do the same for $\beta$. Note there exist at most $2^{k_\alpha} \leq 2^{|\alpha|}$ subsets $S$ for $\alpha$ and at most $2^{k_\beta}\leq 2^{|\beta|}$ subsets for $\beta$. For each $S \subseteq \mathcal{C}_\alpha$, we can take a sub-pattern $\alpha_S$ of $\alpha$ that only contains the variables that are not substituted by the empty word. As $\alpha_S$ is a regular pattern, it can be observed that $L_E(\alpha,r_\alpha,S) = L_{NE}(\alpha_S)$ and we can construct an NFA $A_{\alpha_S}$ with $L(A_{\alpha_S}) = L_{NE}(\alpha_S) = L_{NE}(\alpha,r_\alpha,S)$. Doing that for each such $S$, combining all these NFAs into one big NFA $A_\alpha$ for $L_E(\alpha,r_\alpha)$ and doing the same for $(\beta,r_\beta)$ in an NFA $A_{\beta}$, we can compare the languages $L_{E}(\alpha,r_\alpha)$ and $L_{E}(\beta,r_\beta)$ by comparing $L(A_\alpha) = L(A_\beta)$, concluding the result.
\end{proof}

Next, we introduce approaches with an improved running time for some pattern structures.

\begin{lemma}[\Rightscissors \ref{proof:lemma:empty-together-poly-equivalence}]\label{lemma:empty-together-poly-equivalence}
    Let $R = \{r_{\varepsilon_=}\}$ and $\Sigma$ be some alphabet. Let $(\alpha,r_\alpha),(\beta,r_\beta)\in\relpat$ be some relational patterns.
    \begin{enumerate}
        \item If $\alpha$ and $\beta$ are \emph{terminal-free}, we can check in linear time whether $L_{E}(\alpha,r_\alpha) = L_{E}(\beta,r_\beta)$.
        \item If $|\Sigma| = 1$, we can check in linear time whether $L_{E}(\alpha,r_\alpha) = L_{E}(\beta,r_\beta)$.
        \item If there exists $c \in \Sigma$ that is not a terminal in $\alpha$ and $\beta$, we can check in polynomial time whether $L_{E}(\alpha,r_\alpha) = L_{E}(\beta,r_\beta)$.
    \end{enumerate}
\end{lemma}
\begin{proof}[Proof Sketch.]
The key observation is that every non-empty equivalence class contributes at least one symbol, while otherwise variables are unrestricted. Hence, in 1. the accepted languages of $\alpha$ and $\beta$ are completely determined by the minimum size of an equivalence class (and, in 2., additionally by the number of terminals). These quantities can be computed by a single scan over the patterns.
In 3., the patterns are decomposed into terminal and variable blocks, and each equivalence class is represented by its occurrence vector across the variable blocks. By removing dominated classes, we obtain a minimal pattern in polynomial time. The two languages are equivalent if and only if these minimal patterns coincide.
\end{proof}

Finally, we consider some other relations in the context of the equivalence problem for erasing relational pattern languages. We can show that this problem must be at least as hard as in the case of $r_=$ for the relations $r_{*}$, $r_{ssq}$, and $r_{rev}$, in terms of time complexity. While the first two follow trivially by simulating $r_=$, the last one requires a more technical construction.

\begin{proposition}\label{proposition:equivalence-equalityhardrelations}
    Let $R = \{r\}$, $r\in\{r_{ssq},r_*,r_{rev}\}$. and $(\alpha,r_\alpha),(\beta,r_\beta)\in\relpat$. The  equivalence problem for erasing pattern languages polynomial-time many-one reduces to the equivalence problem for erasing relational pattern languages over $R$.
\end{proposition}
\begin{proof}
We reduce from the case $R={r_=}$. Let $(\alpha,r_\alpha)$ be a relational pattern over equality. For $r\in\{r_*,r_{ssq}\}$, replace each constraint $(r_=,x_i,x_j)\in r_\alpha$ by $(r,x_i,x_j)$ and $(r,x_j,x_i)$. This preserves valid substitutions, since equality is equivalent both to being mutual subsequences and to being mutual powers. Applying this change to both input patterns, respectively, gives the reductions for $r_*$ and $r_{ssq}$.

For $r=r_{rev}$, introduce for every $x_i\in\var(\alpha)$ a fresh variable $x_i'$. Let $\alpha^{\leftarrow}$ be the reversal of $\alpha$ with each $x_i$ replaced by $x_i'$, and set $\alpha^{\leftrightarrow}:=\alpha\alpha^{\leftarrow}$. Define $r_{\alpha^{\leftrightarrow}}$ by adding $(r_{rev},x_i,x_i')$ for all $x_i\in\var(\alpha)$ and $(r_{rev},x_i,x_j'),(r_{rev},x_j,x_i')$ for every $(r_=,x_i,x_j)\in r_\alpha$. Then
\[ L_E(\alpha^{\leftrightarrow},r_{\alpha^{\leftrightarrow}}) = \{ww^R\mid w\in L_E(\alpha,r_\alpha)\}. \]
Indeed, the constraints $(r_{rev},x_i,x_i')$ force $x_i'$ to be substituted by the reversal of $x_i$, while each constraint $(r_{rev},x_i,x_j')$ enforces $h(x_i)=h(x_j)$ using $h(x_j')=h(x_j)^R$. Conversely, every $r_\alpha$-valid substitution can be extended by setting $h(x_i'):=h(x_i)^R$. Thus, the constructed language is exactly the image of $L_E(\alpha,r_\alpha)$ under the injective map $f(w)=ww^R$. Applying the construction to both input patterns therefore preserves language equality. The construction is polynomial-time computable.
\end{proof}

This concludes our examination of the equivalence problem. For the non-erasing cases, we either obtain efficient decidability results or structural restrictions that may support future decidability proofs. For the erasing cases, however, most questions remain open, indicating that even for relatively weak structure-enforcing relations, equivalence is far from trivial.

%% file: 05_00_inclusion.tex
In this section, we continue with the inclusion problem for relational pattern languages. We recall that the inclusion problem for non-erasing as well as erasing pattern languages over $r_=$ is generally undecidable for all alphabets $\Sigma$ of size $|\Sigma| \geq 2$, as shown and refined in~\cite{Jiang1995,Freydenberger2010,BREMER201215, Bremer2025Melting}. As we will see in this section, for several of the considered relations in $\mathcal{R}$, we also obtain that this problem is undecidable in the non-erasing or erasing cases. In addition to that, we also identify decidable cases.

First, consider the relations $r_{ssq}$ and $r_*$. As we observed before, we can simulate the exact behavior of $r_=$ (e.g., see notes in the beginning of the proof of Proposition~\ref{proposition:equivalence-equalityhardrelations}). 
By the fact that the inclusion problem is undecidable for $r_=$ in both the erasing and non-erasing cases, we obtain the following.

\begin{corollary}\label{thm:eqrels}\cite{Jiang1995,Freydenberger2010,BREMER201215}
    Let $\Sigma$ be any finite alphabet with $|\Sigma|\geq 2$. Let $r\in\{r_=,r_{ssq},r_*\}$ and let $R = \{r\}$. Given $(\alpha,r_\alpha),(\beta,r_\beta)\in\relpat$, it is undecidable to answer whether $L_\mathcal{X}(\alpha,r_\alpha)\subseteq L_\mathcal{X}(\beta,r_\beta)$, for $\mathcal{X}\in\{E,NE\}$.
\end{corollary}

Now, we continue with the analysis for relations for which we cannot immediately infer the undecidability from such trivial relationships.

\subsection{The Non-Erasing Cases}
\input{05_01_inclusion_nonerasing}

\subsection{The Erasing Cases}
\label{subsec:inc-era}
\input{05_02_inclusion_erasing}

This concludes our results for the inclusion problem. For the remaining cases $r_{|w|}$, $r_{perm}$, $r_{com^*}$, and $r_{com^+}$, decidability remains open, as the reduction properties used above seem not to carry over: $r_{|w|}$ and $r_{perm}$ do not enforce equality of letters or unary words over at least binary alphabets, while $r_{com^*}$ or $r_{com^+}$ do not even enforce equal lengths. Thus, the decidability boundary for structurally weaker relations remains open.

%% file: 05_01_inclusion_nonerasing.tex
In the non-erasing case, first notice that, similar to the equivalence problem, in the case of $r_{\varepsilon_=}$, we compare two regular patterns (as $\varepsilon$ substitutions are excluded anyways). Here, inclusion is known to be decidable in polynomial time~\cite{Shinohara1982,Terada2005}. For the relations $r_{ab}$ and $r_{rev}$ we observe something interesting. Reconstructing the reduction from \cite{BREMER201215} in the setting of relational pattern languages, the inclusion problem 
for NE relational pattern languages under $r_{ab}$ or $r_{rev}$ is indeed also undecidable, even though $r_{ab}$ generally implies a significantly less restrictive behavior than $r_=$, and $r_{rev}$ implies a generally distinct structure.

\begin{theorem}[\Rightscissors \ref{appendix:proof-inclusion-nonerasing-abel-undec}]\label{prop:inclusion_non-erasing_abel-undec}
    Let $R := \{r\}$ for $r\in\{r_{ab},r_{rev}\}$ and $|\Sigma| \geq 2$ for some alphabet $\Sigma$. For relational patterns for $(\alpha,r_\alpha),(\beta,r_\beta)\in\relpat$, it is undecidable whether $L_{NE}(\alpha,r_\alpha) \subseteq L_{NE}(\beta,r_{\beta})$.
\end{theorem}
\begin{proof}[Proof Sketch.]
    We adapt the construction of Bremer and Freydenberger~\cite{BREMER201215} for the undecidability of inclusion for non-erasing pattern languages over a binary alphabet. The key observation is that, on the words used in their construction, the relations $r_{ab}$ and $r_{rev}$ can simulate equality whenever needed. If two related variables are forced to be substituted by single letters, then both relations force these letters to be equal. And, if one of two related variables is substituted by a unary word, then the other one must be substituted by the same unary word as well. Thus, all equality tests used in the original construction can be reproduced by $r_{ab}$ and $r_{rev}$ on the relevant parts of the pattern.

    We first consider a binary alphabet $\{0,\#\}$ and write $\bar{r}\in\{r_{ab},r_{rev}\}$. Let $U$ be the universal Turing machines used in~\cite{BREMER201215}. We construct relational patterns $(\alpha,r_\alpha)$ and $(\beta,r_{\beta})$ such that
    \[ L_{NE}(\alpha,r_\alpha) \subseteq L_{NE}(\beta,r_\beta) \]
    if and only if $U$ has no accepting computation. Hence, the undecidable emptiness problem for universal Turing machines reduces to the inclusion problem for non-erasing relational pattern languages under $\bar{r}$.

    The pattern $\beta$ is the relational-pattern analogue of the pattern from ~\cite{BREMER201215}. It has the form
    \[ \beta := a_1b_1\ \#^5\ a_2x_{1,1}x_{2,1}\cdots x_{\mu,1}b_2\ \#^5\ r_1\hat{\beta}_1r_2\hat{\beta}_2\cdots r_\mu\hat{\beta}_\mu r_{\mu+1} \]
    where, $a_j,b_j,r_k,x_{\ell,1}$ are variables, for $j\in[2]$, $k\in[\mu+1]$, $\ell\in[\mu]$, and for each $i\in[\mu]$,
    \[ \hat{\beta}_i := 0x_{i,2}\cdots x_{i,5}0\ \upgamma_i\ 0x_{i,6}\cdots x_{i,9}0\ \updelta_i\ 0x_{i,10}\cdots x_{i,13}0. \]
    for variables $x_{i,j}$, $j\in[2\cdots 13]$, and patterns $\upgamma_i,\updelta_i$. For now, the relevant constraints in $r_{\beta}$ are
    \[ (\bar r,a_1,a_2),\quad (\bar r,b_1,b_2),\quad (\bar r,x_{i,j},x_{i,j'}) \quad\text{for all } i\in[\mu]\text{ and }j,j'\in[13]. \]
    The variables $r_1,\ldots,r_{\mu+1}$ are unconstrained. The patterns $\upgamma_i$ and $\updelta_i$ are the components from ~\cite{BREMER201215} encoding the possible reasons why a string is not an accepting computation encoding. In particular, $\upgamma_i$ and $\updelta_i$ contain no factor $\#^4$. Moreover, for each $i\in[\mu]$, all variables occurring in $\upgamma_i\updelta_i$ are not related to any other variable occurring in some other  $\upgamma_j\updelta_j$ for $i\neq j$ and $i,j\in[\mu]$, nor to any variable $x_{k,\ell}$, for $k\in[\mu]$ and $j\in[13]$, as well as the variables $a_j$ and $b_j$.

    The pattern $\alpha$ is defined by
    \[ \alpha := 0^{\mu+1}\ \#^5\ 0^\mu\#0^\mu\ \#^5\ tv0\alpha_10v0\alpha_20vt \]
    where $v = 0\#^4 0$, and $\alpha_1$, $\alpha_2$ are relational pattern analogues to the ones used in ~\cite{BREMER201215}. In particular, $\alpha_1$ contains no factor $\#^3$, while $\alpha_2$ contains no occurrence of $\#$ at all. We have $r_\alpha = \emptyset$. Finally, $t$ is obtained from the suffix $r_1\hat{\beta}_1\cdots r_\mu\hat{\beta}_\mu r_{\mu+1}$ by replacing every variable by $0$.

    It remains to explain why this construction still works with $\bar{r}$ instead of equality. The factors $\#^5$ force any matching of a word generated by $\alpha$ against $\beta$ to match each of the three by $\#^5$ divided parts against each other. In particular, the middle block $0^\mu\# 0^\mu$ has to match
    \[ a_2x_{1,1}x_{2,1}\cdots x_{\mu,1}b_2. \]
    Since $a_2$ and $b_2$ are related to variables that match unary words in the first block, they are forced to behave as in the equality construction. Consequently, exactly one variable $x_{i,1}$ is substituted by $\#$, while all other variables $x_{k,1}$ are substituted by $0$. Because all variables $x_{i,j}$, $j\in[13]$, are mutually related, the selected component $\hat{\beta}_i$ receives the marker $\#$ in all $x_{i,j}$ positions, while all other components $\hat{\beta}_k$, $k\neq i$, receive a $0$ in all $x_{k,j}$ positions. Thus, the construction selects exactly one component $\hat{\beta}_i$ to match against $v0\alpha_10v0\alpha_20v$ in $\alpha$.

    This selected component checks one possible error regarding encodings of accepting computations of $U$. More precisely, the matching forces 
    \[ h(0\alpha_10)=h'(\upgamma_i) \quad\text{and}\quad h(0\alpha_20)=h'(\updelta_i) \]
    for some $r_\beta$-valid substitution $h'$. Hence, a word generated by $\alpha$ is covered by $\beta$ exactly when the encoded computation violates at least one of the conditions of accepting computations, encoded in some component $\hat{\beta}_i$. Hence, an accepting computation encoding satisfies none of these conditions, therefore resulting in a word in $L_{NE}(\alpha,r_\alpha)$ that is not in $L_{NE}(\beta,r_\beta)$.

    The next difference from the equality-based construction is the use of related variables inside the components $\upgamma_i$ and $\updelta_i$. These variables are used either to propagate selected single letters such as $0$ and $\#$, or to compare unary words that encode counter values. Both cases are handled by $r_{ab}$ and $r_{rev}$ analogously to $r_=$. Thus, each component $\hat{\beta}_i$ can be constructed analogously as in ~\cite{BREMER201215}. Therefore
    \[ L_{NE}(\alpha,r_\alpha) \subseteq L_{NE}(\beta,r_\beta) \]
    holds if and only if $U$ has no accepting computation. This proves undecidability over the binary alphabet. For larger alphabets, the extension from ~\cite{BREMER201215} also applies. Additional components $\hat{\beta}_i$ can be introduced to handle the occurrence of any other letter than $0$ or $\#$, making a selection of the first $\mu$ components only necessary if these letters do not occur in a substitution of $\alpha$. The full formal verification of the adapted construction, due to its extensive length, is given in Appendix \ref{appendix:proof-inclusion-nonerasing-abel-undec}.
\end{proof}

%% file: 05_02_inclusion_erasing.tex
We continue with erasing relational pattern languages. Similar to the non-erasing cases, we show that the inclusion problem is decidable under $r_{\varepsilon_=}$ and undecidable under the reversal relation $r_{rev}$ as well as under abelian equivalence $r_{ab}$. Indeed, in the case of $r_{\varepsilon_=}$, we can use the exact same approach as for the equivalence problem from Theorem~\ref{theorem:equivalence-commonEmptiness-decidable-pspace} to obtain decidability using the construction of two NFAs that represent $L_{E}(\alpha,r_\alpha)$ and $L_{E}(\beta,r_\beta)$ respectively.

\begin{theorem}
    Let $R = \{r_{\varepsilon_=}\}$. Given the relational patterns $(\alpha,r_\alpha),(\beta,r_\beta)\in\relpat$, checking $L_{E}(\alpha,r_\alpha) \subseteq L_{E}(\beta,r_\beta)$ is decidable.
\end{theorem}

Again, for some pattern structures, we introduce some approaches with improved polynomial running time.

\begin{lemma}[\Rightscissors \ref{proof:lemma:empty-together-poly-inclusion}]\label{lemma:empty-together-poly-inclusion}
    Let $R = \{r_{\varepsilon_=}\}$. Given some alphabet $\Sigma$ and relational patterns $(\alpha,r_\alpha),(\beta,r_\beta)\in\relpat$.
    \begin{enumerate}
        \item If $\alpha$ and $\beta$ are \emph{terminal-free}, we can check in linear time whether $L_{E}(\alpha,r_\alpha) \subseteq L_{E}(\beta,r_\beta)$.
        \item If $|\Sigma| = 1$, we can check in linear time whether $L_{E}(\alpha,r_\alpha) \subseteq L_{E}(\beta,r_\beta)$.
    \end{enumerate}
\end{lemma}
\begin{proof}[Proof Sketch.]
    By the same arguments as in \Cref{lemma:empty-together-poly-equivalence}, inclusion in the terminal-free case holds if and only if the minimum class size of $\alpha$ is at least that of $\beta$.
    For unary alphabets, the language is additionally determined by the number of terminals. Hence, inclusion holds if and only if $\alpha$ has at least as many terminals as $\beta$ and its minimum equivalence class size is at least that of $\beta$. In both cases, we can check the properties in linear time.
\end{proof}

For the undecidability proofs, we adapt the framework of Freydenberger and Reidenbach~\cite{Freydenberger2010}, which shows undecidability in the erasing case under $r_=$ for all alphabets $\Sigma$ with $|\Sigma|\geq 2$. The required adaptations for $r_{ab}$ and $r_{rev}$ rely on the same principle as in Theorem~\ref{prop:inclusion_non-erasing_abel-undec} (related single letters and unary words enforce equality). We therefore give only the main ideas for \Cref{thm:uirev} below. The full proof, due to its extensive length, is provided in Appendix~\ref{appendix:proof-inclusion-erasing-rev} for verification reasons.

\begin{theorem}[\Rightscissors \ref{appendix:proof-inclusion-erasing-rev}]\label{thm:uirev}
    Let $\bar{r} \in \{r_{rev},r_{ab}\}$ and let $R = \{\bar{r}\}$. For all alphabets $\Sigma$ with $|\Sigma|\geq 2$, given $(\alpha,r_{\alpha}),(\beta,r_{\beta}) \in \relpat$, it is generally undecidable whether $L_E(\alpha,r_\alpha) \subseteq L_E(\beta,r_{\beta})$.
\end{theorem} 
\begin{proof}[Proof Sketch.]
    Let $A$ be a nondeterministic 2-counter automaton (see Appendix~\ref{appendix:def-n2cmwi} for a definition - for this sketch we can work with $A$ abstractly and just acknowledge that the emptiness problem for nondeterministic 2-counter automata is undecidable). Consider, first, the following construction under the equality relation. Let
    \[\alpha_A := v\ v\ \#^6\ v\ x\ v\ y\ v\ \#^6\ v\ u\ v\]
    where $x,y$ are distinct variables, $v = 0\#^30$ and $u = 0\#^50$ and 
    \[\beta_A:=x_1x_{1'}\ \ldots x_{\mu}x_{\mu'}\#^6\hat{\beta}_1\ldots\hat{\beta}_{\mu}\#^6\ddot{\beta_1}\ldots\ddot{\beta_{\mu}}\]
    with, for all $i \in \{1, \ldots, \mu\}$, and $\ddot{\beta_i} := x_{i_4}\ \eta_i\ x_{i_5}$, where $x_i$, $x_{i'}$, $x_{i_4}$, $x_{i_5}$, $\ldots$ are distinct variables with $(r_=,x_i, x_{i'})$, $(r_=,x_i, x_{i_1})$, $(r_=,x_i, x_{i_2})$, $(r_=,x_i,x_{i_3})$, $(r_=,x_i,x_{i_4})$, $(r_=,x_i,x_{i_5}) \in r_{\beta_A}$ and all $\eta_i \in X^*$ are terminal-free patterns.
    We have
    $\eta_i := z_{i}\ \hat{z}_{i_1}\ \hat{z}_{i_2}\ \hat{z}_{i_3}\ \hat{z}_{i_4}\ \hat{z}_{i_5}\ z_{i'}$ with distinct variables $z_i,z_{i'},\hat{z}_{i_1},\hat{z}_{i_2},\hat{z}_{i_3},\hat{z}_{i_4},\hat{z}_{i_5}$ and $(r_=,z_i,z_{i'})$, $(r_=,\hat{z}_{i_1},\hat{z}_{i_2})$, $(r_=,\hat{z}_{i_1},\hat{z}_{i_3})$, $(r_=,\hat{z}_{i_1},\hat{z}_{i_4})$, $(r_=,\hat{z}_{i_1},\hat{z}_{i_5}) \in r_{\beta_A}$.
    
    The main concept of the proof of \Cref{thm:uirev} is to show that $L_E(\alpha_A,r_{\alpha_A}) \subseteq L_E(\beta_A,r_ {\beta_A})$ iff $A$ has no accepting computation. All possible properties of non accepting computations are listed in the $\hat{\beta}_1 \ldots \hat{\beta}_{\mu}$ part where each $\hat{\beta}_i, i \in \{1, \ldots, \mu \}$ encodes property of a non accepting computation together with selecting variables $x_{i_j}$ for $j \in {1,2,3}$. The part of $\beta_A$ and $\alpha_A$ before the first $\#^6$ ensures that exactly one pair $x_i,x_i', i \in \{1, \ldots, \mu\}$ is selected, resulting in that exactly one non accepting computation is selected and that exactly one $\ddot{\beta}_i, i \in \{1, \ldots, \mu\}$ is selected. With this, we get that $\eta_i, i \in \{1, \ldots, \mu\}$ matches $u$ and thus, that $z_i = 0$ and $\hat{z}_{i_1} = \#$. We use these two variables to encode the $0$ and the $\#$ in the (non accepting) computation encodings. If we can show that this selecting process also works in the setting $R = \bar{r}$ for $\bar{r} \in \{r_{rev},r_{ab}\}$, i.e., we can construct two variables with which we can encode the letters $0$ and $\#$, and we can model an exact match for two unary words to match counter value encodings of configuration of $A$ inside some $\hat{\beta}_i$, then this proof also works in this setting. For two variables $x,y$, we have if $x = 0$ (or $x = \#$) and $(x,y) \in \bar{r}$ for $\bar{r} \in \{r_{rev},r_{ab}\}$, then $y = 0$ (or $y = \#$ respectively). Thus, we can model with the relation $r_{rev}$ as well as with the relation $r_{ab}$ identical variables if we know a one letter substitution of one variable. Moreover, we get for two variables $x,y$ with $(x,y) \in \bar{r}$ for $\bar{r} \in \{r_{rev},r_{ab}\}$ that $x=y$ must hold if $x$ or $y$ gets substituted by a unary word.

    If we want $L_E(\alpha_A,r_{\alpha_A}) \subseteq L_E(\beta_A,r_{\beta_A})$, we need to match the factor $v\ v$ to $x_1x_{1'} \ldots x_{\mu}x_{\mu'}$ and $v\ u\ v$ to $\ddot{\beta}_1 \ldots \ddot{\beta}_{\mu}$ due to our delimiter $\#^6$. We start with the match of $v\ v$ to $x_1x_{1'} \ldots x_{\mu}x_{\mu'}$. We have $v\ v = 0\ \#^3 0\ 0 \#^3 0$. Since the variables are pairwise in relation, we can always assume that we start substituting with the left variable. If a variable gets substituted to something different from $\varepsilon$, then its corresponding variable also gets substituted to something different from $\varepsilon$. For $r_{rev}$, we get the following reasonable possibilities for the beginning of the match:
    \begin{table}[h!]
        \centering
            \begin{tabular}{l|l|l}
                $x_i$ &  $x_{i'}$ & Possible prefix for $v\ v = 0\#^300\#^30$? \\
                \hline
                 $0$ & $0$ & $00 \neq 0\#$ \\
                 $0\#$ & $\# 0$ & $\ 0\#^20 \neq 0\#^3$ \\
                 $0\#^2$ & $\#^20$ & $0\#^4 0 \neq 0\#^300$ \\
                 $0\#^3$ & $\#^30$ &  $\ 0\#^60 \neq 0\#^300\#^2$ \\
                 $0\#^30$ & $0\#^3 0$ & $0 \#^300\#^30 = v\ v$
            \end{tabular}
        \label{tab:posvrev}
    \end{table}
    
    \noindent For $r_{ab}$, we get the following reasonable possibilities for the beginning of the match:
    \begin{table}[h!]
        \centering
            \begin{tabular}{l|l|l}
                $x_i$ &  $x_{i'}$ & Possible prefix for $v\ v = 0\#^300\#^30$? \\
                \hline
                 $0$ & $0$ & $00 \neq 0\#$ \\
                 $0\#$ & $0\#$ or $\# 0$ & $0\#0\#,\ 0\#^20 \neq 0\#^3$ \\
                 $0\#^2$ & $0\#^2$ or $\#0\#$ or $\#\#0$ & $0\#^20\#^2,\ 0\#^30\#,\ 0\#^4 0 \neq 0\#^300$ \\
                 $0\#^3$ & $0\#^3$ or $\#0\#^2$ or $\#^20\#$ or $\#^30$ &  $0\#^30\#^3,\ 0\#^3\#0\#^2,\ 0\#^50\#,\ 0\#^60 \neq 0\#^300\#^2$ \\
                 $0\#^30$ & $0\#^3 0$ & $0 \#^300\#^30 = v\ v$
            \end{tabular}
        \label{tab:posv}
    \end{table}
    
    Thus, the only possibility is that for exactly one $i \in \{1, \ldots, \mu\}$ the variables $x_i$ and $x_{i'}$ get substituted with $v$ and all other variables $x_j,x_{j'}$ for $j \neq i, j \in \{1, \ldots, \mu\}$ get substituted with $\varepsilon$. This together with the matching from $v\ u\ v$ to $\ddot{\beta}_1 \ldots \ddot{\beta}_{\mu}$ implies that $\eta_i$ needs to match $u$. For this matching between $\eta_i = z_{i}\ \hat{z}_{i_1}\ \hat{z}_{i_2}\ \hat{z}_{i_3}\ \hat{z}_{i_4}\ \hat{z}_{i_5}\ z_{i'}$ and $u = 0\#^50$, we get the following possibilities for both relations, $r_{ab}$ and $r_{rev}$ (recall that $(\bar{r},z_i,z_{i'}),(\bar{r},\hat{z}_{i_1},\hat{z}_{i_2}),(\bar{r},\hat{z}_{i_1},\hat{z}_{i_3}),(\bar{r},\hat{z}_{i_1},\hat{z}_{i_4}),(\bar{r},\hat{z}_{i_1},\hat{z}_{i_5}) \in r_{\beta_A}$ in this setting):
    
    \begin{table}[ht]
        \centering
        \begin{tabular}{l|l|l}
           $z_i$  & $\hat{z}_{i_1}$ & Possible substitution for $u = 0\#^50$? \\ \hline
            $\varepsilon$ & $0$ or $\#$ & $0^5,\ \#^5 \neq 0\#^50$ \\
            $0$ & $\#$ & $0\#^50 = 0\#^5 0 = u$ \\
            $0\#$ & $\#$ & $0\#\#^5\#0 \neq 0 \#^50$ \\
            $0\#\#$ & $\#$ & $0\#^2\#^5\#^20 \neq 0\#^50$ \\
            $0\#\#\#$ & $\varepsilon$ & $0\#^60 \neq 0\#^50$
        \end{tabular}
        \label{tab:posu}
    \end{table}

    Thus, the only possibility is that $z_i$ and $z_{i'}$ get substituted with $0$ and $\hat{z}_{i_k}, k \in \{1,2,3,4,5\}$ get substituted with $\#$. With this, we can encode the $0$ and the $\#$ and thus, the proof of Freydenberger and Reidenbach~\cite{Freydenberger2010} works in this setting analogously.
\end{proof}

%% file: 06_membership.tex
In this final section, we briefly discuss the membership problem -- also known as the matching problem -- for relational pattern languages for all relations in $\mathcal{R}$ and general patterns.
For the case of $r_{=}$, NP-completeness of the membership problem was originally shown in \cite{Angluin1980} and \cite{Jiang1994} for the non-erasing and erasing cases respectively. In fact, independent communities have shown this under several considered restrictions (cf. \cite{FERNAU2015287,Manea2019}). As shown in this section, all but one cases remain NP-complete.

\begin{proposition}\label{prop:membership-innp-relational}
   Let $r\in\mathcal{R}$ and $R = \{r\}$. Let $w\in\Sigma^*$ and $(\alpha,r_\alpha)\in\relpat$. Deciding $w\in L_{X}(\alpha,r_\alpha)$, for $X\in\{E,NE\}$, is in NP.
\end{proposition}

Indeed, NP-containment follows by the known arguments from the cases arising in $r_=$. We can guess a substitution $h$ with $|h(\alpha)| = |w|$, check the relational constraints in polynomial time, and finally check whether $h(\alpha) = w$. Regarding NP-hardness, we revisit the exposition given in \cite{Manea2019}, where a reduction is given from the known NP-complete positive 1-in-3-SAT (no negated variables, 3-CNF clauses, satisfiable with exactly one literal per clause set to true) problem. For example, assuming $\Sigma = \{0,\#\}$, the formula $$\{\{X_1,X_2,X_3\},\{X_2,X_4,X_5\},\{X_1,X_4,X_5\},\{X_3,X_4,X_6\}\}$$ is in 3-CNF and 1-in-3 satisfiable if and only if the word $w = 0\#0\#0\#0$ is in the erasing pattern language of the pattern $\alpha = x_1x_2x_3\#x_2x_4x_5\#x_1x_4x_5\#x_3x_4x_6$. Using $w = 0^4\#0^4\#0^4\#0^4$, this construction works for the NE case as well. For most of the relations considered in this paper, this framework suffices.

\begin{lemma}[\Rightscissors \ref{proof:lemma:membership-trivial-hardness}]\label{lemma:membership-trivial-hardness}
    Let $r\in\{r_=, r_{|w|}, r_{ssq}, r_{ab}, r_{perm}, r_{rev}, r_{*},r_{com^+},r_{\varepsilon_=}\}$ s.t. $R = \{r\}$ and let $|\Sigma|\geq 2$. For $w\in\Sigma^*$, $(\alpha,r_{\alpha})\in\relpat$, deciding $w\in L_{E}(\alpha,r_\alpha)$ is NP-hard. If $r \notin \{r_{com^+},r_{\varepsilon_=}\}$, deciding $w\in L_{NE}(\alpha,r_\alpha)$ is NP-hard.
\end{lemma}

This leaves us with the cases of $r_{com^+} (r_{com^*})$ in the non-erasing case and with $r_{com^*}$ in the erasing case. For these, a slightly different reduction is necessary but possible to achieve. As the general idea is still a natural reduction from satisfiability, in this case the variant 3-CNF-SAT (or just 3-SAT), we omit the formal proof to Appendix~\ref{proof:theorem:membership-np-complete-cases}. This concludes the following.

\begin{theorem}[\Rightscissors \ref{proof:theorem:membership-np-complete-cases}]\label{theorem:membership-np-complete-cases}
    Let $r\in\mathcal{R}$ s.t. $R = \{r\}$ and assume $|\Sigma|\geq 2$. For $w\in\Sigma^*$ and $(\alpha,r_{\alpha})\in\relpat$, deciding $w\in L_{E}(\alpha,r_\alpha)$ is NP-complete. If $r \neq r_{\varepsilon_=}$, deciding $w\in L_{NE}(\alpha,r_\alpha)$ is NP-complete.
\end{theorem}

Finally, we recall from \Cref{remark:equivalence-nonerasing-commonempty-linear} that the relation $r_{\varepsilon_=}$ has no effect in the non-erasing case, resulting in each pattern language $L_{NE}(\alpha,r_\alpha)$, for $(\alpha,r_\alpha)\in\relpat$ and $R = \{r_{\varepsilon_=}\}$, being equal to the pattern language $L_{NE}(\alpha)$ of the regular pattern $\alpha$. The membership problem for regular patterns has been shown to be decidable in linear time by Shinohara in \cite{Shinohara1983}. Hence, the following immediately holds.

\begin{corollary}\label{corollary:membership-poly-case}
    Let $R = \{r_{\varepsilon_=}\}$. For $w\in\Sigma^*$ and $(\alpha,r_{\alpha})\in\relpat$, we can decide in linear time whether $w\in L_{NE}(\alpha,r_\alpha)$.
\end{corollary}

This concludes our discussion of the membership problem. Except for $r_{\varepsilon_=}$ in the non-erasing case, all considered cases are NP-complete. Thus, hardness persists even for structurally weak relations such as $r_{\varepsilon_=}$, $r_{com^*}$, and $r_{com^+}$, confirming that matching words against compact relational pattern descriptions is generally computationally hard.

%% file: 10_conclusion.tex
We considered the three main decision problems equivalence, inclusion, and membership problems for relational pattern languages under various typical binary relations, i.e., $r_=$, $r_{ssq}$, $r_{ab}$, $r_{perm}$, $r_{rev}$, $r_{com^+}$, $r_{com^*}$, $r_{*}$, and $r_{\varepsilon_=}$. For all problems, we obtained a selection of combinatorial and complexity results that can serve as a basis for several settings where encoding problems into the relational pattern languages with different relations might be a natural choice. In addition to that, these results together with the remaining open cases discussed in the respective sections can serve as a basis for future research in all the directions considered in this paper.

In addition to these results, the problems considered in this paper could also be considered for other natural relations that have not been considered in this context, yet. Among these, for example, the non-equal relation that allows for everything but equal substitutions among related variables is an interesting candidate that also finds relevance in the literature (e.g., see \cite{Manea2020}). There, it has been shown that it is NP-complete to decide whether a word $w$ can be factorized into $k$ distinct factors. By constructing a pattern $x_1...x_k$ and relating all variables with each other using non-equality, this directly reduces to the relational pattern language setting using that relation, resulting in NP-completeness of Membership. It would be interesting to know, how the inclusion and equivalence problems behave with respect to this relation. Another candidate could be a cyclic shift relation that restricts substitutions of related variables to conjugates of a word. As an equivalence relation, we expect similar results as obtained in this paper to also be obtainable in this case. 

Finally, we only considered patterns using a single relation. Allowing multiple relations could increase expressiveness substantially: decidable single-relation cases might become undecidable when relations are combined, and open cases such as erasing equivalence under equality might become undecidable as well when additional relations weaker than regular or length constraints are allowed.

%% file: 13_appendix_equivalence.tex
\subsection{Proof of Corollary~\ref{corollary:equivalence-ne-polydec-somerel}}\label{proof:corollary:equivalence-ne-polydec-somerel}
\begin{proof}
    Since every relation in ${r_=, r_{ab}, r_{con^+}}$ is an equivalence relation enforcing equality of related letters, we first normalize $(\alpha,r_\alpha)$ and $(\beta,r_\beta)$ by renaming variables in order of first occurrence to $x_1,x_2,\dots$. We then test whether $\alpha=\beta$. If not, we return false. Otherwise, we compute $[r_\alpha]$ and $[r_\beta]$ via their reflexive, symmetric, and transitive closures by a single scan of the patterns. These classes have total size at most the number of variables, so this also takes linear time. We return true exactly if $[r_\alpha]=[r_\beta]$, and false otherwise, by \Cref{theorem:equivalence-subst1prop-equivrel-equality}.
\end{proof}

\subsection{Proof of Theorem~\ref{theorem:equivalence-commonEmptiness-decidable-pspace}}
\label{proof:theorem:equivalence-commonEmptiness-decidable-pspace}
\begin{proof}
    Again, consider $(\alpha,r_\alpha)$ and for each subset $S \subseteq \mathcal{C}_\alpha$, replace the variables $v \in \bigcup_{C_{\alpha,i} \in S} C_{\alpha,i}$ by any word in $\Sigma^+$ and any other variable $v \notin \bigcup_{C_{\alpha,i} \in S} C_{\alpha,i}$ by $\varepsilon$. We call the resulting language $L_{E}(\alpha,r_\alpha,S)$. Using that, we obtain $L_{E}(\alpha,r_\alpha) = \bigcup_{S \subseteq \mathcal{C}_\alpha} L_{E}(\alpha,r_\alpha,S)$. 
    Do the same for $\beta$. Note there exist at most $2^{k_\alpha} \leq 2^{|\alpha|}$ subsets $S$ for $\alpha$ and at most $2^{k_\beta}\leq 2^{|\beta|}$ subsets for $\beta$.

    Here, we provide a more detailed construction of the previously mentioned NFAs: Starting with $(\alpha,r_\alpha)$, for each such $S$ defined above, we can take a sub-pattern $\alpha_S$ of $\alpha$ that only contains the variables that are not substituted by the empty word. As $r_{\varepsilon_=}$ imposes no structural restriction other than being substituted (non)emptily commonly, essentially $\alpha_S$ is a regular pattern and we obtain the relationship $L_{NE}(\alpha_S) = L_{NE}(\alpha,r_\alpha,S)$.  So, for each such regular sub-pattern $\alpha_S$, we can create an NFA $A_{\alpha_S} = (Q,\Sigma,\delta,q_0),F$ with $L(A_{\alpha_S}) = L_{NE}(\alpha_S)$ using a straightforward construction that represents that regular pattern. Consider the following sketch. We set $Q = \{q_{\alpha_S,0},...,q_{\alpha_S,|\alpha_S|}\}$ where each state $q_{\alpha_S,i}$ represents a position between some letters in the pattern $\alpha_S$ (or a position in front of the first letter or after the last letter). Then we set the state corresponding to the last position in $\alpha_S$ (i.e., $q_{\alpha_S,|\alpha_S|}$) to the only final state. Hence, $F = \{q_{\alpha_S,|\alpha_S|}\}$. For all other positions $i < |\alpha_S|$, if $\alpha_S[i+1]$ is some terminal letter $\ta\in\Sigma$, we add only one outgoing transition and set $\delta(q_{\alpha_S,i},\ta) = \{q_{\alpha_S,i+1}\}$. On the other hand, if $\alpha_S[i+1]\in\var(\alpha_S)$ is some variable, add two outgoing transitions for each letter, one going to the state representing the next position, the other one staying in the same state. Thus, set $\delta(q_{\alpha_S,i},\ta) = \{q_{\alpha_S,i}, q_{\alpha_S,i+1}\}$, for all $\ta\in\Sigma$. As $\alpha_S$ is a regular pattern as established before, the relationship $L(A_{\alpha_S}) = L_{NE}(\alpha_S)$ is apparent. Now, for each $S \subseteq \mathcal{C}_\alpha$, we can construct such an NFA $A_{\alpha_S}$. Finally, we can combine all of these NFA into one common NFA $A_{\alpha}$ where a new initial state is introduced that is connected to each initial state of the sub-automata $A_{\alpha_S}$ by an $\varepsilon$-transition, and each final state is preserved. So, using one outgoing transition of the initial state of $A_{\alpha}$, we choose which blocks of variables $S$ are to be substituted emptily together and which are not. Then, the respective sub-automaton $A_{\alpha_S}$ produces some word of the corresponding sub-language $L_E{\alpha,r_{\alpha_S}}$. In total, we see that the NFA $A_\alpha$ can be constructed and represents $L(A_{\alpha}) = L_E{\alpha,r_\alpha}$. So the same for $(\beta,r_\beta)$ by constructing an NFA $A_{\beta}$ analogously. It is well known that checking equivalence of the languages of two NFA is PSPACE-complete (see, e.g., \cite{Aho1974TheDA}), hence, in particular, decidable. Thus, we can check $L(A_{\alpha}) = L(A_{\beta})$ and use the result to decide $L_E(\alpha,r_\alpha) = L_E(\beta,r_\beta)$.
\end{proof}

\subsection{Proof of Lemma~\ref{lemma:empty-together-poly-equivalence}}
\label{proof:lemma:empty-together-poly-equivalence}
\begin{proof}

\textbf{(1):} 
        Let $\alpha = x_1 x_2 \dots x_{n_\alpha}$ and $\beta = y_1 y_2,\dots y_{n_\beta}$ be terminal-free and let $\mathcal{C}_\alpha$ and $\mathcal{C}_\beta$ be the equivalence classes. 
        Note that the empty word can be obtained by both patterns. Thus, in the following, we focus on non-empty words.
        
        Consider the pattern $\alpha$. Let $S \subseteq \mathcal{C}_\alpha$ with $S \neq \emptyset$ be the subset of equivalence classes chosen to be non-empty. Assume that all variables in the remaining classes are replaced by $\varepsilon$.
        Since each non-empty variable contributes at least one letter to the resulting word but has no other restriction, we can build all words of length at least $\sum_{C_{\alpha,i} \in S} |C_{\alpha,i}|$. 
        Now, when taking all possible subsets $S$ into account, no word $w$ of positive length below $\min_{C_{\alpha,i} \in \mathcal{C}_\alpha}|C_{\alpha,i}|$ is possible. Thus, the language of accepted words by $\alpha$ is
        \begin{align*}
            L_{E}(\alpha,r_\alpha) = \{\varepsilon\} \cup \{w \mid |w| \geq \min_{C_{\alpha,i} \in \mathcal{C}_\alpha}|C_{\alpha,i}|\}.
        \end{align*}
        Respectively, the language of accepted words by $\beta$ is 
        \begin{align*}
            L_{E}(\beta,r_\beta) = \{\varepsilon\} \cup \{w \mid |w| \geq \min_{C_{\beta,i} \in \mathcal{C}_\beta}|C_{\beta,i}|\}.
        \end{align*}
    
        Now, $\min_{C_{\alpha,i} \in \mathcal{C}_\alpha}|C_{\alpha,i}| = \min_{C_{\beta,i} \in \mathcal{C}_\beta}|C_{\beta,i}|$  if and only if both patterns are equivalent.
        This property can be tested in time $O(n_\alpha + n_\beta)$ by scanning once over the patterns and counting the sizes of the equivalence classes.
        
\textbf{(2):}
Assume $|\Sigma| = 1$ and $a \in \Sigma$. Note that the patterns $\alpha$ and $\beta$ must have the same number of terminals in order to be equivalent. Otherwise we could set all variables to $\varepsilon$ and obtain distinct words. Let $\ell$ be the number of terminals in each pattern.
By the same arguments as in (1), either all variables are empty, i.e., $\{a^\ell\} \in L_{E}(\alpha,r_\alpha)$ and $\{a^\ell\} \in L_{E}(\beta,r_\beta)$, or a non-empty subset $S$ of equivalence classes contributes $\sum_{C_{\alpha,i} \in S} |C_{\alpha,i}|$ (respectively $\sum_{C_{\beta,i} \in S} |C_{\beta,i}|$) $a$'s to the words. As we consider a single-letter alphabet, the language of accepted words by $\alpha$ is
        \begin{align*}
            L_{E}(\alpha,r_\alpha) = \{a^\ell\} \cup \{w \mid |w| \geq \ell + \min_{C_{\alpha,i} \in \mathcal{C}_\alpha}|C_{\alpha,i}|\}.
        \end{align*}
        Respectively, the language of accepted words by $\beta$ is 
        \begin{align*}
            L_{E}(\beta,r_\beta) = \{a^\ell\} \cup \{w \mid |w| \geq \ell + \min_{C_{\beta,i} \in \mathcal{C}_\beta}|C_{\beta,i}|\}.
        \end{align*}
        We obtain that both patterns are equivalent if and only if they have the same number of terminals and $\min_{C_{\alpha,i} \in \mathcal{C}_\alpha}|C_{\alpha,i}| = \min_{C_{\beta,i} \in \mathcal{C}_\beta}|C_{\beta,i}|$. This can again be tested in time $O(n_\alpha + n_\beta)$ by scanning once over the patterns and counting the sizes of the equivalence classes.
        
\textbf{(3):} 
        Assume there exists a letter $c \in \Sigma$ that does not appear as a terminal in either $\alpha$ or $\beta$.
        If all variables are set to $\varepsilon$, both patterns generate a word that matches their terminals. Hence, a necessary condition for equivalence is that the terminal words of $\alpha$ and $\beta$ are identical. This, we assume in the following and focus on words where at least one variable is non-empty.

        Let $m$ be the number of terminals in $\alpha$ or $\beta$. We write 
        \begin{align*}
            \alpha = V_{\alpha,0} T_{1} V_{\alpha,1} T_{2} \dots T_{m} V_{\alpha,m},    
        \end{align*}
        where $T_j, j \in \{1,\dots,m\}$ are blocks terminals and $V_{\alpha,j}, j \in \{0,\dots,m\}$ are blocks of variables. Since the terminals of both patterns are identical, $\beta$ has the same decomposition but the blocks of variables may be different, i.e., 
        \begin{align*}
            \beta =  V_{\beta,0} T_{1} V_{\beta,1} T_{2} \dots T_{m} V_{\beta,m}.
        \end{align*}

        First, consider the pattern $\alpha$. For each equivalence class $C_{\alpha,i}$, let $\mathbf{occ}_{\alpha,i}(j)$ be the number of occurrences of that class in block $V_{\alpha,j}$ and define 
        \begin{align*}
            \mathbf{occ}_{\alpha,i} := (\mathbf{occ}_{\alpha,i}(0),\mathbf{occ}_{\alpha,i}(1),\dots,\mathbf{occ}_{\alpha,i}(m)).
        \end{align*}

        If a subset $S\subseteq\mathcal{C}_\alpha$ is non-empty, the minimal number of letters in block $V_j$ is $\sum_{C_{\alpha,i}\in S} \mathbf{occ}_{\alpha,i}(j)$. Thus, every word of length at least $\sum_{C_{\alpha,i}\in S} \mathbf{occ}_{\alpha,i}(j)$ can be generated in that block. 

        We say that an equivalence class $C_{\alpha,i}$ is \emph{dominated} if there exists a set $S'$ of other classes, i.e., $S'\subset \mathcal{C}_\alpha$ and $C_{\alpha,i} \notin S'$, such that for all blocks $j \in \{0,\dots,m\}$ we have 
        \[\mathbf{occ}_{\alpha,i}(j) = 0 \Longleftrightarrow \sum_{C_{\alpha,i'} \in S'} \mathbf{occ}_{\alpha,i'}(j) = 0\] and
        \[\mathbf{occ}_{\alpha,i}(j) \geq \sum_{C_{\alpha,i'} \in S'}\mathbf{occ}_{\alpha,i'}(j).\]
        Then all words generated in each block individually by $C_{\alpha,i}$ can also be generated by the variables in $S'$ and we may remove all variables of that equivalence class without changing the language. 
        We may consider each block individually since $c$ does not occur as a terminal, we may assign the letter $c$ to all variable occurrences.
        By iteratively removing dominated classes we obtain a \emph{minimal pattern} of $\alpha$. For a removed class $i$, we set $\mathbf{occ}_{\alpha,i}(j) = 0$ for all $j \in \{0,\dots,m\}$.
        We do the same reduction step for $\beta$.
        Now, the languages are equal if and only if the minimal patterns are identical, i.e., for all $i \in \{1,\dots,k_\alpha\}$ there exists $i' \in \{1,\dots,k_\beta\}$ such that $\mathbf{occ}_{\alpha,i} = \mathbf{occ}_{\beta,i'}$.

        We now argue that this condition is correct. 
        Suppose the minimal patterns differ.  
        Then there exists a class $C$ for one pattern that has no matching class in the other pattern but is also not dominated. 
        Because $c$ is unused as a terminal, construct a word that assigns all variables of $C$ the letter $c$ and sets all other variables to $\varepsilon$. Due to the distinct number of occurrences, this gives in at least one block a word length that can not be matched by the other pattern.
        Thus, the languages differ.
        
        The block decomposition and counting the occurrences both require linear time.  
        Testing the dominance between classes takes at most $O(k_\alpha^2 m)$, and respectively $O(k_\beta^2 m)$ comparisons. Comparing the minimal patterns is also possible in linear time
        Hence, the complete procedure runs in polynomial time.
\end{proof}

%% file: 12_appendix_inclusion.tex
\subsection{Proof of Theorem~\ref{prop:inclusion_non-erasing_abel-undec}}
\label{appendix:proof-inclusion-nonerasing-abel-undec}

\subsubsection{Definition of the Considered Universal Turing Machine}
\input{definition_univTM}

\subsubsection{The Proof}
\input{inclusion_long_nonerasing_abelAndRev_undec}

\subsection{Proof of Lemma~\ref{lemma:empty-together-poly-inclusion}}
\label{proof:lemma:empty-together-poly-inclusion}
\begin{proof}
    The proof relies on the proof of \Cref{lemma:empty-together-poly-equivalence}.
    
    \textbf{(1:)} As shown in the first part of proof of \Cref{lemma:empty-together-poly-equivalence}, the language of accepted words of $\alpha$ is
        \begin{align*}
            L_{E}(\alpha,r_\alpha) = \{\varepsilon\} \cup \{w \mid |w| \geq \min_{C_{\alpha,i} \in \mathcal{C}_\alpha}|C_{\alpha,i}|\}
        \end{align*}
        and the language of accepted words of $\beta$ is 
        \begin{align*}
            L_{E}(\beta,r_\beta) = \{\varepsilon\} \cup \{w \mid |w| \geq \min_{C_{\beta,i} \in \mathcal{C}_\beta}|C_{\beta,i}|\}.
        \end{align*}
        Thus, $L_{E}(\alpha,r_\alpha) \subseteq L_{E}(\beta,r_\beta)$ if and only if $\min_{C_{\alpha,i} \in \mathcal{C}_\alpha}|C_{\alpha,i}| \geq \min_{C_{\beta,i} \in \mathcal{C}_\beta}|C_{\beta,i}|$.
        This property can be tested in time $O(n_\alpha + n_\beta)$ by scanning once over the patterns and counting the sizes of the equivalence classes.

    \textbf{(2:)} Assume $|\Sigma| = 1$ and $a \in \Sigma$. In order to be included, $\alpha$ must have at least as many terminals as $\beta$. Otherwise we could set all variables in $\alpha$ to $\varepsilon$ and obtain a word that is not in the language of $\beta$. Let $\ell_\alpha$ be the number of terminals in $\alpha$ and $\ell_\beta$ be the number of terminals in $\beta$.
        As shown in the second part of proof of \Cref{lemma:empty-together-poly-equivalence}, the language of accepted words of $\alpha$ is
        \begin{align*}
            L_{E}(\alpha,r_\alpha) = \{a^{\ell_\alpha}\} \cup \{w \mid |w| \geq \ell_\alpha + \min_{C_{\alpha,i} \in \mathcal{C}_\alpha}|C_{\alpha,i}|\}
        \end{align*}
        and the language of accepted words of $\beta$ is 
        \begin{align*}
            L_{E}(\beta,r_\beta) = \{a^{\ell_\beta}\} \cup \{w \mid |w| \geq \ell_\beta + \min_{C_{\beta,i} \in \mathcal{C}_\beta}|C_{\beta,i}|\}.
        \end{align*}
        Thus, $L_{E}(\alpha,r_\alpha) \subseteq L_{E}(\beta,r_\beta)$ if and only $\ell_\alpha \geq \ell_\beta$ and $\min_{C_{\alpha,i} \in \mathcal{C}_\alpha}|C_{\alpha,i}| \geq \min_{C_{\beta,i} \in \mathcal{C}_\beta}|C_{\beta,i}|$. This can be tested in time $O(n_\alpha + n_\beta)$ by scanning once over the patterns and counting the sizes of the equivalence classes.
\end{proof}

\subsection{Proof of Theorem~\ref{thm:uirev}}
\label{appendix:proof-inclusion-erasing-rev}

\subsubsection{Definition of Nondeterministic 2-Counter Automata Without Input}
\label{appendix:def-n2cmwi}
\input{definition_n2cmwi}

\subsubsection{Proof of the Case $r_{rev}$}
\input{inclusion_long_erasing_rev_undec}

\subsubsection{Proof of the Case $r_{ab}$}
\input{inclusion_long_erasing_abel}

%% file: definition_univTM.tex
Here, we define the universal Turing machine $U$ used in the referenced proof from \cite{BREMER201215} and referred to in the proofs for the undecidability of the inclusion problem for NE relational pattern languages under $r_{ab}$ and $_{rev}$ in Theorem~\ref{prop:inclusion_non-erasing_abel-undec}. Let $U = (Q,\Gamma,\delta)$ be the universal Turing machine $U_{15,2}$ wih 2 symbols and 15 states as described by Neary and Woods \cite{Neary2009}. This machine hat the states $Q = \{q_1,...,q_{15}\}$ and the tape alphabet $\Gamma = \{0,1\}$. The transition function $\delta:\Gamma\times Q\rightarrow(\Gamma\times\{L,R\}\times Q)\cup\mathtt{HALT}$ is given in Table~\ref{tab:UTM-delta}.

We follow with the definition of encodings of computations of $U$ as depicted in \cite{BREMER201215}.  The following conventions are needed to discuss configurations of $U$. The tape content of any configuration of $U$ is characterized by two infinite sequences $t_L = (t_{L,n})_{n\geq0}$ and $t_R = (r_{R,n})_{n\geq0}$ over $\Gamma$. The sequence $t_L$ describes the \emph{left side of the tape}, the sequence starting at the head position of $U$ (including) and extending to the left. Analogously, $t_R$ describes the \emph{right side of the tape}, the sequence starting directly after the head position and extending to the right. A \emph{configuration} $C = (q_i,t_L,t_R)$ of $U$ is a triple consisting of a state $q_i$, a left side of the tape $t_L$ and a right side of the tape $t_R$.

Let $e : \Gamma \rightarrow N$ be a function defined by $e(0) := 0$, $e(1) := 1$, and the extension to to infinite sequences $t = (t_n)_{n\geq0}$ over $\Gamma$ by $e(t) := \sum_{i=0}^{\infty}e(t_i)$. As in each configuration of $U$ only a finite number of cells consist of no blank symbol $(0)$, $e(t)$ is always finite and well-defined. Notice that we can always obtain the symbol that is closest to the head by $e(t)\mod 2$ (the symbol at the head position in the case of $t_L$ and the symbol right of the head position in the case of $t_R$). By multiplying or dividing the encoding $e(t)$ by $2$, each side can be lengthened or shortened, respectively. The \emph{encoding of configurations of $U$} indirectly referred to in this paper is defined by
$$\enc_{NE}(q_i,t_L,t_R) = 0^70^{e(t_R)}\#0^70^{e(t_L)}\#0^{i+6}$$
for every configuration $(q_i,t_L,t_R)$. Recall that $i > 0$ as $q_i\in\{q_1,...,q_{15}\}$. A \emph{computation} $\mathcal{C} = (C_1,...,C_n)$ on $U$ is a finite sequence of configurations of $U$. It is \emph{valid} if $C_1 = I$ ($I$ being some initial configuration), $C_n$ is a halting configuration, and $C_{i+1}$ is a valid successor configuration of $C_i$, for $i\in[n-1]$, as defined by $\delta$. In \cite{BREMER201215}, the notion is adopted that any possible configuration where both tape sides have a finite value under $e$ is a valid successor configuration of a halting configuration. The \emph{encoding of computations} of $U$ is given analogously to the definition in the case of nondeterministic 2-counter automata without input, i.e., for some computation $\mathcal{C} = (C_1,...,c_n)$, we have
$$\enc_{NE}(\mathcal{C}) = \#\#\enc_{NE}(C_1)\#\#\enc_{NE}(c_2)\#\#\ ...\ \#\#\enc_{NE}(C_n)\#\#.$$
Finally, also analogous to nondeterministic 2-counter automata without input, let
$$ \ValC_U(I) = \{ \enc_{NE}(\mathcal{C}) \mid\mathcal{C}\text{ is a valid computation from } I\ \}.$$

\begin{table}
    \centering
    \resizebox{\textwidth}{!}{
    \begin{tabular}{ | c | c c c c c c c c | }

        \hline

        & $q_{1}$
        & $q_{2}$
        & $q_{3}$
        & $q_{4}$
        & $q_{5}$
        & $q_{6}$
        & $q_{7}$
        & $q_{8}$ \\

        \hline

        $0$
        & $(0,R,q_2)$
        & $(1,R,q_3)$
        & $(0,L,q_7)$
        & $(0,L,q_6)$
        & $(1,R,q_1)$
        & $(1,L,q_4)$
        & $(0,L,q_8)$
        & $(1,L,q_9)$ \\

        $1$
        & $(1,R,q_1)$
        & $(1,R,q_1)$
        & $(0,L,q_5)$
        & $(1,L,q_5)$
        & $(1,L,q_4)$
        & $(1,L,q_4)$
        & $(1,L,q_7)$
        & $(1,L,q_7)$\\

        \hline

        & $q_{9}$
        & $q_{10}$
        & $q_{11}$
        & $q_{12}$
        & $q_{13}$
        & $q_{14}$
        & $q_{15}$
        & \\

        \hline

        $0$
        & $(0,R,q_1)$ 
        & $(1,L,q_{11})$ 
        & $(0,R,q_{12})$
        & $(0,R,q_{13})$
        & $(0,L,q_{2})$
        & $(0,L,q_{3})$
        & $(0,R,q_{14})$
        & \\

        $1$
        & $(1,L,q_{10})$
        & $\mathtt{HALT}$
        & $(1,R,q_{14})$
        & $(1,R,q_{12})$
        & $(1,R,q_{12})$
        & $(0,R,q_{15})$
        & $(1,R,q_{14})$
        &\\

        \hline
    \end{tabular}
    }
    \caption{Transition table of $U$, i.e., definition of $\delta$, as it is given in \cite{BREMER201215} or \cite{Neary2009}.}
    \label{tab:UTM-delta}
\end{table}

%% file: inclusion_long_nonerasing_abelAndRev_undec.tex
\begin{proof}
As this proof works almost identical to the proof of Bremer and Freydenberer in \cite{BREMER201215}, we only
give an extended sketch and focus on the formal intricacies that occur as we are now considering abelian equivalence instead of the equivalence relation.

Let $\Sigma = \{0,\#\}$ and $R := \{r_{ab}\}$.
To show undecidability of this problem, we reduce the problem of deciding whether $\mathtt{UnivValc}_{NE}(I) = \emptyset$ for
some initial configuration $I$ of $U$ to the problem of deciding whether $L_{NE}(p_1,\alpha_R) \subseteq L_{NE}(p_2,\beta_R)$
for two relational patterns $(\alpha,r_\alpha), (\beta,r_\beta)\in \relpat$.
To be specific, we construct two such patterns $(\alpha,r_\alpha), (\beta,r_\beta)\in \relpat$ such that
$L_{NE}(\alpha,r_\alpha) \subseteq L_{NE}(\beta,r_\beta)$ if and only if $\ValC_U(I) = \emptyset$.

We define the relational pattern $(\beta,r_\beta)$ by
$$ \beta := a_1b_1\ \#^5\ a_2x_{1,1}x_{2,1}\cdots x_{\mu,1}b_2\ \#^5\ r_1\hat{\beta}_1r_2\hat{\beta}_2\cdots r_\mu\hat{\beta}_\mu r_{\mu+1} $$
where for all $i\in\{1,...,\mu\}$ we have
$$ \hat{\beta}_i := 0x_{i,2}\cdots x_{i,5}0\ \upgamma_i\ 0x_{i,6}\cdots x_{i,9}0\ \updelta_i\ 0x_{i,10}\cdots x_{i,13}0$$
such that $a_1,a_2,b_1,b_2,r_i,x_{j,k}$ are variables for all $i\in\{1,...,\mu+1\}$, $j\in\{1,...,\mu\}$, and $k\in\{1,...,13\}$ with
$(r_{ab},a_1,a_2),(r_{ab},b_1,b_2),(r_{ab},x_{i,j},x_{i,j'})\in r_\beta$, for $i\in\{1,...,\mu\}$ and $j,j'\in\{1,...,13\}$. Also, $\upgamma_i\updelta_i \in \relpat$
are patterns for all $i\in\{1,...,\mu\}$ which are described later, but for which we already mention that the factor $\#^4$ never appears in them. Additionally, we have that each variable $r_i$ is not related to any other variable
in $\beta$ and we have that the variables $a_1,a_2,b_1,$ and $b_2$ are not related to any other variable than the ones mentioned here.
For the relational patterns $\upgamma_i$ and $\updelta_i$ we also mention that all variables occurring in $\upgamma_i\updelta_i$ are not related to any other
variable occurring in some other $\upgamma_j\updelta_j$ for $i\neq j$ with $i,j\in\{1,...,\mu\}$ and that each variable $x_{i,j}$ is not related to any other variable in
$\upgamma_k\updelta_k$ for $i,k\in\{1,...,\mu\}$ and $j\in\{1,...,13\}$.

Now, we define the second relational pattern $(\alpha,r_\alpha)$ by
$$ \alpha := 0^{\mu+1}\ \#^5\ 0^\mu\#0^\mu\ \#^5\ tv0\alpha_10v0\alpha_20vt $$
where $v := 0\#^40$, $\alpha_1$ is some relational pattern not containing $\#^3$ as a factor, $\alpha_2$ is a relational pattern not containing $\#$,
and $t$ is some other terminal string. As we shall see later, we have $r_\alpha = \emptyset$. This will become clear once we specifically define $\alpha_1$ and $\alpha_2$.
We define $t$ using the non-erasing substitution $\psi : (\var(\beta)\cup\Sigma)^* \rightarrow \Sigma^*$ with
$\psi(x) = 0$ for all $x\in\var(\beta)$. Using that, we say $t := \psi(r_1\hat{\beta}_1\cdots r_\mu\hat{\beta}_\mu r_{\mu+1})$.
We see that both patterns are basically constructed in the same way as in \cite{BREMER201215}, just using the formalism of relational
patterns, i.e. all variables which are equal in the construction in \cite{BREMER201215}, are in this construction distinct variables that are, however, related to each other through the sets $r_\alpha$ and $r_\beta$.

The following lemma follows directly from the proof in \cite{BREMER201215} which can be converted directly to the formalism using
substitutions over relational patterns.

\begin{lemma}\label{lemma:ne-includ-t-props}\cite{BREMER201215}
    All $\psi(\hat{\beta}_i)$ with $1 \leq i \leq \mu$ and $t$ begin and end with the terminal letter $0$ and do not contain the terminal word $\#^4$ as a factor.
\end{lemma}

We recall, that $H_{r_\alpha}$ denotes the set of all $r_\alpha$-valid substitutions.
Let $H_{r_\alpha^+}$ denote the subset of all non-erasing substitutions of $H_{r_\alpha}$.
Each pair $(\upgamma_i,\updelta_i)$ is interpreted as a predicate $\pi_i : H_{r_\alpha^+} \rightarrow \{0,1\}$ such that
any $h\in H_{r_\alpha^+}$ satisfies $\pi_i$ if there exists some $r_\beta$-valid non-erasing substitution
$\tau : (\var(\upgamma_i\updelta_i)\cap \Sigma)^* \rightarrow \Sigma^*$ with $\tau(\upgamma_i) = h(0\alpha_1 0)$ and $\tau(\updelta_i) = h(0\alpha_2 0)$.
We need to say that $\tau$ is $r_\beta$-valid for formal reasons, but as all variables in $\updelta_i\upgamma_i$ are only related to variables
in $\updelta_i\upgamma_i$ itself, we can consider such an isolated substitution as it does not interfere with the substitution of any other variable
of $\beta$. Using those predicates, we will see that $L_{NE}(\alpha,r_\alpha) \setminus L_{NE}(\beta,r_\beta)$ contains
those words $h(\alpha)$ for $h\in H_{r_\alpha^+}$ for which we have that $h$ does not satisfy any predicate $\pi_i$.
Essentially, all predicates $\pi_1$ to $\pi_\mu$ will give an exhaustive list of criteria for membership of $\overline{\ValC_U(I)}$.
We proceed with more technical preperations that we need for the easy and dynamic construction of predicates. Once this is settled, as
the constructed predicates from \cite{BREMER201215} also work in this case, we are done.

We say that a substitution $h \in H_{r_\alpha^+}$ is of \emph{bad form} if $h(0\alpha_10)$ contains a factor $\#^3$ 
or if $h(0\alpha_20)$ contains the letter $\#$. We use the predicates $\pi_1$ and $\pi_2$ to handle all such substitutions of bad form
and construct them exactly the same way as they are constructed in \cite{BREMER201215} by saying

\vspace{3mm}
\begin{tabular}{l l}
    $\upgamma_1 := y_{1,1}\#^ky_{1,2}$  & $\upgamma_2 := 0y_{2}0$ \\
    $\updelta_1 := 0 \hat{y}_1 0$       & $\updelta_2 := \hat{y}_{2,1} \# \hat{y}_{2,2}$
\end{tabular}
\vspace{3mm}

\noindent for new variables $y_{1,1}, y_{1,2}, y_{2}, \hat{y}_1, \hat{y}_{2,1}, \hat{y}_{2,2} \in X$ that are not related to any other variable in $\beta$.
As they are not related to any other variable, no intricacies regarding abelian equivalence in comparison to normal equality occur and we directly
get the following lemma from \cite{BREMER201215} as their proof needs no adaptation.

\begin{lemma}\label{lemma:ne-includ-bad-form}\cite{BREMER201215}
    A substitution $h\in H_{r_\alpha^+}$ is of \emph{bad-form} if and only if $h$ satisfies $\pi_1$ or $\pi_2$.
\end{lemma}

The following key lemma allows us to focus only on predicates from now on and is independent of the shape of the specific
predicates $\pi_3$ to $\pi_\mu$.

\begin{lemma}\label{lemma:ne-includ-predicate-language-contaiment}
    For every non-erasing substitution $h \in H_{r_\alpha^+}$, we have that 
    $h(\alpha) \in L_{NE}(\beta,r_\beta)$ if and only if $h$ satisfies some predicate $\pi_1$ to $\pi_\mu$.
\end{lemma}
\begin{proof}
    Due to the different formalism, we need to rewrite certain parts of the proof from \cite{BREMER201215}.
    However, the general idea is the same. We just have to show that it also works for abelian equivalence.

    ($\Rightarrow: $) Assume some $h\in H_{r_\alpha^+}$ satisfies some predicate $\pi_i$ for $i\in\{1,...,\mu\}$.
    Then there exists a non-erasing $r_\beta$-valid substitution $\tau : (\var(\upgamma_i\updelta_i)\cup \Sigma)^* \rightarrow \Sigma^*$ with
    $\tau(\upgamma_i) = h(0\alpha_10)$ and $\tau(\updelta_i) = h(0\alpha_20)$.
    We extend $\tau$ to a non-erasing $r_\beta$-valid substitution $\tau' \in H_{r_\beta^+}$ by 
    \begin{enumerate}
        \item $\tau'(x) := \begin{cases} \tau(x) & \text{ for all } x\in \var(\upgamma_i\updelta_i) \\ 0 & \text{ for all } x \in \var(\upgamma_j\updelta_j) \text{ with } j \neq i \end{cases}$
        \item $\tau'(x_{j,k}) := \begin{cases} \# & \text{ for } j = i, k\in\N \\ 0 & \text{ for } j \neq i, k\in\N \end{cases}$
        \item $\tau'(r_j) := \begin{cases}  
                                \psi(r_i\hat{\beta}_i \cdots r_\mu\hat{\beta}_\mu r_{\mu+1}) & \text{ for } j = i \\
                                \psi(r_1\hat{\beta}_{1} \cdots r_i\hat{\beta}_ir_{i+1} & \text{ for } j = i+1) \\
                                0 & \text{ else, }
                             \end{cases}$
        \item $\tau'(a_k) := 0^{\mu-i+1}$ for $k\in\{1,2\}$, and 
        \item $\tau'(b_k) := 0^i$ for $k\in\{1,2\}$.
    \end{enumerate}
    Notice that none of the variables in one factor $\var(\upgamma_i\updelta_i)$ appear outside of  
    $\upgamma_i$ and $\updelta_i$. By that, we can always define $\tau'$ that way. Also note, that the only significant
    difference to \cite{BREMER201215} lies in the definition of $\tau'(x_{j,k})$, $\tau'(a_k)$, and $\tau'(b_k)$ as no 
    variables can occur multiple timesin relational patterns; the rest remains the same.

    As in \cite{BREMER201215} we obtain $\tau'(\upgamma_i) = \tau(\gamma_i) = h(0\alpha_10)$,
    $\tau'(\updelta_i) = \tau(\updelta_i) = h(0\alpha_1 0)$, and
    \begin{align*}
    \tau'(a_1b_1\#^5 a_2 x_{1,1}\cdots x_{\mu,1}b_2\#^5) & = 0^{\mu-i+1}0^i\#^50^{\mu-i+1}0^{i-1}\#0^{\mu-i}0^i\#^5 \\
                                                 & = 0^{\mu+1}\#^50^\mu\#0^\mu\#^5
    \end{align*}
    as well as 
    \begin{align*}
        \tau'(\hat{\beta}_i) & = \tau'(0 x_{i,2}\cdots x_{i,5}0\upgamma_i0x_{i,6}\cdots x_{i,9}0\updelta_i0x_{i,10}\cdots x_{i,13}0) \\
                             & = 0\#^40\tau'(\upgamma_i)0\#^40\tau'(\updelta_i)0\#^40 \\
                             & = v h(0\alpha_10) v h(0\alpha_20) v.
    \end{align*}
    Also, as $\tau'(x) = \psi(x)$ for all $x\in \var(\beta)_j$ for $j \neq i$, we get for all $j \neq i$ that $\tau'(\hat{\beta}_j) = \psi(\hat{\beta}_j)$
    and by that we finally obtain
    \begin{align*}
        \tau'(\beta) = &\ \tau'(a_1b_2\#^5a_2x_{1,1}\cdots x_{\mu,1}\#^5r_1\hat{\beta}_1r_2\cdots r_\mu\hat{\beta}_\mu r_{\mu+1}) \\
                     = &\ 0^{\mu+1}\#^50^\mu\#0^\mu\#^5 \tau'(r_1\hat{\beta}_1r_2\cdots r_\mu\hat{\beta}_\mu r_{\mu+1}) \\
                     = &\ 0^{\mu+1}\#^50^\mu\#0^\mu\#^5 \tau'(r_1\hat{\beta}_1r_2\cdots r_{i-1} \hat{\beta}_{i-1}) \tau'(r_i) \tau'(\hat{\beta}_i) \tau'(r_{i+1}) \\
                       &\ \tau'(\hat{\beta}_{i+1} \cdots r_\mu\hat{\beta}_\mu r_{\mu+1}) \\
                     = &\ 0^{\mu+1}\#^50^\mu\#0^\mu\#^5 \psi(r_1\hat{\beta}_1r_2\cdots r_{i-1} \hat{\beta}_{i-1}) \psi(r_i\hat{\beta}_i \cdots r_\mu \hat{\beta}_\mu r_{\mu+1}) \tau'(\hat{\beta}_i) \\
                       &\ \psi(r_1\hat{\beta}_1 r_2 \cdots r_i \hat{\beta}_i r_{i+1}) \psi(\hat{\beta}_{i+1} \cdots r_\mu\hat{\beta}_\mu r_{\mu+1}) \\
                     = &\ 0^{\mu+1}\#^50^\mu\#0^\mu\#^5 \psi(r_1\hat{\beta}_1r_2 \cdots r_\mu\hat{\beta}_\mu r_{\mu+1}) \tau'(\hat{\beta}_i) \psi(r_1\hat{\beta}_ir_2 \cdots r_\mu\hat{\beta}_\mu r_{\mu+1}) \\
                     = &\ 0^{\mu+1}\#^50^\mu\#0^\mu\#^5 t \tau'(\hat{\beta}_i) t \\
                     = &\ 0^{\mu+1}\#^50^\mu\#0^\mu\#^5 t v h(0\alpha_10) v h(0\alpha_20) v t \\
                     = &\ h(\alpha).
    \end{align*}
    By that we get $h(\alpha) \in L_{NE}(\beta,r_\beta)$.

    ($\Leftarrow$): Now, assume $h(\alpha) \in L_{NE}(\beta,r_\beta)$.
    In the case of $h$ being of bad form, we know by Lemma \ref{lemma:ne-includ-bad-form} that $h$ satisfies $\pi_1$ or $\pi_2$.
    So, assume $h(0\alpha0)$ does not contain $\#^3$ as a factor and that $h(0\alpha_20)$ has no occurrence of the letter $\#$, hence being
    of the form $h(0\alpha_20) \in 00^+0$. As $h(\alpha) \in L_{NE}(\beta,r_\beta)$, there exists some non-erasing $r_\beta$-valid substitution
    $\tau\in H_{r_\beta^+}$ such that $\tau(\beta) = h(\alpha)$.

    As $h$ is of good form and as $t$ begins and ends with $0$ and does not contain any occurrence of $\#^4$ by
    Lemma \ref{lemma:ne-includ-t-props}, we know that $h(\alpha)$ contains the factor $\#^5$ exactly twice.
    That is also the case for $\tau(\beta)$ as $h(\alpha) = \tau(\beta)$.
    Hence, we can decompose the equation $h(\alpha) = \tau(\beta)$ into the following system consisting of the equations
    \begin{align}
        0^{\mu+1} & = \tau(a_1b_1) \\
        0^\mu\#0^\mu & = \tau(a_2x_{1,1}\cdots x_{\mu,1}b_2) \\
        tvh(0\alpha_10)vh(0\alpha_20)vt & = \tau(r_1\hat{\beta}_1r_2\cdots r_\mu\hat{\beta}_\mu r_{\mu+1}).
    \end{align}
    We see that in equation $2$ each $x_{i,1}$ for $i\in\{1,...,\mu\}$ we have that it has to be substituted by a single symbol, 
    as otherwise equation $1$ would be unsatisfied,
    because $(r_{ab},a_1,a_2),(r_{ab},b_1,b_2)\in r_\beta$. Also, we get by that that $a_2$ and $b_2$ can also be substituted to strings consisting of
    only $0's$. So, for some single $j\in\{1,...,\mu\}$ we have $\tau(x_j) = \#$ and for all other $j' \neq j$ we have $\tau(x_j') = 0$.
    Now, as $(r_{ab},x_{j,k},x_{j,k'})\in r_\beta$ for $k,k'\in\{1,...,13\}$, we get for that specific $j$ that 
    \begin{align*}
        \tau(\hat{\beta}_j) &= \tau(0x_{j,2}\cdots x_{j,5}0\upgamma_i0x_{j,6}\cdots x_{j,9}0\updelta_j0x_{j,10}\cdots x_{j,13}0) \\
                            &= 0\#^40 \tau(\upgamma_i) 0\#^40 \tau(\updelta_i) 0\#^40 \\
                            &= v \tau(\upgamma_i) v \tau(\updelta_j) v.
    \end{align*}
    This can be used to convert the right side of equation $3$ in the following manner.
    \begin{align*}
          &\ \tau(r_1\hat{\beta}_1r_2\cdots r_\mu\hat{\beta}_\mu r_{\mu+1}) \\
        = &\ \tau(r_1\hat{\beta}_1r_2\cdots r_{j-1}\hat{\beta}_{j-1}r_j)\tau(\hat{\beta}_j)\tau(r_{j+1}\hat{\beta}_{j+1}\cdots r_\mu\hat{\beta}_\mu r_{\mu+1}) \\
        = &\ \tau(r_1\hat{\beta}_1r_2\cdots r_{j-1}\hat{\beta}_{j-1}r_j)v \tau(\upgamma_i) v \tau(\updelta_j) v\tau(r_{j+1}\hat{\beta}_{j+1}\cdots r_\mu\hat{\beta}_\mu r_{\mu+1}) \\
    \end{align*}
    This results in the following adaptation of equation $3$.
    \begin{align*}
          &\ tvh(0\alpha_10)vh(0\alpha_20)vt \\
        = &\ \tau(r_1\hat{\beta}_1r_2\cdots r_{j-1}\hat{\beta}_{j-1}r_j)v \tau(\upgamma_i) v \tau(\updelta_j) v\tau(r_{j+1}\hat{\beta}_{j+1}\cdots r_\mu\hat{\beta}_\mu r_{\mu+1})
    \end{align*}
    We know that $h$ is of good form and we know that $t$ does not contain any factor $\#^4$.
    So, the left side of the equation contains exactly $3$ occurrences of $\#^4$ as a factor, that is, in each occurrernce of one $v$.
    The right side of the equation also contains exactly three occurrences of $v$ which themselves make up $3$ occurrences of $\#^4$.
    Thus, we obtain the following system of equations.
    \begin{align*}
        t &= \tau(r_1\hat{\beta}_1\cdots r_{j-1}\hat{\beta}_{j-1}r_j) \\
        h(0\alpha_10) &= \tau(\upgamma_j) \\
        h(0\alpha_20) &= \tau(\updelta_j) \\
        t &= \tau(r_{j+1}\hat{\beta}_{j+1}\cdots r_\mu\hat{\beta}_\mu r_{\mu+1})
    \end{align*}
    As these equations must hold and especially as $h(0\alpha_10) = \tau(\upgamma_j)$ and $h(0\alpha_20) = \tau(\updelta_j)$, we know
    that $h$ satisfies the predicate $\pi_j$.
\end{proof}    

By that, we can now select appropriate predicates $\pi_1$ to $\pi_\mu$ such that we obtain  
\[L_{NE}(\alpha,r_\alpha) \setminus L_{NE}(\beta,r_\beta) = \emptyset\]
if and only if $\ValC_U(I) = \emptyset$. Lemma \ref{lemma:ne-includ-predicate-language-contaiment} shows us that if
for some $h\in H_{r_\alpha^+}$ and $\tau\in H_{r_\beta^+}$ with $h(\alpha) = \tau(\beta)$ we have that
$h$ is of good form, then we can pick exactly one $i\in\{1,...,\mu\}$ such that we fulfill $\tau(0x_{i,j}\cdots x_{i,j+3}0) = 0\#^40 = v$ for $j\in\{2,6,9\}$.

For technical reasons, we define a predicate that sets a lower bound to the length of $h(\alpha_2)$ if unsatisfied for a substitution $h\in H_{r_\alpha^+}$
of good form. As in \cite{BREMER201215}, the predicate $\pi_3$ is defined by
\begin{align*}
    \upgamma_3 &:= y_{3,1}\hat{y}_{3,1}y_{3,2}\hat{y}_{3,2}y_{3,3}\hat{y}_{3,3}y_{3,4} \\
    \updelta_3 &:= 0\hat{y}_{3,1}\hat{y}_{3,2}\hat{y}_{3,3}
\end{align*}
for pairwise non-related variables $y_{3,1},...,y_{3,4}$ and $\hat{y}_{3,1},...,\hat{3,3}$.
As in \cite{BREMER201215}, if some $h\in H_{r_\alpha^+}$ satisfies $\pi_3$, then $h(\alpha_2)$ is a concatenation of three non-empty factors
of $h(\alpha_1)$. In other words, if some $h$ does not satisfy any of $\pi_1$,$\pi_2$, and $\pi_3$,
then $h(\alpha_2)\in 0^+$ has to be longer than the $3$ longest non-overlapping and non-touching factors of $0$'s in $h(\alpha_1)$.
This restriction allows for a simpler way to construct all other predicates $\pi_4$ to $\pi_\mu$.

As we will see, the next part regarding simple predicates is also very similar to the one defined and used in \cite{BREMER201215} but
had to be adapted formally to also work in the context of relational patterns.
Let $X_s := \{\hat{x}_{1,1},\hat{x}_{2,1},\hat{x}_{3,1},\hat{x}_{1,2},\hat{x}_{2,2},\hat{x}_{3,2}, ...\} \subset X$
and always assume that $(r_{ab},\hat{x}_{i,j},\hat{x}_{i,j'})\in r_p$ for $i\in\{1,2,3\}$ and $j,j'\in\N$ in any set of relational constraints $r_p$ of any relational pattern $(p,r_p)\in \relpat$
if any variable of $X_s$ is used in $p$.
Let $G_{r_\alpha^+}\subset H_{r_\alpha^+}$ denote the subset of $r_\alpha$-valid substitutions that are of good form
and let $S$ be the set of all non-erasing substitutions $s : (\Sigma \cup X_s)^* \rightarrow \Sigma^*$ for which $s(\hat{x}_{i,j})\in \{0\}^+$, for
all $i\in\{1,2,3\}$ and $j\in\N$, and for which we have $s(\hat{x}_{i,j}) = s(\hat{x}_{i',j'})$ if $i = i'$ for $i,i'\in\{1,2,3\}$ and $j,j'\in\N$.
For any relational pattern $(p_s,r_{p_s})$ with $p_s\in(\Sigma\cup X_s)^*\cap\relpat$ we define $S(p_s,r_{p_s}) := \{s(p_s)|s \in S\}$.
The next definition is exactly the same as in \cite{BREMER201215}, just adapted to the notation of relational patterns.

\begin{definition}\cite{BREMER201215}
    A predicate $\pi : G_{r_\alpha^+} \rightarrow \{0,1\}$ is called a NE-simple predicate for $0\alpha_10$ if there exists a relational
    pattern $(p_s,r_{p_s})$ with $p_s\in(\Sigma\cup X_s)^*\cap\relpat$ and $(r_{ab},\hat{x}_{i,j},\hat{x}_{i,j'})\in r_{p_s}$ for $i\in\{1,2,3\}$ and $j,j'\in\N$
    and languages $L_1\in\{0\Sigma^*,\{0\}\}$ and $L_2\in\{\Sigma^*0,\{0\}\}$ such that any non-erasing substitution $h\in H_{r_\alpha^+}$ satisfies
    $\pi$ if and only if $h(0\alpha_10) \in L_1 S(p_s,r_{p_s}) L_2$.
    If $L_1 = 0\Sigma^*$ and $L_2 = \Sigma^*0$, then we call $\pi$ an infix-predicate. If $L_1 = \{0\}$ and $L_2 = \Sigma^*0$, then we call
    $\pi$ a prefix-predicate in if its the other way around, we call $\pi$ a suffix predicate.
\end{definition}

We understand elements of $X_s$ as numerical paramters that describe (concatenative) powers of $0$ with non-erasing substitutions
$s : (\Sigma \cup X_s)^* \rightarrow \Sigma^*$ acting as assignments. The next lemma is the final step we need to take
in order to finish the adaptation of the original proof to the new case. It shows that for each simple predicate $\pi_s$, we
can construct a predicate $\pi$ defined over some $\upgamma$ and $\updelta$ such that each substitution of good form that
would have satisfied $\pi_s$ now satisfies either $\pi_3$ or $\pi_s$ which suffices the demands of the construction.
As we can see, the argument is again very similar to the one in \cite{BREMER201215} and had to be adapted mainly to the new formalism.

\begin{lemma}\label{lemma:ne-includ-simple-pi}
    For every NE-simple predicate $\pi_s$, there exists a predicate $\pi$ defined by relational patterns $(\upgamma,r_{\upgamma})$ and $(\updelta,r_{\updelta})$
    which can be embedded in $(\beta,r_\beta)$ such that for all non-erasing substitutions $h\in G_{r_\alpha^+}$ we have
    \begin{enumerate}
        \item if $h$ satisfies $\pi_s$, then $h$ also satisfies $\pi_3$ or the newly constructed $\pi$, and
        \item if $h$ satisfied the newly constructed $\pi$, then $h$ also satisfies $\pi_s$.
    \end{enumerate}
\end{lemma}
\begin{proof}
    First, consider the case of $L_1 = 0\Sigma^*$ and $L_2 = \Sigma^*0$.
    Assume $\pi_s$ is a NE-simple predicate and that $(p_s,r_{p_s})$ is a relational pattern with $p_s\in(\Sigma\cup X_s)^*\cap\relpat$ and
    $r_{p_s} := \{(r_{ab},\hat{x}_{i,j},\hat{x}_{i,j'}) | \hat{x}_{i,j},\hat{x}_{i,j'}\in X_s; i\in\{1,2,3\}; j,j'\in\N\}$ such that $h\in G_{r_\alpha^+}$
    satisfies $\pi_s$ if and only if $h(0\alpha_10) \in L_1S(p_s,r_{p_s})L_2$.
    Then we can define the relational patterns $(\upgamma,r_\upgamma)$ by 
    $\upgamma = y_1p_sy_2$ for two new and to no other variable related variables $y_1$ and $y_2$
    and set $r_\upgamma = r_{p_s}$. Additionally, we define the relational pattern $(\updelta,r_\updelta)$ by
    $\updelta = 0\theta\hat{y}0$ where $\theta := \{ \hat{i,y} | \exists\hat{x}_{i,j}\in\var(p_s) \text{ and } y = |\var(p_s)|+1 \}$ and
    $\hat{y}$ is a new to no other variable related variable, and by setting $r_\updelta = r_\upgamma$ (notice that for as long as we do not merge 
    the relational pattern $(\updelta,r_\updelta)$ with any other relational pattern, that no variable in $\updelta$ is related to each other,
    however, if we merge $(\updelta,r_\updelta)$ with $(\upgamma,r_\upgamma)$, then the corresponding $\hat{x}_{i,j}$ variables become related to
    each other. Also notice, that each variable in either of those two relational patterns will not be related to any other variable in $\beta$,
    allowing for them to be directly embedded into $(\beta,r_\beta)$ without further steps).

    (1.) Assume that $h\in G_{r_\alpha^+}$ satisfies $\pi_s$.
    Then there exist words $w_1\in 0\Sigma^*$ and $w_2\in\Sigma^*0$ and a non-erasing substitution $s\in S$ such that
    $h(0\alpha_10) = w_1s(p_s)w_2$.
    If $h(\alpha)_2$ is shorter than the $3$ longest non-overlapping and non-touching factors of the form $0^+$ in $h(\alpha_1)$ concatenated,
    then $\pi_3$ is satisfied.
    If that is not the case, we define a substitution $\tau$ by setting $\tau(y_1) := w_1$, $\tau(y_2) = w_2$ and $\tau(\hat{x}_{i,j}) := s(\hat{x}_{i,j})$
    for all $i\in\{1,2,3\}$ and $j\in\N$. We also set $\tau(\hat{y}) = 0^m$ where
    $$ m := |h(\alpha_2)| - \sum_{i\in\{1,2,3\}}\begin{cases}\tau(\hat{x}_{i,j}) & \text{ if there exists some } x_{i,j}\in\var(p_s) \text{ f.s. } j\in\N \\ 0 & \text{ otherwise. } \end{cases}$$
    As $h$ does not satisfy $\pi_3$, we have $m > 0$. We obtain
    \begin{align*}
        \tau(\upgamma) &= \tau(y_1)\tau(p_s)\tau(y_2) = w_1s(p_s)w_2 = h(0\alpha_1 0) \\
        \tau(\updelta) &= 00^{|h(\alpha_2)|}0 = h(0\alpha_20).
    \end{align*}
    By that, we see that $h$ also satisfies the constructed predicate $\pi$.

    (2.) Now assume that $h\in G_{r_\alpha^+}$ satisfies the constructed predicate $\pi$ (assuming it is embedded in $\beta$).
    Hence, there exists a non-erasing substitution $\tau\in H_{r_\beta^+}$ with $h(0\alpha_10) = \tau(\upgamma)$ and $h(0\alpha_20) = \tau(\updelta)$.
    We get by the construction that $\tau(y_1) \in 0\Sigma^*$ and $\tau(y_2)\in \Sigma^*0$. Now define $s(\hat{x}_{i,j}) = \tau(\hat{x}_{i,j})$ for all 
    $\hat{x}_{i,j} \in \var(\updelta)$.
    We see that $h(0\alpha_10) \in L_1S(p_s,v_R)L_2$ holds.
    So, $h$ must satisfy the simple predicate $\pi_s$.

    The other cases, where $L_1$ or $L_2$ is chosen to be the language containing only the word $0$, are handled analogously by omitting $y_1$ or $y_2$
    and keeping the rest of the construction.
\end{proof}

So, in general, if $h$ is of good form but does not satisfy $\pi_3$, then $h(\alpha_2)\in 0^+$ is long enough to provide building
blocks for NE-simple predicates that use variables from $X_s$.

Lemma \ref{lemma:ne-includ-simple-pi} shows us that the construction of NE-simple predicates in \cite{BREMER201215} can be used in the
setting of relational patterns under abelian equivalence in exactly the same way. Mainly, this is caused by the fact that each variable $\hat{x}_{i,j}\in X_s$
is substituted by a unary string. By that, each related variable has to be substituted exactly in the same way as they are substituted in \cite{BREMER201215}, i.e. equally.
Having shown that the same simple predicates from \cite{BREMER201215} also work in this setting in exactly the same way, we define $\pi_4$ to $\pi_\mu$
exactly as they are defined in \cite{BREMER201215}, as all of those predicates are given described as simple predicates.

We set $\alpha_1 := \#\#enc_{NE}(I)\#\#x\#0^60^10\#\#$ and $\alpha_2 := y00$ for pairwise distinct variables $x$ and $y$ (notice, that $x$ and $y$
are the only variables in the pattern $\alpha$, hence for $(\alpha,r_\alpha)$ we have $r_\alpha = \emptyset$).
By the construction of $\pi_1$ to $\pi_\mu$ and essentially the same argument over the embedding of all invalid computations in $(\beta,r_\beta)$ from \cite{BREMER201215},
we obtain that $L_{NE}(\alpha,r_\alpha) \setminus L_{NE}(\beta,r_\beta) = \emptyset$ if and only if $\ValC_U(I) = \emptyset$. This concludes the case $|\Sigma| = 2$.

For the case of larger alphabets, i.e. alphabets of size $|\Sigma| \geq 3$, the exact construction in \cite{BREMER201215} can also be adapted to the formality of our
setting and works by the same arguments as in \cite{BREMER201215}. The relational patterns $(\alpha,r_\alpha)$ and $(\beta,r_\beta)$ need to be slightly
adapted to work with larger alphabets. Essentially, suffixes to both the patterns $\alpha$ and $\beta$ are added to allow for another new variable $\hat{x}$ to be mapped 
to any letter of $\Sigma$ (similar to how we ensured that one variable $x_{i,j}$ is mapped to $\#$). Then, $2$ new predicates are added to ensure that
if this letter occurs in any substitution of $h(\alpha_1)$ or $h(\alpha_2)$, then one of these predicates is satisfied. 

Assume $\Sigma = \{\ta_1, \ta_2, ... , \ta_\sigma\}$ with $\sigma \geq 3$.
The adapted relational patterns $(\tilde{\alpha},r_{\tilde{\alpha}})$ and $(\tilde{\beta},r_{\tilde{\beta}})$ are defined by
\begin{align*}
    \tilde{\alpha} &:= \alpha \#^5 0 \ta_1 \ta_2 \cdots \ta_\sigma 0 \#^5 0 \ta_1 \ta_2 \cdots \ta_\sigma 0 \#^5, \\
    \tilde{\beta} &= \beta \#^5 \tilde{y}_1 \tilde{x}_{1,1} \tilde{z} \#^5 \tilde{y}_2 \tilde{x}_{2,1} \tilde{z}_2 \#^5
\end{align*}
for new and pairwise non-related variables $\tilde{y}_1, \tilde{y}_2, \tilde{x}_{1,1}, \tilde{x}_{2,1}, \tilde{z}_1, \tilde{z_2}$.
We mention here that also the variables $\tilde{x}_{1,2}$ and $\tilde{x}_{2,2}$ will occur in $\tilde{\beta}$. For those, we
have $(r_{ab},\tilde{x}_{1,1},\tilde{x}_{1,2}),(r_{ab},\tilde{x}_{2,1},\tilde{x}_{2,2})\in 
r_{\tilde{\beta}}$.
To work with the new construction and the two new predicates defined below, 
we need an extended version of the definition of $\psi : (\var{\hat{\beta}_1...\hat{\beta}_{\mu+2}}) \rightarrow \Sigma^*$.
Now, we also have $\psi(\tilde{x}_1) = \psi(\tilde{x}_2) = \ta_1 \cdots \ta_\sigma$ and $\psi(x) = 0$ for all $x\in\var(\hat{\beta}_1 \cdots \hat{\beta}_{\mu+1})\setminus\{\tilde{x}_1,\tilde{x}_2\}$.
Finally, we define an adapted version of the two new predicates $\pi_{\mu+1}$ and $\pi_{\mu+2}$ by

\addtolength{\tabcolsep}{4pt}    
\vspace{3mm}
\begin{tabular}{l l}
    $\upgamma_{\mu+1} := y_{\mu+1,1}\ \tilde{x}_{1,2}\ y_{\mu+1,2}$  & $\upgamma_2 := 0\ y_{\mu+2}\ 0$ \\
    $\updelta_{\mu+1} := 0\ \hat{y}_{\mu+1}\ 0$       & $\updelta_2 := \hat{y}_{\mu+2,1}\ \tilde{x}_{2,2}\ \hat{y}_{\mu+2,2}$
\end{tabular}
\vspace{3mm}
\addtolength{\tabcolsep}{-4pt}

where $y_{\mu+1,1}, y_{\mu+1,2}, y_{\mu+2}, \hat{y}_{\mu+1}, \hat{y}_{\mu+2,1},$ and $\hat{y}_{\mu+2,2}$ are new and pairwise distinct and unrelated variables.
Remember, that we have $(r_{ab},\tilde{x}_{1,1},\tilde{x}_{1,2})\in r_{\tilde{\beta}}$ and $(r_{ab},\tilde{x}_{2,1},\tilde{x}_{2,2})\in r_{\tilde{\beta}}$. Now, using these constructions,
we obtain the exact argument as in \cite{BREMER201215} to obtain the final result. The main argument why this construction can be used by the same
argument is that it works over unary substitutions of $\tilde{x}_{1,1}, \tilde{x}_{1,2}, \tilde{x}_{2,1},$ and $\tilde{x}_{2,2}$ of length $1$.
Hence, using substitutions of this type, abelian equivalence behaves exactly the same as equality between variables.
Thus, we omit the details at this point and refer to \cite{BREMER201215} for a more detailed proof.

This concludes the case of larger alphabets of size $|\Sigma|\geq 3$. We obtain undecidability of the inclusion problem for non-erasing relational pattern languages under abelian equivalence for alphabets of size $|\Sigma|\geq 2$.
\end{proof}

%% file: definition_n2cmwi.tex
A \emph{nondeterministic 2-counter automaton without input} (see e.g. \cite{Iberra1978}) is a 4-tuple $A = (Q,\delta,q_0,F)$ which consists of a set of states $Q$, a transition function $\delta : Q \times \{0,1\}^2 \rightarrow \mathcal{P}(Q \times \{\-1,0,+1\}^2)$, an initial state $q_0\in Q$, and a set of accepting states $F \subseteq Q$. A \emph{configuration} of $A$ is defined as a triple $(q,m_1,m_2)\in Q\times\N\times\N$ in which $q$ indicates the current state and $m_1$ and $m_2$ indicate the contents of the first and second counter. We define the relation $\vdash_A$ on $Q\times\N\times\N$ by $\delta$ as follows. For two configurations $(p,m_1,m_2)$ and $(q,n_1,n_2)$ we say that $(p,m_1,m_2) \vdash_A (q,n_1,n_2)$ if and only if there exist $c_1,c_2\in\{0,1\}$ and $r_1,r_2\in\{-1,0,+1\}$ such that
\begin{enumerate}
	\item if $m_i = 0$ then $c_i = 0$, otherwise if $m_i > 0$, then $c_i = 1$, for $i\in\{1,2\}$,
	\item $n_i = m_i + r_i$ for $i\in\{1,2\}$,
	\item $(q,r_1,r_2)\in \delta(p,c_1,c_2)$, and
	\item we assume if $c_i = 0$ then $r_i \neq -1$ for $i\in\{1,2\}$.
\end{enumerate}
Essentially, the machine checks in every state whether the counters equal $0$ and then changes the value of each counter by at most one per transition before entering a new state. A \emph{computation} is a sequence of configurations. An \emph{accepting computation} of $A$ is a sequence $C_1,...,C_n\in (Q\times\N\times\N)^n$ with $C_1 = (q_0,0,0)$, $C_i \vdash_A C_{i+1}$ for all $i\in\{1,...,n-1\}$, and $C_n\in F\times\N\times\N$ for some $n\in\N$ with $n > 0$.

We \emph{encode} configurations of $A$ by assuming $Q = \{q_0,...,q_e\}$ for some $e\in\N$ and defining a function $\enc$ $:$ $Q\times\N\times\N \rightarrow \{0,\#\}^*$ by 
$$ \enc(q_i,m_1,m_2) := 0^{x+i}\#0^{c_1 + y_2\cdot m_1}\#0^{c_2 + y_2\cdot m_2} $$
for some numbers $x,c_1,y_2,c_1,y_2\in\N$. The values for these numbers depend on the construction of the respective proofs and are not specified here. This is extended to encodings of computations by defining for every $n\geq 1$ and every sequence $C_1,...,C_n\in Q\times\N\times\N$
$$ \enc(C_1,...,C_n) := \#\#\ \enc(C_1)\ \#\#\ ...\ \#\#\ \enc(C_n)\ \#\#. $$
Encodings of this kind are used to prove different undecidability results regarding the inclusion problem for erasing relational pattern languages for various considered relations, in particular the results in Subsection~\ref{subsec:inc-era}. 
For some nondeterministic 2-counter automaton without input $A$, define the set of encodings of accepting computations 
$$\mathtt{ValC}(A) := \{\enc(C_1,...,C_n)\ |\ C_1,...,C_n \text{ is an accepting computation of } A\}$$
and let $\mathtt{InvalC}(A) = \{0,\#\}^*\setminus \mathtt{ValC}(A)$. 
The emptiness problem for deterministic 2-counter-automata is undecidable (cf. e.g. \cite{Iberra1978,Minsky1961}), thus it is also undecidable whether a nondeterministic 2-counter automaton without input has an accepting computation \cite{Freydenberger2010, Jiang1995}. That the emptiness problem for universal Turing machines is undecidable is a known fact.

%% file: inclusion_long_erasing_rev_undec.tex
\begin{proof}
    The proof is based on the construction by Freydenberger and Reidenbach~\cite{Freydenberger2010}. For better understanding for those who don't know the construction and for a really detailed verification, we copied the proof of Theorem 2 of~\cite{Freydenberger2010} and adapted it at several points. The main difference is the definition of $u$ and the definition of the delimiter. Why these adaptations work is shorter described in the proof of the case $r_{ab}$ for those who already are familiar with the construction. 

    We begin with the case $|\Sigma| = 2$, so let $\Sigma :=\{0, \#\}$. Let $A := (Q, \delta, q_0, F)$ be a nondeterministic 2-counter automaton; w.l.o.g.\ let $Q :=\{q_0, \ldots, q_s\}$ for some $s \in \mathbb{N}_0$. Our goal is to construct patterns $\alpha_A, \beta_A \in \relpat$ such that $L_E(\alpha_A,r_{\alpha_A}) \subseteq L_E(\beta_A, r_{\beta_A})$ if and only if $\mathtt{ValC}(A) = \emptyset$. We note, here, that in the remaining parts of this proof, we sometimes use $(x,y)\in R$ instead of $(r_{rev},x,y)\in r_{\beta_A}$. This is done for readability reasons and due to the fact that the relational pattern $(\alpha,r_{\alpha_A})$ does not use any relational constraints at all. We define
    $$\alpha_A := v\ v\ \#^6\ v\ x\ v\ y\ v\ \#^6\ v\ u\ v$$
    where $x,y$ are distinct variables, $v = 0\#^30$ and $u = 0\#^50$. Furthermore, for a yet unspecified $\mu \in \mathbb{N}$ that shall be defined later, let 
    \[ \beta_A:=x_1x_{1'}\ \ldots x_{\mu}x_{\mu'}\#^6\hat{\beta}_1\ldots\hat{\beta}_{\mu}\#^6\ddot{\beta_1}\ldots\ddot{\beta_{\mu}}, \]
    with, for all $i \in \{1, \ldots, \mu\}$, $\hat{\beta}_i := x_{i_1}\ \gamma_i\ x_{i_2}\ \delta_i\ x_{i_3}$ and $\ddot{\beta_i} := x_{i_4}\ \eta_i\ x_{i_5}$, where $x_1$, $x_{1'}$, $x_{1_1}$, $x_{1_2}$ ,$x_{1_3}$ ,$x_{1_4}$ ,$x_{1_5}$, $\ldots$, $x_{\mu}$, $x_{\mu'}$, $x_{\mu_1}$, $x_{\mu_2}$, $x_{\mu_3}$, $x_{\mu_4}$, $x_{\mu_5}$ are distinct variables with $(x_i, x_{i'}), (x_i, x_{i_1}), (x_i, x_{i_2}),(x_i,x_{i_3}),(x_i,x_{i_4}),(x_i,x_{i_5}) \in R$ and all 
    
    \noindent
    $\gamma_i, \delta_i, \eta_i \in X^*$ are terminal-free patterns. The patterns $\gamma_i$ and $\delta_i$ shall be defined later; for now, we only mention:
    \begin{enumerate}
        \item $\eta_i := z_{i}\ \hat{z}_{i_1}\ \hat{z}_{i_2}\ \hat{z}_{i_3}\ \hat{z}_{i_4}\ \hat{z}_{i_5}\ z_{i'}$,
        \item var($\gamma_i\delta_i\eta_i$) $\cap$ var($\gamma_j\delta_j\eta_j$) = $\emptyset$ for all $i,j \in \{1, \ldots, \mu\}$
    \end{enumerate}
    Thus, for every $i$, the elements of var($\gamma_i\delta_i\eta_i$) appear nowhere but in these three factors. Let $H$ be the set of all substitutions $\sigma :(\Sigma\ \cup\ \{x,y\})^* \rightarrow \Sigma^*$. We interpret each triple ($\gamma_i, \delta_i, \eta_i$) as a predicate $\pi_i: H \rightarrow \{0,1\}$ in such a way that $\sigma \in H$ \textit{satisfies} $\pi_i$ if there exists a morphism $\tau: \text{var}(\gamma_i\delta_i\eta_i)^* \rightarrow \Sigma^*$ with $\tau(\gamma_i) = \sigma(x)$, $\tau(\delta_i) = \sigma(y)$ and $\tau(\eta_i) = u$ - in the terminology of word equations (cf. Karhumäki et al.~\cite{DBLP:journals/jacm/KarhumakiMP00}), this means that $\sigma$ satisfies $\pi_i$ if and only if the system consisting of the three equations $\tau(\gamma_i) = \sigma(x)$, $\tau(\delta_i) = \sigma(y)$ and $\tau(\eta_i) = u$ has a solution $\tau$. Later, we shall see that $L_E(\alpha_A,r_{\alpha_A}) \setminus L_E(\beta_A,r_{\beta_A})$ exactly contains those $\sigma(\alpha_A)$ for which $\sigma$ does not satisfy any of $\pi_1$ to $\pi_{\mu}$, and choose these predicates to describe $\mathtt{InvalC}(A)$. The encoding of $\mathtt{InvalC}(A)$ shall be handled by $\pi_4$ to $\pi_{\mu}$, as each of these predicates describes a sufficient criterium for membership in $\mathtt{InvalC}(A)$. But at first we need a considerable amount of technical preparations. A substitution $\sigma$ is of \textit{good form} if $\sigma(x) \in \{0, \#\}^*$, $\sigma(x)$ does not contain $\#^3$ as a factor, and $\sigma(y) \in 0^*$. Otherwise, $\sigma$ is of \textit{bad form}. The predicates $\pi_1$ and $\pi_2$ handle all cases where $\sigma$ is of bad form and are defined through
    
    \begin{table}[h]
        \centering
        \begin{tabular}{ll}
          $\gamma_1 := y_{1,1}\ \hat{z}_{1_6}\ \hat{z}_{1_7}\ \hat{z}_{1_8}\ y_{1,2},$   & $\gamma_2 := y_2,$ \\
           $\delta_1 := \hat{y}_1,$  & $\delta_2 := \hat{y}_{2,1}\ \hat{z}_{2_6}\ \hat{y}_{2,2}$
        \end{tabular}
        \label{tab:my_label}
    \end{table}
    
    where $y_{1,1}$, $y_{1,2}$, $y_2$, $\hat{y}_1$, $\hat{y}_{2,1}$, $\hat{y}_{2,2}$, $\hat{z}_{1_6}$, $\hat{z}_{1_7}$, $\hat{z}_{1_8}$ and $\hat{z}_{2_6}$ are pairwise distinct variables with $(\hat{z}_{1_1}, \hat{z}_{1_6})$, $(\hat{z}_{1_1}, \hat{z}_{1_7})$, $(\hat{z}_{1_1}, \hat{z}_{1_8})$, $(\hat{z}_{2_1}, \hat{z}_{2_6}) \in R$. Recall that 

    \noindent
    $\eta_i := z_{i}\ \hat{z}_{i_1}\ \hat{z}_{i_2}\ \hat{z}_{i_3}\ \hat{z}_{i_4}\ \hat{z}_{i_5}\ z_{i'}$ for all $i$. It is not very difficult to see that $\pi_1$ and $\pi_2$ characterize the morphisms that are of bad form:
    \begin{lemma}\label{lem:pi1}
        A substitution $\sigma \in H$ is of bad form if and only if $\sigma$ satisfies $\pi_1$ or $\pi_2$.
    \end{lemma}
    \begin{proof}
        We begin with the \textit{only if} direction. If $\sigma(x) = w_1\#^3w_2$ for some $w_1,w_2 \in \Sigma^*$, choose $\tau(y_{1,1}) := w_1$, $\tau(y_{1,2}) := w_2$, $\tau(\hat{z}_{1_1}) := \#$, $\tau(\hat{z}_{1_2}) := \#$, $\tau(\hat{z}_{1_3}) := \#$, $\tau(\hat{z}_{1_4}) := \#$, $\tau(\hat{z}_{1_5}) := \#$, $\tau(\hat{z}_{1_6}) := \#$, $\tau(\hat{z}_{1_7}) := \#$, $\tau(\hat{z}_{1_8}) := \#$, $\tau(\hat{y}_1) := \sigma(y)$, $\tau(z_1) := 0$ and $\tau(z_{1'}) := 0$. Then $\tau(y_1) = \sigma(x), \tau(\delta_1) = \sigma(y)$ and $\tau(\eta_1) = u$; thus $\sigma$ satisfies $\pi_1$.

        If $\sigma(y) = w_1 \# w_2$, for some $w_1,w_2 \in \Sigma^*$, let $\tau(y_2) := \sigma(x)$, $\tau(\hat{y}_{2,1}):= w_1$, $\tau(\hat{y}_{2,2}) := w_2$, and $\tau(\hat{z}_{2_1}) := \#$, $\tau(\hat{z}_{2_2}) := \#$, $\tau(\hat{z}_{2_3}) := \#$, $\tau(\hat{z}_{2_4}) := \#$, $\tau(\hat{z}_{2_5}) := \#$, $\tau(\hat{z}_{2_6}) := \#$, and $\tau(z_2) := 0, \tau(z_{2'}) := 0$. It is easy to see that $\sigma$ satisfies $\pi_2$.

        For the \textit{if} direction, if $\sigma$ satisfies $\pi_1$, then there exists a morphism $\tau$ with $\tau(\gamma_1) = \sigma(x)$ and $\tau(\eta_1) = 0 \#^50$. Thus, $\tau(\hat{z}_{1_1}) = \tau(\hat{z}_{1_2}) = \tau(\hat{z}_{1_3}) = \tau(\hat{z}_{1_4}) = \tau(\hat{z}_{1_5}) = \#$ and $\tau(z_1) = \tau(z_{1'}) = 0$ must hold. Consequently, $\tau(\hat{z}_{1_6}) = \tau(\hat{z}_{1_7}) = \tau(\hat{z}_{1_8}) = \#$ and thus,
        $\sigma(x)$ contains $\#^3$, and $\sigma$ is of bad form.

        Analogously, if $\sigma$ satisfies $\pi_2$, then $\sigma(y)$ contains the letter $\#$, and $\sigma$ is of bad form.
    \end{proof}

    This allows us to make the following observation, which serves as the central part of the construction and is independent from the exact shape of $\pi_3$ to $\pi_{\mu}$:

    \begin{lemma}\label{lem:pi}
        For every substitution $\sigma \in H$, $\sigma(\alpha_A) \in L_E(\beta_A, r_{\beta_A})$ if and only if $\sigma$ satisfies one of the predicates $\pi_1$ to $\pi_{\mu}$.
    \end{lemma}
    \begin{proof}
        We begin with the \textit{if} direction. Assume $\sigma \in H$ satisfies some predicate $\pi_i$. Then there exists a morphism $\tau : \text{var}(\gamma_i\delta_i\eta_i) \rightarrow \Sigma^*$ such that $\tau(\gamma_i) = \sigma(x), \tau(\delta_i) = \sigma(y)$ and $\tau(\eta_i)=u$. We extend $\tau$ to a substitution $\tau'$ defined by
        \begin{enumerate}
            \item $\tau'(x) := \tau(x)$ for all $x \in \text{var}(\gamma_i\delta_i\eta_i)$,
            \item $\tau'(x_i) := 0 \#^30 = v$ and $\tau'(x_{i'}) := \tau'(x_{i_1}) := \tau'(x_{i_2}) := \tau'(x_{i_3}) := \tau'(x_{i_4}) := \tau'(x_{i_5}) := 0\#^30 = v$,
            \item $\tau'(0) := 0$ and $\tau'(\#) := \#$,
            \item $\tau'(x) := \epsilon$ in all other cases.
        \end{enumerate}
        By definition, none of the variables in var($\gamma_i\delta_i\eta_i$) appear outside of these factors. Thus, $\tau'$ can always be defined in this way. We obtain
        \begin{align*}
            \tau'(\hat{\beta}_i) &= \tau'(x_{i_1}\ \gamma_i\ x_{i_2}\ \delta_i\ x_{i_3}) \\ &= v\ \tau(\gamma_i)\ v\ \tau(\delta_i)\ v \\ &= v\ \sigma(x)\ v\ \sigma(y)\ v, \\\tau'(\ddot{\beta_i})\ &= \tau'(x_{i_4}\ \eta_i\ x_{i_5}) \\ &= v\ \tau(\eta_i)\ v \\ &= v\ u\ v.
        \end{align*}
        As $\tau'(\gamma_j) = \tau'(\delta_j) = \tau'(\eta_j) = \tau'(\hat{\beta}_j) = \tau'(\ddot{\beta_j}) = \epsilon$ for all $j \neq i$, this leads to
        \begin{align*}
            \tau'(\beta_A) &= \tau'(x_1x_{1'}\ \ldots x_{\mu}x_{\mu'}\#^6\hat{\beta}_1\ldots\hat{\beta}_{\mu}\#^6\ddot{\beta_1}\ldots\ddot{\beta_{\mu}}) \\ &= \tau'(x_ix_{i'})\#^6\tau'(\hat{\beta}_i)\#^6\tau'(\ddot{\beta_i}) \\ &= v\ v\ \#^6\ v\ \sigma(x)\ v\ \sigma(y)\ v\ \#^6\ v\ u\ v \\&= \sigma(\alpha_A)
        \end{align*}
        This proves $\sigma(\alpha_A) \in L_E(\beta_A,r_{\beta_A})$.

        For the other direction, assume that $\sigma(\alpha_A) \in L_E(\beta_A,r_{\beta_A})$. If $\sigma$ is of bad form, then by \Cref{lem:pi1}, $\sigma$ satisfies $\pi_1$ of $\pi_2$. Thus, assume $\sigma(x)$ does not contain $\#^3$ as a factor, and $\sigma(y) \in 0^*$. Let $\tau$ be a substitution with $\tau(\beta_A) = \sigma(\alpha_A)$.

        Now, as $\sigma$ is of good form, $\sigma(\alpha_A)$ contains exactly two occurrences of $\#^6$, and these are non-overlapping. As $\sigma(\alpha_A) = \tau(\beta_A)$, the same holds for $\tau(\beta_A)$. Thus, the equation $\sigma(\alpha_A) = \tau(\beta_A)$ can be decomposed into the system consisting of the following three equations:
        \begin{align}
            &0\#^30\ 0\#^30 = \tau(x_1x_{1'}\ \ldots x_{\mu}x_{\mu'}), \label{eq:start}\\
            &0\#^30\ \sigma(x)\ 0\#^30\ \sigma(y)\ 0\#^30 = \tau(\hat{\beta}_1\ldots\hat{\beta}_{\mu}), \label{eq:middle}\\
            &0\#^30\ u\ 0\#^30 = \tau(\ddot{\beta_1}\ldots\ddot{\beta_{\mu}})\label{eq:end}
        \end{align}
        First, consider \Cref{eq:start} and choose the smallest $i$ for which $\tau(x_i) \neq \epsilon$. Then $\tau(x_i)$ has to start with $0$, and as
        $$\tau(x_1x_{1'}\ \ldots x_{\mu}x_{\mu'}) = 0\#^30\ 0\#^30$$
        and $(x_i,x_{i'})\in R$ for all $i \in \{1, \ldots, \mu\}$, it is easy to see that $\tau(x_i) = 0\#^30 = v=  \tau(x_{i'})$ and $\tau(x_j) = \tau(x_{j'}) = \epsilon$ for all $j \neq i$ must hold. Furthermore, as $(x_i, x_{i_1}), (x_i, x_{i_2}),(x_i,x_{i_3}),(x_i,x_{i_4}),(x_i,x_{i_5}) \in R$ for all $i \in \{1, \ldots, \mu\}$, we get $\tau(x_{i_1}) = \tau(x_{i_2}) = \tau(x_{i_3}) = \tau(x_{i_4}) = \tau(x_{i_5}) = v$ and  $\tau(x_{j_1}) = \tau(x_{j_2}) = \tau(x_{j_3}) = \tau(x_{j_4}) = \tau(x_{j_5}) = \epsilon$ for all $j \neq i$.

        Note that $u$ does not contain $0\#^30$ as a factor, and does neither begin with $\#^30$, nor end on $0\#^3$. But as $\tau(\ddot{\beta_i})$ begins with and ends on $0\#^30$, we can use \Cref{eq:end} to obtain $0\#^30\ u\ 0\#^30 = \tau(\ddot{\beta_i})$ and $\tau(\ddot{\beta_j}) = \epsilon$ for all $j \neq i$. As $\ddot{\beta_i} = x_{i_4}\ \eta_i\ x_{i_5}$ and $\tau(x_{i_4}) = \tau(x_{i_5}) = 0\#^30$, $\tau(\eta_i) = u$ must hold.

        As $\sigma$ is of good form, $\sigma(0\#^30\ x\ 0\#^30\ y\ 0\#^30)$ contains exactly three occurrences of $\#^3$. But there are already three occurrences of $\#^3$ in $\tau(\hat{\beta}_i) = 0\#^30\ \tau(\gamma_i)\ 0\#^30\ \tau(\delta_i)\ 0\#^30$. This, and \Cref{eq:middle}, lead to $\tau(\hat{\beta}_j) = \epsilon$ for all $j \neq i$ and, more importantly, $\tau(\gamma_i) = \sigma(x)$ and $\tau(\delta_i) = \sigma(y)$. Therefore, $\sigma$ satisfies the predicate $\pi_i$.
    \end{proof}

    Thus, we can select predicates $\pi_1$ to $\pi_{\mu}$ in such a way that $L_E(\alpha_A, r_{\alpha_A}) \setminus L_E(\beta_A,r_{\beta_A}) = \emptyset$ if and only if $\mathtt{ValC}(A) = \emptyset$ by describing $\mathtt{InvalC}(A)$ through a disjunction of predicates on $H$. The proof of \Cref{lem:pi} shows that if $\sigma(\alpha_A) = \tau(\beta_A)$ for substitutions $\sigma, \tau$, where $\sigma$ is of good form, there exists exactly one $i, 3 \leq i \leq \mu$, such that $\tau(x_i) = 0\#^30$.

    Due to technical reasons, we need a predicate $\pi_3$ that, if unsatisfied, sets a lower bound on the length of $\sigma(y)$, defined by
    \begin{align*}
        \gamma_3\ &:= \gamma_{3,1}\hat{\gamma}_{3,1}\gamma_{3,2}\hat{\gamma}_{3,2}\gamma_{3,3}\hat{\gamma}_{3,3}\gamma_{3,4}, \\
        \delta_3\ &:= \hat{\gamma}_{3,1'}\hat{\gamma}_{3,2'}\hat{\gamma}_{3,3'},
    \end{align*}
    where all of $\gamma_{3,1}$ to $\gamma_{3,4}$ and $\hat{\gamma}_{3,1}$ to $\hat{\gamma}_{3,3}$ and $\hat{\gamma}_{3,1'}$ to $\hat{\gamma}_{3,3'}$ are pairwise distinct variables with $(\hat{\gamma}_{3,1}, \hat{\gamma}_{3,1'})$, $(\hat{\gamma}_{3,2}, \hat{\gamma}_{3,2'})$, $(\hat{\gamma}_{3,3}, \hat{\gamma}_{3,3'}) \in R$. Clearly, if some $\sigma \in H$ satisfies $\pi_3$, $\sigma(y)$ is a concatenation of three (possibly empty) reversed factors of $\sigma(x)$. Thus, if $\sigma$ satisfies none of $\pi_1$ to $\pi_3$, $\sigma(y)$ must be longer than the three longest non-overlapping sequences of $0$s in $\sigma(x)$. This allows us to identify a class of predicates definable by a rather simple kind of expression, which we use to define $\pi_4$ to $\pi_{\mu}$ in a less technical way. The following definitions are adapted to work with our setting of relational patterns.

    Let $X' := \{\hat{x}_{1,1},\hat{x}_{2,1},\hat{x}_{3,1},\hat{x}_{1,2},\hat{x}_{2,2},\hat{x}_{3,2},\ldots\} \subset X$ be an infinite set of new variables, let $G_{v_{R}} \subset H_{v_{R}}$ denote the set of all $\alpha$-valid substitutions of good form, and let $S$ be the set of all substitutions $\rho: (\Sigma\ \cup\ X')^* \rightarrow \Sigma^*$ for which $\rho(0)=0, \rho(\#) = \#$ and $\rho(\hat{x}_{i,j}) \in 0^*$ for $i \in \{1,2,3\}$ and $j \in \mathbb{N}$. For relational patterns $(p,v_{R})$ with $p \in (\Sigma\ \cup\ X')^*$ and $v_{R} := \{(r_{rev},\hat{x}_{i,j},\hat{x}_{i,j'})\ |\ i \in \{1,2,3\};j,j' \in \mathbb{N}\}$, we define $S(p,v_{R}) := \{\rho(p)\ |\ \rho \in S \cap H_{v_R}\}$.

    \begin{definition}
        A predicate $\pi: G_{v_R} \rightarrow \{0,1\}$ is called a \textit{simple predicate} if there exists a pattern $p \in (\Sigma\ \cup\ X')^*$ and languages $L_1,L_2 \in \{\Sigma^*,\{\epsilon\}\}$ such that $\sigma$ satisfies $\pi$ if and only if $\sigma(x) \in L_1\ S(p, v_{R})\ L_2$.
    \end{definition}

    From a slightly different point of view, the elements of $X'$ can be understood as numerical parameters describing (concatenational) powers of $0$, with substitutions $\rho \in S$ acting as assignments. For example, if $\sigma \in G_{v_R}$ satisfies a simple predicate $\pi$ if and only if $\sigma(x) \in \Sigma^*S(\#\hat{x}_{1,1}\#\hat{x}_{2,1}\#\hat{x}_{1,2},v_R)$, we can also write that $\sigma$ satisfies $\pi$ if and only if $\sigma(x)$ has a suffix of the form $\#0^m\#0^n0\#0^m$ (with $m,n \in \mathbb{N}_0$), which could also be written as $\#0^m\#0^*0\#0^m$, as $n$ occurs only once in this expression. Using $\pi_3$, our construction is able to express all simple predicates: 
    \begin{lemma}\label{lem:simplepred}
        For every simple predicate $\pi_S$ over variables from $X'$, there exists a predicate $\pi$ defined by terminal-free patterns $\gamma,\delta,\eta$ such that for all substitutions $\sigma \in G_{v_R}$:
        \begin{enumerate}
            \item if $\sigma$ satisfies $\pi_s$, then $\sigma$ also satisfies $\pi$ or $\pi_3$,
            \item if $\sigma$ satisfies $\pi$, then $\sigma$ also satisfies $\pi_S$.
        \end{enumerate}
    \end{lemma}
    \begin{proof}
        We first consider the case of $L_1 = L_2 = \Sigma^*$. Assume that $\pi_S$ is a simple predicate, and $p \in (\Sigma\ \cup\ X')^*$ is a pattern such that $\sigma \in G_{v_R}$ satisfies $\pi_S$ if and only if $\sigma(x) \in L_1\ S(p,v_R)\ L_2$. Then define $\gamma := y_1p'y_2$, where $p'$ is obtained from $p$ by replacing all $k, k\in \mathbb{N}_0$ occurrences of $0$ with new variables $\tilde{z}_j, j \in [k]$ with $(\tilde{z}_1,\tilde{z}_j) \in R$ for all $j \in [k]$ and all $g, g \in \mathbb{N}_0$ occurrences of $\#$ with different variables $\ddot{z}_j, j \in [g]$ with $(\ddot{z}_1,\ddot{z}_j) \in R$ for all $j \in [g]$, while leaving all present elements of $X'$ unchanged. Furthermore, let $\delta := \hat{x}_{1,j}\hat{x}_{2,j}\hat{x}_{3,j}\hat{y}$ with $j \in \mathbb{N}$ such that $\hat{x}_{1,j},\hat{x}_{2,j},\hat{x}_{3,j}$ are new variables and (in order to stay consistent with the $\eta_i$ appearing in $\beta_A$) $\eta :=\tilde{z}\ \ddot{z}_{g+1}\ \ddot{z}_{g+2}\ \ddot{z}_{g+3}\ \ddot{z}_{g+4}\ \ddot{z}_{g+5}\ \tilde{z}'$ with $(\tilde{z}_1, \tilde{z})$, $(\tilde{z}_1,\tilde{z}')$, $(\ddot{z}_1,\ddot{z}_{g+1})$, $(\ddot{z}_1,\ddot{z}_{g+2})$, $(\ddot{z}_1,\ddot{z}_{g+3})$, $(\ddot{z}_1,\ddot{z}_{g+4})$, $(\ddot{z}_1,\ddot{z}_{g+5}) \in R$.

        Now, assume that $\sigma \in G_{v_R}$ satisfies $\pi_S$. Then there exist two words $w_1,w_2 \in \Sigma^*$ and a substitution $\rho \in S$ such that $\sigma
        (x) = w_1\rho(p)w_2$. If $\sigma(y)$ is not longer than any three non-overlapping factors of the form $0^*$ of $\sigma(x)$ combined, $\pi_3$ is satisfied. Otherwise, we can define $\tau$ by setting $\tau(y_1) := w_1, \tau(y_2) := w_2, \tau(\tilde{z}) :=0, \tau(\tilde{z}') :=0,\tau(\tilde{z}_j) :=0$ for all $j \in [k], \tau(\ddot{z}_j) := \#$ for all $j \in [g+5], \tau(\hat{x}_{i,j}) := \rho(\hat{x}_{i,j})$ for all $i \in \{1,2,3\},j \in \mathbb{N}$ where $\hat{x}_{i,j}$ appears in $p$ and $\tau(\hat{x}_{i,j}) := \epsilon$ where $\hat{x}_{i,j}$ does not appear in $p$. Finally, let $\tau(\hat{x}_{i,j}) = \tau(\hat{x}_{i,j'})$ for all other $i \in \{1,2,3\},j \in \mathbb{N}$ for one arbitrary, specific $j' \in \mathbb{N}$ such that $\tau(\hat{x}_{i,j'})$ is defined and $\tau(\hat{y}):= 0^m$, where
        $$m := |\sigma(y)| - \sum_{\hat{x} \in \text{var}(p)}{|\sigma(\hat{x})|}$$
        ($m > 0$ must hold, as $\sigma$ does not satisfy $\pi_3$). Then $\tau(p') = \rho(p)$, and
        \begin{align*}
            \tau(\gamma)\ &=\ \tau(y_1)\ \tau(p)\ \tau(y_2) \\
            &= w_1\ \rho(p)\ w_2 = \sigma(x), \\
            \tau(\delta)\ &=\ \tau(\hat{x}_{1,j}\hat{x}_{2,j}\hat{x}_{3,j}\hat{y}) \\ &= 0^{|\sigma(y)|}= \sigma(y), \\
            \tau(\eta)\ &=\ \tau(\tilde{z}\ \ddot{z}_{g+1}\ \ddot{z}_{g+2}\ \ddot{z}_{g+3}\ \ddot{z}_{g+4}\ \ddot{z}_{g+5}\ \tilde{z}') \\
            &= 0\#\#\#\#\#0 = u.
        \end{align*}
        Therefore, $\sigma$ satisfies $\pi$, which concludes this direction.

        For the other direction, assume that $\sigma \in G_{v_R}$ satisfies $\pi$. Then there is a morphism $\tau$ such that $\sigma(x) = \tau(\gamma),\sigma(y) = \tau(\delta)$ and $\tau(\eta) = u$. As $\eta = \tau(\tilde{z}\ \ddot{z}_{g+1}\ \ddot{z}_{g+2}\ \ddot{z}_{g+3}\ \ddot{z}_{g+4}\ \ddot{z}_{g+5}\ \tilde{z}')$ with 
        $(\tilde{z}_1, \tilde{z})$, $(\tilde{z}_1,\tilde{z}')$, $(\ddot{z}_1,\ddot{z}_{g+1})$, $(\ddot{z}_1,\ddot{z}_{g+2})$, $(\ddot{z}_1,\ddot{z}_{g+3})$, $(\ddot{z}_1,\ddot{z}_{g+4})$, $(\ddot{z}_1,\ddot{z}_{g+5}) \in R$ and $u = 0\#\#\#\#\#0$, $\tau(\tilde{z}) = \tau(\tilde{z}') = 0$ and $\tau(\ddot{z}_j) = \#$ for all $j \in \{g+1, \ldots, g+5\}$ must hold. By definition $\tau(y_1),\tau(y_2) \in \Sigma^*$. If we define $\rho(\hat{x}_{i,j}) := \tau(\hat{x}_{i,j})$ for all $i \in \{1,2,3\},j \in \mathbb{N}$ such that $\tau(\hat{x}_{i,j})$ is defined and $\rho(\hat{x}_{i,j}) := \rho(\hat{x}_{i,j'})$ for all other  $i \in \{1,2,3\},j \in \mathbb{N}$ for one arbitrary, specific $j' \in \mathbb{N}$ such that $\rho(\hat{x}_{i,j'})$ is defined, we see that $\sigma(x) \in L_1S(p,v_R)L_2$ holds. Thus, $\sigma$ satisfies $\pi_S$ as well.

        The other three cases for choices of $L_1$ and $L_2$ can be handled analogously by omitting $y_1$ or $y_2$ as needed. Note that this proof also works in the case $p = \epsilon$.
    \end{proof}

    Roughly speaking, if $\sigma$ does not satisfy $\pi_3$, then $\sigma(y)$ (which is in $0^*$, due to $\sigma \in G_{v_R}$) is long enough to provide building blocks for simple predicates using variables from $X$.

    Our next goal is a set of predicates that (if unsatisfied) forces $\sigma(x)$ into a basic shape common to all elements of $\mathtt{ValC}(A)$. We
    say that a word $w \in \{0,\#\}^*$ is of \textit{good structure} if $w\in (\#\#0^+\#0^+\#0^+)^+\#\#$. Otherwise, $w$ is of \textit{bad structure}. Recall that due to the definition of enc, all elements of $\mathtt{ValC}(A)$ are of good structure, thus being of bad structure, we define predicates $\pi_4$ to $\pi_{13}$ through simple predicates as follows:
    
    \begin{tabular}{ll}
        $\pi_4: \sigma(x) = \epsilon$, & $\pi_9: \sigma(x)$ ends on $0$, \\
        $\pi_5: \sigma(x) = \#$, & $\pi_{10}: \sigma(x)$ ends on $0\#$, \\
        $\pi_6: \sigma(x) = \#\#$, & $\pi_{11}:\sigma(x)$ contains a factor $\#\#0^*\#\#$, \\
        $\pi_7: \sigma(x)$ begins with $0$, & $\pi_{12}: \sigma(x)$ contains a factor $\#\#0^*\#0^*\#\#$, \\
        $\pi_8: \sigma(x)$ begins with $\#0$, & $\pi_{13}: \sigma(x)$ contains a factor $\#\#0^*\#0^*\#0^*\#0$.
    \end{tabular}
    
    Due to \Cref{lem:simplepred}, the predicates $\pi_1$ to $\pi_{13}$ do not strictly give rise to a characterization of substitutions with images that are of bad structure, as there are $\sigma \in G_{v_R}$ where $\sigma(x)$ is of good structure, but $\pi_3$ is satisfied due to $\sigma(y)$ being too short. But this problem can be avoided by choosing $\sigma(y)$ long enough to leave $\pi_3$ unsatisfied, and the following holds:
    \begin{lemma}\label{lem:allpi}
        A word $w \in \Sigma^*$ is of good structure if and only if there exists a substitution $\sigma \in H_{v_R}$ with $\sigma(x) = w$ such that $\sigma$ satisfies none of the predicates $\pi_1$ to $\pi_{13}$.
    \end{lemma}
    \begin{proof}
        We begin with the \textit{if} direction. Assume $\sigma \in H_{v_R}$ such that there is no $i \in \{1, \ldots, 13\}$ for which $\sigma$ satisfies $\pi_i$. Due to \Cref{lem:pi1}, $\sigma$ is of good form and, thus, $\sigma(y) \in 0^*$. As $\pi_3$ does not hold, $\sigma(y)$ is also longer than any three non-overlapping factors of $0^*$ of $\sigma(x)$. Thus, the structure of $\sigma(x)$ can be inferred by intersecting the complements of the simple predicates given in the definitions $\pi_4$ to $\pi_{13}$.

        As $\sigma$ does not satisfy $\pi_4$, $\sigma(x) \neq \epsilon$. Due to $\pi_7$ and $\pi_9$, the first and the last letter of $\sigma(x)$ is $\#$, and neither is $\#0$ a prefix, nor $0\#$ a suffix of $\sigma(x)$, as otherwise $\pi_8$ or $\pi_{10}$ would be satisfied. Therefore, $\sigma(x)$ has $\#\#$ as prefix and suffix, but, as $\pi_6$ is not satisfied, $\sigma(x) \neq \#\#$. As $\sigma$ is of good form, $\sigma(x)$ does not contain $\#\#\#$ as a factor, and
        $$\sigma(x) \in \#\#0^+\Sigma^*\#\#$$
        must hold. But as $\pi_{11}$ is not satisfied, it is possible to refine this observation to
        $$\sigma(x) \in \#\#0^+\#0^+\Sigma^*\#\#,$$
        which in turn leads to
        $$\sigma(x) \in \#\#0^+\#0^+\#0^+\Sigma^*\#\#$$
        by considering $\pi_{12}$ as well. Now, there are two possibilities. If
        $$\sigma(x) \in \#\#0^+\#0^+\#0^+\#\#,$$
        then $\sigma(x)$ is of good structure, but if
        $$\sigma(x) \in \#\#0^+\#0^+\#0^+\#\Sigma^*\#\#,$$
        then $\pi_{13}$ and $\sigma \in G_{v_R}$ lead to
        $$\sigma(x) \in \#\#0^+\#0^+\#0^+\#\#0^+\Sigma^*\#\#.$$
        In this case, we can continue referring subsequently to one of $\pi_{11}$ to $\pi_{13}$ together with $\sigma \in G_{v_R}$, and conclude
        $$\sigma(x) \in (\#\#0^+\#0^+\#0^+)^+\#\#.$$
        Therefore, if $\sigma$ satisfies none of $\pi_1$ to $\pi_{13}$, then $\sigma(x)$ has to be of good structure.

        Regarding the \textit{only if} direction, assume some $w \in \Sigma^*$ is of good structure. Define $\sigma$ by $\sigma(x) = w$ and $\sigma(y) = 0^{|w| +1}$. As $\sigma$ is of good form, \Cref{lem:pi1} demonstrates that $\sigma$ satisfies neither $\pi_1$ nor $\pi_2$; and as $\sigma(y)$ is longer than any word which results from concatenating any number of non-overlapping factors of the form $0^*$ of $w$, $\pi_3$ cannot be satisfied either. By looking at the cases used above to define $\pi_4$ to $\pi_{13}$, we see that none of these predicates can be satisfied.
    \end{proof}

    For every $w$ of good structure, there exist uniquely determined 
    
    $n,i_1,j_1,k_1, \ldots,i_n,j_n,k_n \in \mathbb{N}$ such that
    
    $w = \#\#0^{i_1}\#0^{j_1}\#0^{k_1}\#\#\ldots\#\#0^{i_n}\#0^{j_n}\#0^{k_n}\#\#$. Thus, if $\sigma\in H_{v_R}$ does not satisfy any of $\pi_1$ to $\pi_{13}$, $\sigma(x)$ can be understood as an encoding of a sequence $T_1, \ldots, T_n$ of triples $T_i \in (\mathbb{N})^3$, and for every sequence of that form, there is a $\sigma \in H_{v_R}$ such that $\sigma(x)$ encodes a sequence of triples of positive integers, and $\sigma$ does not satisfy any of $\pi_1$ to $\pi_{13}$.

    In the encoding of computations that is defined by enc, $\#\#$ is always a border between the encodings of configurations, whereas single $\#$ separates the elements of configurations. As we encode every state $q_i$ with $0^{i+1}$, the predicate $\pi_{14}$, which is to be satisfied whenever $\sigma(x)$ contains a factor $\#\#00^{s+1}$, handles all encoded triples $(i,j,k)$ with $i > s+1$. If $\sigma$ does not satisfy this simple predicate (in addition to the previous ones), there is a computation $C_1, \ldots, C_n$ of $A$ with enc$(C_1, \ldots, C_n) = \sigma(x)$.

    All that remains is to choose an appropriate set of predicates that describe all cases where $C_1$ is not the initial configuration, $C_n$ is not an accepting configuration, or there are configurations $C_i,C_{i+1}$ such that $C_i \vdash_A C_{i+1}$ does not hold (thus, the exact value of $\mu$ depends on the number of invalid transitions in $A$).

    To ensure $C_1 = (q_0,0,0)$, we define a predicate
    \begin{enumerate}
        \item $\sigma(x)$ has a prefix of the form $\#\#00$,
    \end{enumerate}
    that is satisfied if $C_1$ has a state $q_i$ with $i > 0$, and the two predicates
    \begin{enumerate}[resume]
        \item $\sigma(x)$ has a prefix of the form $\#\#0^*\#00$,
        \item $\sigma(x)$ has a prefix of the form $\#\#0^*\#0^*\#00$,
    \end{enumerate}
    to cover all cases where one of the counters is set to a value other than 0. Next, we handle the cases where the last state is not an accepting state. For every $i$ with $q_i \in Q \setminus F$, we define a predicate that is satisfied if
    \begin{enumerate}[resume]
        \item $\sigma(x)$ has a suffix of the form $\#\#00^i\#0^*\#\#$.
    \end{enumerate}
    Thus, if $\sigma \in H_{v_R}$ satisfies none of the predicates defined up to this point, $\sigma(x) = \text{enc}(C_1, \ldots, C_n)$ for some computation $C_1, \ldots, C_n$ with $C_1 =(q_0,0,0)$ and $C_n \in F \times \mathbb{N}_0 \times \mathbb{N}_0$, there is a $\sigma \in H_{v_R}$ with $\sigma(x)= \text{enc}(C_1, \ldots, C_n)$, and $\sigma$ satisfies none of these predicates.

    All that remains is to define a set of predicates that describe those $C_i,C_{i+1}$ for which $C_i \vdash_A C_{i+1}$ does not hold. To simplify this task, we define the following four predicates that are satisfied if one of the counters is changed by more than 1:
    \begin{enumerate}[resume]
        \item $\sigma(x)$ contains a factor of the form $\#0^m\#0^*\#\#0^*\#00\ 0^m$ for some $m \in \mathbb{N}_0$,
        \item $\sigma(x)$ contains a factor of the form $0^m\ 00\#0^*\#\#0^*\#0^m\#$ for some $m \in \mathbb{N}_0$,
        \item $\sigma(x)$ contains a factor of the form $\#0^m\#\#0^*\#0^*\#00\ 0^m$ for some $m \in \mathbb{N}_0$,
        \item $\sigma(x)$ contains a factor of the form $0^m\ 00\#\#0^*\#0^*\#0^m\#$ for some $m \in \mathbb{N}_0$,
    \end{enumerate}
    Here, the first two predicates cover incrementing (or decrementing) the first counter by 2 or more; the other two do the same for the second counter. Then, for all $i,j \in \{1, \ldots, s\}$, all $c_1,c_2 \in \{0,1\}$ and all $r_1,r_2 \in \{-1,0,+1\}$ for which $(q_j,r_1,r_2) \notin \delta(q_i,c_1,c_2)$, we define a predicate that is satisfied if $\sigma(x)$ contains such a transition. We demonstrate this only for the exemplary case $c_1 =0,c_2 = 1,r_1 = +1,r_2=0$ without naming $i$ or $j$ explicitly. The predicate covering non-existing transitions of this form is
    \begin{enumerate}[resume]
        \item $\sigma(x)$ contains a factor of the form $\#\#0^{i+1}\#0\#000^m\#\#0^{j+1}\#00\#000^m\#\#$ for some $m \in \mathbb{N}_0$.
    \end{enumerate}
    All other predicates for illegal transitions are defined analogously. Note that we can safely assume that none of the counters is changed by more than 1, as these errors are covered by the predicates we defined under points 5-8. The number of predicates required for these points and point 9 determine the exact value of $\mu$.

    Now, if there is a substitution $\sigma$ that does not satisfy any of $\pi_1$ to $\pi_{\mu}$, then $\sigma(x) = \text{enc}(C_1, \ldots, C_n)$ for a computation $C_1, \ldots, C_n$, where $C_1$ is the initial and $C_n$ a final configuration, and for all $i \in \{1, \ldots, n-1\}, C_i \vdash_A C_{i+1}$. Thus, if $\sigma(\alpha_A) \notin L_E(\beta_A,r_{\beta_A})$, then $\sigma(x) \in \mathtt{ValC}(A)$, which means that $A$ has an accepting computation.

    Conversely, if there is some accepting computation $C_1, \ldots, C_n$ of $A$, we can define $\sigma$ through $\sigma(x) := \text{enc}(C_1,\ldots, C_n)$, and choose $\sigma(y)$ to be an appropriately long sequence from $0^*$. Then $\sigma$ does not satisfy any of the predicates $\pi_1$ to $\pi_{\mu}$ defined above, thus $\sigma(\alpha_A) \notin L_E(\beta_A,r_{\beta_A})$, and $L_E(\alpha_A,r_{\alpha_A}) \nsubseteq L_E(\beta_A,r_{\beta_A})$.

    We conclude that $A$ has an accepting computation iff $L_E(\alpha_A,r_{\alpha_A})$ is not a subset of $L_E(\beta_A,r_{\beta_A})$. Therefore, any algorithm deciding the inclusion problem of \Cref{thm:uirev} can be used to decide whether a nondeterministic 2-counter automaton without input has an accepting computation. As this problem is known to be undecidable, the inclusion problem of \Cref{thm:uirev} is also undecidable.

    This proof can be extended to larger (finite) alphabets. Assume that $\Sigma =\{0,\#,a_1, \ldots,a_n\}$ for some $n \geq 1$. We extend $H$ to the set of all substitutions $\sigma:(\Sigma \cup \{x,y\})^* \rightarrow \Sigma^*$, but do not extend the definition of substitutions of good form to our new and larger alphabet. Thus, $\sigma \in H$ is of good form if $\sigma(x) \in \{0,\#\}^*,\sigma(y)\in 0^*$ and $\sigma(x)$ does not contain $\#^3$ as a factor. In addition to the predicates $\pi_1$ to $\pi_{\mu}$, for each new letter $a_i$, we define a predicate $\pi_{\mu +2i-1}$ which implies that $\sigma(x)$ contains an occurrence of $a_i$, and a predicate $\pi_{\mu+2i}$ which implies that $\sigma(y)$ contains an occurrence of $a_i$. To this end, we define
    $$ \alpha_A := v\ v\ \#^6\ v\ x\ v\ y\ v\ \#^6\ v\ u\ v$$
    where $x,y$ are distinct variables, $v = 0\#^30$ and $u=0\#^5a_1^5...a_n^50$ (instead of $u=0\#^50$), add the new predicates $\pi_{\mu +1}$ to $\pi_{\mu +2n}$ (which we still leave unspecified for a moment) to $\beta_A$ and use
    $$\eta_i:=z_i\ \hat{z}_{i_1}\ \hat{z}_{i_2}\ \hat{z}_{i_3}\ \hat{z}_{i_4}\ \hat{z}_{i_5}\ \ddot{z}_{i,1_1}\ \ddot{z}_{i,1_2}\ \ldots\ \ddot{z}_{i,n_1}\ \ddot{z}_{i,n_2}\ z_{i'}$$
    instead of $\eta_i := z_{i}\ \hat{z}_{i_1}\ \hat{z}_{i_2}\ \hat{z}_{i_3}\ \hat{z}_{i_4}\ \hat{z}_{i_5}\ z_{i'}$, where all $z_i,z_{i'},\hat{z}_{i_j},\ddot{z}_{i,k_j}$ are pairwise different variables with
    \begin{align*}
        &(z_i,z_{i'}),(\hat{z}_{i_1},\hat{z}_{i_2}),(\hat{z}_{i_1},\hat{z}_{i_3}),(\hat{z}_{i_1},\hat{z}_{i_4}),(\hat{z}_{i_1},\hat{z}_{i_5}),(\ddot{z}_{i,k_1},\ddot{z}_{i,k_2}) \in R.
    \end{align*}
    Referring to the new shape of $u$, we can make the following observation:
    \begin{lemma}\label{lem:finitealph}
        Let $n \geq 2$, $\{x_1,x_1',x_{2_1},x_{2_2},x_{2_3},x_{2_4},x_{2_5},x_{3_1},x_{3,2},\ldots,x_{n_1},x_{n_2}\} \subset X$ and $\{a_1, \ldots, a_n\} \subseteq \Sigma$. If 
        $$\alpha = x_1\ x_{2_1}\ x_{2_2}\ x_{2_3}\ x_{2_4}\ x_{2_5}\ x_{3_1}\ x_{3_2} \ldots x_{n_1}\ x_{n_2}\ x_1'$$
        with $(x_1,x_1'),(x_{2_1},x_{2_2}),(x_{2_1},x_{2_3}),(x_{2_1},x_{2_4}),(x_{2_1},x_{2_5}),(x_{i_1},x_{i_2}) \in R$ for all $i \in \{3, \ldots,n\}$ and there is a morphism $\sigma: X^* \rightarrow \Sigma^*$ with 

        \noindent
        $\sigma(\alpha)=a_1(a_2)^5(a_3)^2\ldots (a_n)^2a_1$, then $\sigma(x_1)=\sigma(x_1')=a_1$ and $\sigma(x_{i_j}) = a_i$ for each $i \in \{1, \ldots,n\}$.
    \end{lemma}
    \begin{proof}
        Assume 
        $$\sigma( x_1\ x_{2_1}\ x_{2_2}\ x_{2_3}\ x_{2_4}\ x_{2_5}\ x_{3_1}\ x_{3_2} \ldots x_{n_1}\ x_{n_2}\ x_1')= a_1(a_2)^5(a_3)^2\ldots (a_n)^2a_1$$. 
        If $\sigma(x_1) = \varepsilon$, then
        $$\sigma(x_{2_1}\ x_{2_2}\ x_{2_3}\ x_{2_4}\ x_{2_5}\ x_{3_1}\ x_{3_2} \ldots x_{n_1}\ x_{n_2})= a_1(a_2)^5(a_3)^2\ldots (a_n)^2a_1$$
        leads to an immediate contradiction. But $\sigma(x_1) \neq \varepsilon$ implies $\sigma(x_1)=a_1$. Therefore,
        $$\sigma(x_{2_1}\ x_{2_2}\ x_{2_3}\ x_{2_4}\ x_{2_5}\ x_{3_1}\ x_{3_2} \ldots x_{n_1}\ x_{n_2})= (a_2)^5(a_3)^2\ldots (a_n)^2$$
        must hold. Now, $\sigma(x_{2_j})=a_2$ must hold for every $j \in \{1,2,3,4,5\}$ since 
        
        $(x_{2_1},x_{2_2}),(x_{2_1},x_{2_3}),(x_{2_1},x_{2_4}),(x_{2_1},x_{2_5}),(x_{i_1},x_{i_2}) \in R$ for all $i \in \{3, \ldots,n\}$. Thus, we also get that $\sigma(x_{i_1}) = \sigma(x_{i_2}) = a_i$ for every $i \in \{3, \ldots, n\}$.
    \end{proof}
    \Cref{lem:finitealph} allows $\pi_{\mu +1}$ to $\pi_{\mu +2n}$ to be analogously constructed to $\pi_2$. To this end, we define

    \begin{table}[h]
        \centering
        \begin{tabular}{ll}
          $\gamma_{\mu +2i-1} := y_{\mu+2i-1,1}\ \ddot{z}_{\mu +2i -1,i}'\ y_{\mu +2i -1,2},$   & $\gamma_{\mu +2i} := y_{\mu +2i},$ \\
           $\delta_{\mu +2i-1} := \hat{y}_{\mu+2i-1},$  & $\delta_{\mu+2i} := \hat{y}_{\mu+2i,1}\ \ddot{z}_{\mu +2i,i}'\ \hat{y}_{\mu +2i,2}$
        \end{tabular}
        \label{tab:newalph}
    \end{table}

    for each $i \in [n]$. 
    
    Again, all $y_{j,k},\hat{y}_{j,k},\ddot{z}_{j,k}'$ are pairwise different variables with $(\ddot{z}_{i,1_1},\ddot{z}_{\mu +2i -1,i}')$, $(\ddot{z}_{i,1_1},\ddot{z}_{\mu +2i,i}') \in R$. Now \Cref{lem:pi1} applies (mutatis mutandis) as for binary alphabets, and since all substitutions of good form behave for $\Sigma$ as for the binary alphabet, we can use the very same predicates and the same reasoning as before to prove undecidability of the inclusion problem of the case $r_{rev}$ for larger alphabets.

    This concludes the proof of the case $r_{rev}$.
\end{proof}

%% file: inclusion_long_erasing_abel.tex
\begin{proof}
    This proof is analogous to the case $r_{ab}$. Therefore, we only describe how to get to the analogy and refer for a detailed proof to the case $r_{ab}$. Let $A$ be a nondeterministic 2-counter automaton. Consider 
    $$\alpha_A := v\ v\ \#^6\ v\ x\ v\ y\ v\ \#^6\ v\ u\ v$$
    where $x,y$ are distinct variables, $v = 0\#^30$ and $u = 0\#^50$ and 
    \[ \beta_A:=x_1x_{1'}\ \ldots x_{\mu}x_{\mu'}\#^6\hat{\beta}_1\ldots\hat{\beta}_{\mu}\#^6\ddot{\beta_1}\ldots\ddot{\beta_{\mu}} \]
    with, for all $i \in \{1, \ldots, \mu\}$, and $\ddot{\beta_i} := x_{i_4}\ \eta_i\ x_{i_5}$, where $x_1$, $x_{1'}$, $x_{1_4}$, $x_{1_5}$, $\ldots$ are distinct variables with $(x_i, x_{i'}), (x_i, x_{i_1}), (x_i, x_{i_2}),(x_i,x_{i_3}),(x_i,x_{i_4}),(x_i,x_{i_5}) \in R$ and all $\eta_i \in X^*$ are terminal-free patterns for some $\mu$ and some $\hat{\beta}_i$ as in the case $r_{ab}$. We have
    $\eta_i := z_{i}\ \hat{z}_{i_1}\ \hat{z}_{i_2}\ \hat{z}_{i_3}\ \hat{z}_{i_4}\ \hat{z}_{i_5}\ z_{i'}$ with distinct variables $z_i$, $z_{i'}$, $\hat{z}_{i_1}$, $\hat{z}_{i_2}$, $\hat{z}_{i_3}$, $\hat{z}_{i_4}$, $\hat{z}_{i_5}$ and $(z_i,z_{i'}),(\hat{z}_{i_1},\hat{z}_{i_2}),(\hat{z}_{i_1},\hat{z}_{i_3}),(\hat{z}_{i_1},\hat{z}_{i_4}),(\hat{z}_{i_1},\hat{z}_{i_5}) \in R$.
    
    The main concept of the proof for $r_{ab}$ is to show that iff $L_E(\alpha_A,r_{\alpha_A}) \subseteq L_E(\beta_a,r_ {\beta_A})$ then $A$ has no accepting computation. All possible non accepting computations are listed in the $\hat{\beta}_1 \ldots \hat{\beta}_{\mu}$ part where each $\hat{\beta}_i, i \in \{1, \ldots, \mu \}$ encodes one non accepting computation together with selecting variables $x_{i_j}$ for $j \in {1,2,3}$. The part of $\beta_A$ and $\alpha_A$ before the first $\#^6$ ensures that exactly one pair $x_i,x_i', i \in \{1, \ldots, \mu\}$ is selected, resulting in that exactly one non accepting computation is selected and that exactly one $\ddot{\beta}_i, i \in \{1, \ldots, \mu\}$ is selected. With this, we get that $\eta_i, i \in \{1, \ldots, \mu\}$ matches $u$ and thus, that $z_i = 0$ and $\hat{z}_{i_1} = \#$. We use these two variables to encode the $0$ and the $\#$ in the non accepting computation encodings. If we can show that this selecting process also works in the setting $R = r_{ab}$ and that we get two variables with which we can encode the letters $0$ and $\#$, then this proof also works in this setting. For two variables $x,y$, we have if $x = 0$ (or $x = \#$) and $(x,y) \in r_{ab}$, then $y = 0$ (or $y = \#$ respectively). Thus, we can model with the relation $r_{ab}$ identical variables if we know a one letter substitution of one variable.

    If we want $L_E(\alpha_A,r_{\alpha_A}) \subseteq L_E(\beta_A,r_{\beta_A})$, we need to match the factor $v\ v$ to $x_1x_{1'} \ldots x_{\mu}x_{\mu'}$ and $v\ u\ v$ to $\ddot{\beta}_1 \ldots \ddot{\beta}_{\mu}$ due to our delimiter $\#^6$. We start with the match of $v\ v$ to $x_1x_{1'} \ldots x_{\mu}x_{\mu'}$. We have $v\ v = 0\ \#^3 0\ 0 \#^3 0$. Since the variables are pairwise in relation, we can always assume that we start substituting with the left variable. If a variable gets substituted to something different from $\epsilon$, then its corresponding variable also gets substituted to something different from $\epsilon$. We get the following reasonable possibilities for the beginning of the match:
    \begin{table}[h]
        \centering
        \resizebox{\textwidth}{!}{
            \begin{tabular}{l|l|l}
                $x_i$ &  $x_{i'}$ & works? \\
                \hline
                 $0$ & $0$ & $00 \neq 0\#$ \\
                 $0\#$ & $0\#$ or $\# 0$ & $0\#0\#,\ 0\#^20 \neq 0\#^3$ \\
                 $0\#^2$ & $0\#^2$ or $\#0\#$ or $\#\#0$ & $0\#^20\#^2,\ 0\#^30\#,\ 0\#^4 0 \neq 0\#^300$ \\
                 $0\#^3$ & $0\#^3$ or $\#0\#^2$ or $\#^20\#$ or $\#^30$ &  $0\#^30\#^3,\ 0\#^3\#0\#^2,\ 0\#^50\#,\ 0\#^630 \neq 0\#^300\#^2$ \\
                 $0\#^30$ & $0\#^3 0$ & $0 \#^300\#^30 = v\ v$
            \end{tabular}
        }
        \caption{Possibilities for the first match with $v\ v$}
        \label{tab:posv-appdx}
    \end{table}
    
    Thus, the only possibility is that for exactly one $i \in \{1, \ldots, \mu\}$ the variables $x_i$ and $x_{i'}$ get substituted with $v$ and all other variables $x_j,x_{j'}$ for $j \neq i, j \in \{1, \ldots, \mu\}$ get substituted with $\epsilon$. This together with the matching from $v\ u\ v$ to $\ddot{\beta}_1 \ldots \ddot{\beta}_{\mu}$ implies that $\eta_i$ needs to match $u$. For this matching between $\eta_i = z_{i}\ \hat{z}_{i_1}\ \hat{z}_{i_2}\ \hat{z}_{i_3}\ \hat{z}_{i_4}\ \hat{z}_{i_5}\ z_{i'}$ and $u = 0\#^50$, we get the following possibilities. Remember that $(z_i,z_{i'}),(\hat{z}_{i_1},\hat{z}_{i_2}),(\hat{z}_{i_1},\hat{z}_{i_3}),(\hat{z}_{i_1},\hat{z}_{i_4}),(\hat{z}_{i_1},\hat{z}_{i_5}) \in R$.
    
    \begin{table}[h]
        \centering
        \begin{tabular}{l|l|l}
           $z_i$  & $\hat{z}_{i_1}$ & works? \\ \hline
            $\epsilon$ & $0$ or $\#$ & $0^5,\ \#^5 \neq 0\#^50$ \\
            $0$ & $\#$ & $0\#^50 = 0\#^5 0 = u$ \\
            $0\#$ & $\#$ & $0\#\#^5\#0 \neq 0 \#^50$ \\
            $0\#\#$ & $\#$ & $0\#^2\#^5\#^20 \neq 0\#^50$ \\
            $0\#\#\#$ & $\epsilon$ & $0\#^60 \neq 0\#^50$
        \end{tabular}
        \caption{Possibilities for the match with $u$}
        \label{tab:posu-appdx}
    \end{table}

    Thus, the only possibility is that $z_i$ and $z_{i'}$ get substituted with $0$ and $\hat{z}_{i_k}, k \in \{1,2,3,4,5\}$ get substituted with $\#$. With this, we can encode the $0$ and the $\#$ and thus, the proof of the case $r_{rev}$ works in this setting analogously. This concludes the proof of Theorem~\ref{thm:uirev}.
\end{proof}

%% file: 11_appendix_membeship.tex
\subsection{Proof of Lemma~\ref{lemma:membership-trivial-hardness}}\label{proof:lemma:membership-trivial-hardness}
\begin{proof}
    This result follows from known reductions, e.g., \cite{Manea2019}, for all considered cases in this Lemma. For reference, we recall the construction formally in the relational pattern framework and state why it translates over.

    First, we recall \emph{positive 1-in-3-SAT}. A formula $\varphi$ in positive 3-CNF (conjunctive normal form) is a set of clauses $\varphi = \{C_1,...,C_n\}$, for some $n\in\mathbb{N}$, where each clause $C_i$ is a set of 3 non-negated literals $C_i = \{X_{i,1},X_{i,2},X_{i,3}\}$, where each $X_{i,j}\in X_p$ ($X_p$ being the set of all propositional variables). A positive 3-CNF formula $\varphi$ is 1-in-3 satisfiable if there exists an assignment $\beta : X_p \rightarrow \{0,1\}$ that sets exactly one literal per clause to $1$.

    Assume an arbitrary positive 3-CNF formula $\varphi = \{C_1,...,C_n\}$. Construct the word $w = (1\#)^{n-1}1$ over the binary alphabet $\{1,\#\}$. Next, construct the pattern 
    \[\alpha = x_{1,1}x_{1,2}x_{1,3} \# ... \# x_{n,1}x_{n,2}x_{n,3},\] 
    for independent variables $x_{i,j}\in X$ where $i,j\in\mathbb{N}$. Construct the relational constraints $r_\alpha$ by adding, for each $X_{i,j}$ and $X_{i',j'}$ occurring in $\varphi$ where $X_{i,j} = X_{i',j'}$ are the same variable, the pairs $(x_{i,j},x_{i',j'}),(x_{i',j'},x_{i,j})\in r_\alpha$. Notice that $|w|_{\#} = |\alpha|_{\#} = n-1$. Hence, each $r_\alpha$-valid substitution $h$ must set all variables either to $1$ or to $\varepsilon$ to obtain $h(\alpha) = w$.

    Now, we see that if $\varphi$ is 1-in-3-satisfiable, then there exists an assignment of variables $\beta$ that sets exactly one variable per clause to $1$. Based on $\beta$, construct a substitution $h$ where $h(x_{i,j}) = 1$ iff $\beta{X_{i,j}} = 1$. Otherwise set $h(x_{i,j}) = \varepsilon$. As each clause only has one variable set to true by $\beta$, we obtain $h(\alpha) = w$. For the other direction, notice that we can analogously construct an assignment $\beta$ out of a $r_\alpha$-valid substitution $h$ with $h(\alpha) = w$. This concludes the erasing case. For the non-erasing case, notice that we can construct $(\alpha,r_\alpha)$ in exactly the same manner and construct a word $w = (1\#)^4\#1^4$. As the non-erasing case forces each variable to have at least length 1 and all variables must be substituted to multiples of the letter $1$, we can use length $2$ substitutions to represent the positive literals. The argument follows analogously.

    \begin{enumerate}
        \item If $R = \{r_{|w|}\}$, then variables have to be substituted with equal length. As any substitution resulting in $w$ substituted each variable only over unary words over the letter $1$, having equal length also means being the equal word. This concludes this relation for the E and NE cases.
        \item If $R = \{r_{ssq}\}$, then relating two variables $x$ and $y$ in both directions, i.e., putting $(r_{ssq},x,y)\in r_{\alpha'}$ as well as $(r_{ssq},y,x)\in r_{\alpha'}$ results in some substitution $h\in H$ to be valid if and only if $h(x) = h(y)$. This concludes this relation for the E and NE cases.
        \item If $R = \{r_{ab}\}$ the same argument as for $r_{|w|}$ holds for the E and NE cases.
        \item If $R = \{r_{perm}\}$ the same argument as for $r_{|w|}$ holds for the E and NE cases.
        \item If $R = \{r_{rev}\}$ the same argument as for $r_{|w|}$ holds for the E and NE cases.
        \item If $R = \{r_{*}\}$ the the same argument as for $r_{ssq}$ holds for the E and NE cases.
        \item If $R = \{r_{com^+}\}$ then, in the erasing case, substituting one variable to the letter $1$, as $\varepsilon$ is not allowed, forces related variables to also be substituted to the letter $1$ for a $r_\alpha$-valid substitution $h$ to result in $h(\alpha) = w$.
        \item If $R = \{r_{\varepsilon_=}\}$ the same argument as for $r_{com^+}$ holds for the E case.
    \end{enumerate}

    This concludes this proof.
\end{proof}

\subsection{Proof of Theorem~\ref{theorem:membership-np-complete-cases}}\label{proof:theorem:membership-np-complete-cases}
\begin{proof}
    Indeed, NP-containment of all considered cases follows from \Cref{prop:membership-innp-relational}. For most cases, NP-hardness has been confirmed in \ref{lemma:membership-trivial-hardness}. What is left to show is that, first, in the case of non-erasing pattern languages, NP-hardness also holds for the relations $r_{com^+}$ and $r_{com^*}$. Indeed, in the NE case, these two relations result in the same languages, hence they can be considered together. Second, in the erasing case, the relation $r_{com^*}$ is still open. The rest of this proof provides two reductions, showing NP-hardness for the remaining two cases.

    \textbf{(Non-Erasing ($r_{com^+}$/$r_{com^*}$):} W.l.o.g. we consider $R = \{r_{com^+}\}$. We reduce the general 3-SAT problem to the respective membership problem. Let $\varphi = \{C_1,...,C_n\}$ be a 3-CNF formula, for some $n>0$, where $C_i = \{X_{i,1},X_{i,2},X_{i,3}\}$ for clauses $C_i$ and propositional variables $X_{i,j}\in\{X_1,...,X_m,\overline{X_1},...,\overline{X_m}\}$. Notice that in this variant of the problem, negated variables may occur as literals. First, we construct the word
    \[w = \#\#\ s_1\ \#\#\ ...\ \#\#\ s_m\ \#\#\ t_1\ \#\#\ ...\ \#\#\ t_n\ \#\#\]
    where $s_i = 1\#1$ and $t_j = 1^{10}\#1^{10}\#1^{10}$, for each $1 < i \leq n$ and $1 \leq j \leq m$. Next, we construct the regular pattern
    \[\alpha = \#\#\ u_{0}v_{0}\ \#\#\ ...\ \#\#\ u_{m}v_{m}\ \#\#\ \alpha_{1}\ \#\#\ ...\ \#\#\ \alpha_{n}\ \#\#\]
    where $\alpha_{i} = z_{i,1}\ x_{i,1}\ z_{i,2}\ x_{i,2}\ z_{i,3}\ x_{i,3}\ z_{i,4}$, for $1 \leq i \leq n$ and variables $z_{i,1}$ to $z_{i,4}$ as well as $x_{i,1}$ to $x_{i,3}$. 
    
    Notice that the pattern $\alpha$ is indeed regular and that each variable occurs only once. Finally, we construct the regular constraints $r_\alpha$. For all non-negated variables $X_k\in\{X_1,...,X_m\}$ and literals equal to that variable $X_{i,j} = X_k$, we add the triple $(r_{com^*},u_k,x_{i,j})\in r_\alpha$. Analogously, for all negated variables $\overline{X_k}\in\{\overline{X_1},...,\overline{X_m}\}$, we add the triple $(r_{com^*},v_k,x_{i,j})\in r_\alpha$. In conclusion, there exist two parts. First, each factor $s_k$, $k\in[m]$, in $w$ has to be matched against $u_kv_k$ in $\alpha$, and each factor $t_i$, $\in[n]$ has to be matched against $\alpha_i$. We proceed by proving the arguments.
    
    ($\Rightarrow$:) Assume there exists a satisfying assignment of variables $\beta$ for $\varphi$. Let $h$ be a substitution with the following properties: For every variable $X_k$, $k\in[m]$, where $\beta(X_k) = 1$ (thus $\beta(\overline{X_k}) = 0$), we set $h(u_k) = 1$, $h(v_k) = \#1$ and, for every literal $X_{i,j}$ with $X_{i,j} = X_k$, we set $h(x_{i,j}) = 1$, and, for every literal with $X_{i,j} = \overline(X_k)$, we set $h(x_{i,j}) = \#1$. Analogously, if $\beta(X_{k}) = 0$ and $\beta(\overline{X_{k}}) = 1$, we set $h(u_k) = 1\#$ and $h(v_k) = 1$ and the $x_{i,j}$ variables correspondingly. As all variables are assigned some value in $\beta$, all variables of the form $x_{i,j}$ in $\alpha$ are substituted by some value (either $1$ or $1\#$ or $\#1$). For each $\alpha_i$, $i\in[n]$, we can now set the surrounding independent variables $z_{i,1}$ to $z_{i,4}$ to non-empty words such that $h(\alpha_i) = 1^10\#1^10\#1^10$. As there is at least one variable $x_{i,j}$ in each $\alpha_i$ that is substituted by a word not containing $\#$ (as $\beta$ is a satisfying assignment), we will not obtain too many $\#$ by the substitutions for $x_{i,j}$s alone. Additionally, as there are enough $1$s in between each $\#$ in $w$, we can find a non-empty substitution for each $r_{i,j}$ variable. As all variables related by constraints in $r_{\alpha}$ are substituted by equal words, conforming the relation $r_{com^+}$, the chosen substitution is $r_\alpha$-valid. We notice that each $h(u_{k}v_{k}) = 1\#1$, each $h(\alpha_i) = 1^10\#1^10\#1^10$, and therefore $h(\alpha) = w$.

    ($\Leftarrow$:) Now assume there exists a $r_\alpha$-valid substitution $h$ such that $h(\alpha) = w$. We notice that no variable can be substituted to contain the factor $\#\#$, as otherwise $|h(\alpha)|_{\#\#} > |w|_{\#\#}$. Hence, we must have $h(u_kvk) = s_k = 1\#1$ as well as $h(\alpha_i) = t_i = 1^10\#1^10\#1^10$. As this is the non-erasing case, for each $k\in[m]$, either $h(u_k)$ or $h(v_k)$ contains the symbol $\#$ while the other is substituted only by $1$. As each $h(\alpha_i)$ only contains two $\#$ symbols, there is at least one variable of the form $x_{i,j}$ which is substituted by a unary word in $\{1\}^+$. As each such variable is related to some $u_k$ or $v_k$ respectively, this can only be the case if both are substituted by unary words in $\{1\}^+$. Construct an assignment $\beta$ that sets, for each $k\in[m]$, $\beta(X_k) = 1$ and $\beta(\overline{X_k}) = 0$ if $h(u_k) = 1$ (thus $h(v_k) = \#1$), and, otherwise, set $\beta(X_k) = 0$ and $\beta(\overline{X_k}) = 1$ if $h(v_k) = 1$ (thus $h(u_k) = 1\#$). By the argument from above, we now obtain an assignment that sets at least one variable per clause to $1$, satisfying $\varphi$.

    \textbf{(Erasing ($r_{com^*}$):} In this case, we can construct a simplified, but analogoue version of the reduction of the NE case. Assume $\varphi$ to be a formula of the form from above. Construct the word 
    \[w = \#\#\ s_1\ \#\#\ ...\ \#\#\ s_m\ \#\#\ t_1\ \#\#\ ...\ \#\#\ t_n\ \#\#\]
    where, this time, $s_i = 1\#1$ and $t_j = 1$, for $1 \leq i \leq m$ and $1 \leq j \leq n$, and let
    \[\alpha = \#\#\ u_{0}v_{0}\ \#\#\ ...\ \#\#\ u_{m}v_{m}\ \#\#\ \alpha_{1}\ \#\#\ ...\ \#\#\ \alpha_{n}\ \#\#\]
    where $\alpha_{i} = x_{i,1}x_{i,2}x_{i,3}$ for variables $x_{i,1}$ to $x_{i,3}$. We set $r_\alpha$ analogously to the NE case. The first direction of the proof, assuming a satisfying assignment for $\varphi$, we can create a $r_\alpha$-substitution similarly to how it is done in the NE case. Set $u_k$ and $v_k$ analogously to the NE case. For each clause, choose one of the literal variables $x_{i,j}$ that is related to either some $u_k$ or $v_k$ substituted by $1$ and substitute it by $1$. The other direction of the proof follows by the same reasoning from the NE case that the $\#\#$ factors have to align in $w$ and $h(\alpha)$, therefore forcing $h(u_kv_k) = 1\#1$ and $h(\alpha_i) = 1$. Hence, for each $i\in[n]$, some variable $x_{i,j}$ is substituted by $1$. This is only possible if either the related $u_k$ or $v_k$ is substituted by either $\varepsilon$ or $1$ due to the definition of $r_{com^*}$. From that, we can construct a satisfying assignment $\beta$ for $\varphi$, analogously to before, concluding this proof idea.
    
    Hence, the membership problem for all cases considered in \Cref{theorem:membership-np-complete-cases} is NP-complete, concluding this proof.
\end{proof}